\documentclass[a4paper,UKenglish,cleveref, autoref,numberwithinsect]{lipics-v2021}

\usepackage{xspace}
\usepackage{amsmath}

\usepackage{amsthm}

\usepackage{amssymb}

\usepackage{multirow}

\usepackage{graphicx}

\usepackage{todonotes}
\usepackage{xspace}

\usepackage[ruled, vlined,boxed,commentsnumbered,linesnumbered]{algorithm2e}
\usepackage{color}
\definecolor{black}{rgb}{0, 0, 0}

\definecolor{white}{rgb}{1, 1, 1}

\definecolor{grey}{rgb}{.6, .6, .6}

\definecolor{red}{rgb}{1, 0 ,0}

\definecolor{green}{rgb}{0, 1, 0}

\definecolor{blue}{rgb}{0, 0 ,1}

\definecolor{darkred}{rgb}{0.7, 0 ,0}

\definecolor{darkgreen}{rgb}{0, 0.7, 0}

\definecolor{darkblue}{rgb}{0, 0 , 0.55}

\definecolor{magenta}{rgb}{1, 0, 1}

\definecolor{cyan}{rgb}{0, 1, 1}

\definecolor{yellow}{rgb}{1, 0.9, 0}

\definecolor{purple}{rgb}{0.5, 0, 0.5}

\definecolor{orange}{rgb}{1, 0.5, 0}

\definecolor{turquoise}{rgb}{0, 0.7, 0.7}

 \definecolor{mid-green}{rgb}{0.15,0.65,0.15}
\definecolor{dark-green}{rgb}{0.15,0.25,0.15}
\definecolor{dark-red}{rgb}{0.7,0.15,0.15}
\definecolor{dark-blue}{rgb}{0.15,0.15,0.9}
\definecolor{medium-blue}{rgb}{0,0,0.5}
\definecolor{gray}{rgb}{0.5,0.5,0.5}
\definecolor{color-Ig}{rgb}{0.15,0.7,0.15}
\definecolor{darkmagenta}{rgb}{0.30, 0.0, 0.30}

\hypersetup{
    colorlinks, linkcolor={dark-blue},
    citecolor={mid-green}, urlcolor={dark-blue}}

\usepackage{alphabeta,cite}

\RequirePackage{stmaryrd}
\usepackage{textcomp}
\DeclareUnicodeCharacter{2286}{\subseteq}
\DeclareUnicodeCharacter{2192}{\ifmmode\to\else\textrightarrow\fi}
\DeclareUnicodeCharacter{2203}{\ensuremath\exists}
\DeclareUnicodeCharacter{183}{\cdot}
\DeclareUnicodeCharacter{2200}{\forall}
\DeclareUnicodeCharacter{2264}{\leq}
\DeclareUnicodeCharacter{2265}{\geq}
\DeclareUnicodeCharacter{8614}{\mathbin{\mapsto}}
\DeclareUnicodeCharacter{8656}{\Leftarrow}
\DeclareUnicodeCharacter{8657}{\Uparrow}
\DeclareUnicodeCharacter{8658}{\Rightarrow}
\DeclareUnicodeCharacter{8659}{\Downarrow}
\DeclareUnicodeCharacter{8669}{\rightsquigarrow}
\newcommand{\eqdef}{\stackrel{{\scriptsize\rm def}}{=}}
\DeclareUnicodeCharacter{8797}{\eqdef}
\DeclareUnicodeCharacter{8870}{\vdash}
\DeclareUnicodeCharacter{8873}{\Vdash}
\DeclareUnicodeCharacter{22A7}{\models}
\DeclareUnicodeCharacter{9121}{\lceil}
\DeclareUnicodeCharacter{9123}{\lfloor}
\DeclareUnicodeCharacter{9124}{\rceil}
\DeclareUnicodeCharacter{2208}{\in}
\DeclareUnicodeCharacter{9126}{\rfloor}
\DeclareUnicodeCharacter{9655}{\triangleright}
\DeclareUnicodeCharacter{9665}{\triangleleft}
\DeclareUnicodeCharacter{9671}{\diamond}
\DeclareUnicodeCharacter{9675}{\circ}
\DeclareUnicodeCharacter{10178}{\bot}
\DeclareUnicodeCharacter{10214}{} 
\DeclareUnicodeCharacter{10215}{} 
\DeclareUnicodeCharacter{10229}{\longleftarrow}
\DeclareUnicodeCharacter{10230}{\longrightarrow}
\DeclareUnicodeCharacter{10231}{\longleftrightarrow}
\DeclareUnicodeCharacter{10232}{\Longleftarrow}
\DeclareUnicodeCharacter{10233}{\Longrightarrow}
\DeclareUnicodeCharacter{10234}{\Longleftrightarrow}
\DeclareUnicodeCharacter{10236}{\longmapsto}
\DeclareUnicodeCharacter{10238}{\Longmapsto} 
\DeclareUnicodeCharacter{10503}{\Mapsto}    
\DeclareUnicodeCharacter{10971}{\mathrel{\not\hspace{-0.2em}\cap}}
\DeclareUnicodeCharacter{65294}{\ldotp}
\DeclareUnicodeCharacter{65372}{\mid}

\addto\extrasenglish{} 
\addto\extrasenglish{} 
\addto\extrasenglish{}

\newcommand{\NP}{{\sf NP}}

\newcommand{\opt}{{\sf opt}\xspace}

\newcommand{\td}{{\sf td}\xspace}
\newcommand{\XP}{{\sf XP}\xspace}
\newcommand{\yes}{{\sf yes}\xspace}
\renewcommand{\H}{\mathcal{H}}
\newcommand{\sol}{\mathsf{sol}}
\newcommand{\optsol}{\mathsf{optsol}}
\newcommand{\optsolwith}{\mathsf{optsolwith}}

\newcommand{\edcalH}{{\sf ed}_{\H}}

\newcommand{\vedH}{{\sf ved}^+_H\xspace}

\newcommand{\vedt}{{\sf ved}^+_t\xspace}
\newcommand{\bedt}{{\sf bed}^+_t\xspace}

\newcommand{\veds}{{\sf ved}^+_{\F_{\bar H}}\xspace}
\newcommand{\vedis}{{\sf ved}^+_{\F^{\mathsf{ind}}_{\bar H}}\xspace}
\newcommand{\beds}{{\sf bed}^+_{\F_{\bar H}}\xspace}
\newcommand{\bedis}{{\sf bed}^+_{\F^{\mathsf{ind}}_{\bar H}}\xspace}

\newcommand{\vedclique}{{\sf ved}^+_{\F_{\bar{K_t}}}\xspace}
\newcommand{\vediclique}{{\sf ved}^+_{\F^{\mathsf{ind}}_{\bar K_t}}\xspace}
\newcommand{\bedclique}{{\sf bed}^+_{\F_{\bar{K_t}}}\xspace}
\newcommand{\bediclique}{{\sf bed}^+_{\F^{\mathsf{ind}}_{\bar K_t}}\xspace}

\renewcommand{\O}{\mathcal{O}}
\newcommand{\Oh}{\mathcal{O}}

\newcommand{\Bcal}{\mathcal{B}}
\newcommand{\Ccal}{\mathcal{C}}

\newcommand{\Fcal}{\mathcal{F}}

\newcommand{\F}{\Fcal}
\newcommand{\C}{\Ccal}
\newcommand{\B}{\Bcal}

\newcommand{\mmbs}{{\sf mmbs}\xspace}
\newcommand{\pr}{{\sf pr}\xspace}
\newcommand{\add}{{\sf add}\xspace}
\newcommand{\conf}{{\sf conf}\xspace}
\newcommand{\X}{\mathcal{X}}
\newcommand{\T}{\mathcal{T}\xspace}
\newcommand{\M}{\mathcal{M}\xspace}
\renewcommand{\P}{\mathcal{P}\xspace}
\newcommand{\mainMark}{{\tt mainMark}}

\newcommand{\KtH}{\textsc{$K_t$-SH}\xspace}
\newcommand{\EKtH}{\textsc{\sc E$K_t$-SH}\xspace}
\newcommand{\KtHM}{\textsc{$K_t$-SHM}\xspace}
\newcommand{\KtHMdecbis}[1]{\textsc{$K_t$-SHM$_{\sf p}^{#1}$}\xspace}

\newcommand{\HH}{\textsc{$H$-SH}\xspace}
\newcommand{\HiH}{\textsc{$H$-ISH}\xspace}

\newtheorem{property}{Property}
\newtheorem*{theorem*}{Theorem} 

\usepackage{thmtools}
\usepackage{thm-restate}

\title{Kernelization Dichotomies for Hitting Subgraphs under Structural Parameterizations}

\author{Marin Bougeret}{LIRMM, Universit\'e de Montpellier, CNRS, Montpellier, France}{marin.bougeret@lirmm.fr}{https://orcid.org/0000-0002-9910-4656}{}

\author{Bart M.\ P.\ Jansen}{Eindhoven University of Technology, The Netherlands}{b.m.p.jansen@tue.nl}{https://orcid.org/0000-0001-8204-1268}{Supported by the Dutch Research Council (NWO) through Gravitation-grant NETWORKS-024.002.003.}

\author{Ignasi Sau}{LIRMM, Universit\'e de Montpellier, CNRS, Montpellier, France}{ignasi.sau@lirmm.fr}{https://orcid.org/0000-0002-8981-9287}{Supported  by the French project ELIT (ANR-20-CE48-0008-01).}

\authorrunning{Marin Bougeret, Bart M. P. Jansen, and Ignasi Sau} 

\Copyright{Jane Open Access and Joan R. Public} 



\category{} 

\nolinenumbers 

\hideLIPIcs

\EventEditors{Karl Bringmann, Martin Grohe, Gabriele Puppis, and Ola Svensson}
\EventNoEds{4}
\EventLongTitle{51st International Colloquium on Automata, Languages, and Programming (ICALP 2024)}
\EventShortTitle{ICALP 2024}
\EventAcronym{ICALP}
\EventYear{2024}
\EventDate{July 8--12, 2024}
\EventLocation{Tallinn, Estonia}
\EventLogo{}
\SeriesVolume{297}
\ArticleNo{5}

\begin{document}

\maketitle

\begin{abstract}
\noindent For a fixed graph $H$, the $H$-{\sc Subgraph Hitting} problem consists in deleting the minimum number of vertices from an input graph to obtain a graph without any occurrence of $H$ as a subgraph. This problem can be seen as a generalization of \textsc{Vertex Cover}, which corresponds to the case $H = K_2$. We initiate a study of $H$-{\sc Subgraph Hitting} from the point of view of characterizing structural parameterizations that allow for polynomial kernels, within the recently active framework of taking as the parameter the number of vertex deletions to obtain a graph in a ``simple'' class~$\C$. Our main contribution is to identify graph parameters that, when $H$-{\sc Subgraph Hitting} is parameterized by the vertex-deletion distance to a class $\C$ where any of these parameters is bounded, and assuming standard complexity assumptions and that $H$ is biconnected, allow us to prove the following sharp dichotomy: the problem admits a polynomial kernel if and only if $H$ is a clique. These new graph parameters are inspired by the notion of \textit{$\C$-elimination distance} introduced by Bulian and Dawar [Algorithmica 2016], and generalize it in two directions.
Our results also apply to the version of the problem where one wants to hit $H$ as an induced subgraph, and imply in particular, that the problems of hitting minors and hitting (induced) subgraphs have a substantially different behavior with respect to the existence of polynomial kernels under structural parameterizations.
\end{abstract}




\section{Introduction}
\label{sec:intro}

The theory of parameterized complexity deals with \emph{parameterized problems}, which are decision problems in which a positive integer $k$, called the \textit{parameter}, is associated with every instance $x$. One of the pivotal notions in the domain is that of \textit{kernelization}~\cite{Bodlaender16,CyganFKLMPPS15,FellowsJKRW18,FominLSZ19,LokshtanovMS12}, which is a polynomial-time
algorithm that reduces any  instance $(x, k)$ of a parameterized problem to an equivalent instance $(x',k')$ of
the same problem whose size is bounded by $f(k)$ for some function $f$, which is the \textit{size} of
the kernelization. A kernelization algorithm, or just \textit{kernel}, can be seen as a preprocessing procedure with provable guarantees, and it is fundamental to find kernels of the smallest possible size, ideally polynomial. Identifying which parameterized problems admit polynomial kernels is one of the most active areas within Parameterized Complexity (cf. for instance~\cite{FominLSZ19}).

When dealing with a problem where the goal is to find a (say, small) subset of vertices $S$ of an input graph $G$ satisfying some property, such as \textsc{Vertex Cover}, it is natural to consider as the parameter the size of the desired set $S$. Assuming that the problem admits a polynomial kernel parameterized by $|S|$, as it is the case for \textsc{Vertex Cover} and many other problems~\cite{CyganFKLMPPS15,FominLSZ19}, we can ask whether the problem still admits polynomial kernels when the parameter is (asymptotically) smaller than the solution size. The goal of this approach is to provide better preprocessing guarantees, as well as to understand what is the limit of the polynomial-time ``compressibility'' of the considered problem. For problems defined on graphs, apart from using the solution size as a parameter, it is common  to consider so-called \textit{structural parameters}, which quantify some structural property of the input graph that can be seen as a measure of its ``complexity''. Among structural parameters, the most successful is probably \textit{treewidth}~\cite{Kloks94,CyganFKLMPPS15}, but unfortunately taking the treewidth of the input graph as the parameter does not allow for polynomial kernels for essentially all natural optimization problems, unless ${\sf NP} \subseteq {\sf coNP}/{\sf poly}$~\cite{BodlaenderJK14,BodlaenderDFH09}. The same applies to another relevant graph parameter called \textit{treedepth}, denoted by $\td$ and defined as the minimum number of rounds needed to obtain the empty graph, where each round consists of removing one vertex from each connected component.

In fact, the lower bound proofs go through for parameterizations for which the value on a disconnected graph is the \textit{maximum}, rather than the \textit{sum}, of the values of its components, and whose value is polynomially bounded  in the size of the graph. Hence, to be able to obtain positive kernelization results, we need to turn to parameterizations other than \emph{width measures}. This motivates to consider structural parameters that quantify the ``distance from triviality'', a concept first coined by Guo, H{\"{u}}ffner, and  Niedermeier~\cite{GuoHN04}.  The idea is to take as the parameter the vertex-deletion distance of a graph to a ``trivial'' graph class where the considered problem can be solved efficiently. This paradigm has proved very successful for a number of problems, in particular for \textsc{Vertex Cover}. In an influential work, Jansen and Bodlaender~\cite{JansenB13} showed that {\sc Vertex Cover}  admits a polynomial kernel
when parameterized by the feedback vertex number of the input graph,
which is the vertex-deletion distance to the ``trivial'' class of forests. This result triggered a number of results in the area, aiming to characterize the ``trivial'' families $\F$ for which {\sc Vertex Cover} admits a polynomial kernel under this parameterization~\cite{MajumdarRR18,BougeretS18,FominS16,HolsKP22,BougeretJS22}.

Let us mention some of these results that are relevant to our work. Bougeret and Sau~\cite{BougeretS18} proved that \textsc{Vertex Cover} admits a polynomial kernel parameterized by the vertex-deletion distance to a graph of bounded treedepth. This result was further generalized into two orthogonal directions, namely by considering a more general problem or a more general target graph class $\F$. For the former generalization, Jansen and Pieterse~\cite{JansenP20} proved that the following problem also admits a polynomial kernel parameterized by the vertex-deletion distance to a graph of bounded treedepth: for a fixed finite set of connected graphs $\M$, the $\M$-{\sc Minor Deletion} problem consists in deleting the minimum number of vertices from an input graph to obtain a graph that does not contain any of the graphs in $\M$ as a minor. Note that this problem (vastly) generalizes \textsc{Vertex Cover}, which corresponds to the case $\M=\{K_2\}$. For the latter generalization, Bougeret, Jansen, and Sau~\cite{BougeretJS22} proved that \textsc{Vertex Cover} admits a polynomial kernel parameterized by the vertex-deletion distance to a graph of bounded \textit{bridge-depth}, which is a parameter that generalizes treedepth and the feedback vertex number. It turns out that, under the assumption that the target graph class $\F$ is minor-closed, the property of $\F$ having bounded bridge-depth is also a {\sl necessary} condition for \textsc{Vertex Cover} admitting a polynomial kernel. Another complexity dichotomy of this flavor has been achieved by Dekker and Jansen~\cite{DekkerJ22} for the \textsc{Feedback Vertex Set} problem (with another characterization of the target graph class $\F$). These complexity dichotomies, while precious, are unfortunately quite hard to obtain, and the current knowledge seems still far from obtaining dichotomies of this type for general families of problems, such as for $\M$-{\sc Minor Deletion} for any finite family of graphs $\M$. Indeed, it is wide open whether $\M$-{\sc Minor Deletion} admits a polynomial kernel parameterized by the solution size for any set $\M$ containing only non-planar graphs~\cite{FominLMS12,Jansen022}, so considering parameters smaller than the solution size is still out of reach.

\subparagraph*{Our contribution.} We consider an alternative generalization of {\sc Vertex Cover} by considering (induced) \textit{subgraphs} instead of \textit{minors}. Namely, for a fixed graph $H$, the $H$-{\sc Subgraph Hitting} problem is defined as deleting the minimum number of vertices from an input graph to obtain a graph without any occurrence of $H$ as a subgraph. The $H$-{\sc Induced Subgraph Hitting} problem is defined analogously by forbidding occurrences of $H$ as an {\sl induced} subgraph.
(As we shall see later, both problems behave in the same way with respect to our results.) Note that both problems correspond to {\sc Vertex Cover} for the case $H=K_2$ and therefore are indeed generalization of it. As opposed to the case for hitting minors, it is well-known that both the $H$-{\sc Subgraph Hitting} and $H$-{\sc Induced Subgraph Hitting} problems admit polynomial kernels parameterized by the solution size for any graph $H$~\cite{Abu-Khzam10,FlumG06}.

Therefore, it does make sense to parameterize these problems by structural parameters in the ``distance from triviality'' spirit, and this is the main focus of this article. To be best of our knowledge, this is an unexplored topic, besides all the literature for {\sc Vertex Cover} discussed above. Our main result is to identify structural parameters that allow to provide sharp dichotomies for these problems {\sl depending on the forbidden (induced) subgraph $H$}.

Before presenting our results, we proceed to motivate and define these structural parameters. They are inspired by the following parameter, first
introduced by Bulian and Dawar~\cite{BulianD16graph,BulianD17fixe} and further studied, for instance, in~\cite{PilipczukSSTV22,JansenK021verte,HolsKP22,FominGT22param,AgrawalKLPRSZ22delet,MorelleSST23}. For a fixed class of graphs $\H$, the $\H$-\textit{elimination distance} of a graph $G$, denoted by $\edcalH(G)$, is defined by mimicking the above definition of treedepth and replacing ``empty graph'' with ``a graph in~$\H$''. The recursive definitions of treedepth and $\H$-elimination distance suggest the notion of \textit{elimination forest}, which is the forest-like process of vertex removals from the considered graph to obtain either an empty graph for treedepth, or a graph in $\H$ for $\H$-elimination distance. Suppose now that $\H$ is defined by the exclusion of a fixed graph $H$ as a subgraph or as an induced subgraph. Formally, let
$\F_{\bar H}$ (resp. $\F^{\mathsf{ind}}_{\bar H}$) be the class of graphs that exclude a fixed graph $H$ as a subgraph (resp. induced subgraph). For this particular case, the notion of $\F_{\bar H}$-elimination distance (or $\F^{\mathsf{ind}}_{\bar H}$-elimination distance) can be interpreted as a generalization of treedepth where, in the last round of the elimination process, the vertices that do not belong to any occurrence of $\H$ as a subgraph (or induced subgraph) can be deleted ``for free''. We generalize the notion of $\H$-elimination distance by allowing  ``free removal'' of vertices not contained in a copy of $H$ in {\sl every round} of the elimination process, rather than just the last; we denote the corresponding parameter by $\veds$ (or  $\vedis$), where `${\sf v}$' stands for the removal of {\sl vertices}, in order to distinguish this parameter from the one defined below (see \autoref{sec:overview} for the formal definitions of these parameters). Our first main result is the following somehow surprising dichotomy, which states that, under the assumption that $H$ is biconnected,
whenever $H$ has a non-edge, the problem is unlikely to admit a polynomial kernel.
Our result applies to both the induced and non-induced versions of the problem.

\begin{theorem}\label{thm:dichotomomy}
Let $H$ be a biconnected graph, let $\lambda \geq 1$ be an integer, and assume that ${\sf NP} \nsubseteq {\sf coNP}/{\sf poly}$. $H$-{\sc Subgraph Hitting}  (resp. $H$-{\sc Induced Subgraph Hitting}) admits a polynomial kernel parameterized by the size of a given vertex set $X$ of the input graph $G$ such that $\veds(G-X) \leq \lambda$ (resp. $\vedis(G-X) \leq \lambda$) if and only if $H$ is a clique.
\end{theorem}

Note that a graph $G$ satisfies $\veds(G) = 0$ (resp. $\vedis(G) = 0$) if and only if $G$ does not contain $H$ as a subgraph (resp. induced subgraph), so the setting $\lambda=0$ corresponds to the parameterization by solution size which always admits a polynomial kernel~\cite{Abu-Khzam10,FlumG06}; this is why we assume that $\lambda \geq 1$ in the statement of \autoref{thm:dichotomomy}.


\autoref{thm:dichotomomy} shows that the behavior of the considered problems in terms of the existence of polynomial kernels drastically changes as far as one edge is missing from $H$ (under the biconnectivity assumption, which is needed in the reduction). To the best of our knowledge, this is the first time that such a dichotomy, in terms of $H$, is found with respect to the existence of polynomial kernels. It is worth mentioning that, with respect to the existence of certain fixed-parameter tractable algorithms parameterized by treewidth, dichotomies of this flavor exist for hitting subgraphs~\cite{CyganMPP17}, induced subgraphs~\cite{SauS21}, or minors~\cite{BasteST23}.


The proof of \autoref{thm:dichotomomy} consists of two independent pieces. On the one hand, we need to prove that both problems admit a polynomial kernel when $H$ is a clique (note that, in that case, both problems are equivalent, as any $H$-subgraph is induced). On the other hand, we need to provide a kernelization lower bound for all other graphs $H$ (cf. \autoref{thm:hardness-lambda1}), and here is where we need the hypothesis that ${\sf NP} \nsubseteq {\sf coNP}/{\sf poly}$ and, for technical aspects of the reduction, that $H$ is biconnected.

In fact, we provide a kernel that is more general than the one stated in \autoref{thm:dichotomomy}. Also, on the negative side, we present another lower bound incomparable to that of \autoref{thm:dichotomomy}. For the former, we provide a polynomial kernel, when $H$ is a clique, for a parameter that is more powerful than $\veds$ (or $\vedis$).
 This more powerful parameter is somehow inspired by the parameter bridge-depth mentioned before~\cite{BougeretJS22}, which is a generalization of treedepth in which, in every round of the elimination process, we are allowed to remove subgraphs in each component that are more general than just single vertices. In our setting, it turns out that we can afford to remove vertex sets $T \subseteq V(G)$ that induce connected subgraphs that do not contain $H$ as a subgraph (or induced subgraph) and that are ``weakly attached'' to the rest of the graph, meaning that each connected component of $G-T$ has at most one neighbor in $T$. If $H$ is biconnected, it is easily seen that the ``candidate'' sets $T$ to be removed can be assumed to be connected unions of blocks (biconnected components) of $G$, and this is why we call this parameter $\beds$ (or $\bedis$), where `${\sf b}$' stands for the removal of {\sl blocks}. For any two graphs $G$ and $H$, the following inequalities, as well as the corresponding ones for the induced version, follow easily from the definitions (cf. \autoref{sec:overview}):
\begin{equation}\label{eq:parameters}
\td(G) \geq {\sf ed}_{\F_{\bar H}}(G) \geq \veds(G) \geq \beds(G).
\end{equation}
 We prove the following result, where $K_t$ denotes the clique on $t$-vertices, and note that in this case the induced and non-induced versions of the problem coincide.


\begin{theorem}\label{thm:mainkernel}
Let $t\geq 3$ and $\lambda\geq 1$ be fixed integers. The $K_t$-{\sc Subgraph Hitting}  problem admits a polynomial kernel parameterized by the size of a given vertex set $X$ of the input graph $G$ such that $\bedclique(G-X) \leq \lambda$.
\end{theorem}

Note that for $t=2$, the parameter $\bedclique$ is exactly treedepth, and therefore \autoref{thm:mainkernel} can be seen as a far-reaching generalization of the main result of~\cite{BougeretS18}, that is, a polynomial kernel for {\sc Vertex Cover} parameterized by the vertex deletion distance to a graph of bounded treedepth.
Also, note that \autoref{eq:parameters}
and \autoref{thm:mainkernel} imply that the same dichotomy stated in \autoref{thm:dichotomomy} holds if we replace `$\veds(G-X) \leq \lambda$' with `$\beds(G-X) \leq \lambda$', and the same for the induced version.



As for strengthening our hardness results, we present kernelization lower bounds for $H$-{\sc Subgraph Hitting} and $H$-{\sc Induced Subgraph Hitting}, when $H$ is not a clique, parameterized by the vertex-deletion distance to a graph of constant {\sl treedepth}. By \autoref{eq:parameters}, lower bounds for treedepth are stronger than the ones of
\autoref{thm:dichotomomy}. However, in our next main result, we need a condition on $H$ that is stronger than biconnectivity, namely the non-existence of a \textit{stable cutset}, that is, a vertex separator that induces an independent set.


\begin{restatable}{theorem}{hardness-treedepth}
\label{thm:hardness-treedepth}
Let $H$ be a graph on $h$ vertices that is not a clique and that has no stable cutset. $H$-{\sc Subgraph Hitting} and $H$-{\sc Induced Subgraph Hitting} do not admit a polynomial kernel parameterized by the size of a given vertex set $X$ of the input graph $G$ such that $\td(G-X) = \O(h)$, unless ${\sf NP} \subseteq {\sf coNP}/{\sf poly}$.
\end{restatable}

Note that, for $t \geq 4$, the graph $K_t$ minus one edge satisfies the conditions of \autoref{thm:hardness-treedepth}. The mere existence of a graph $H$ satisfying the conditions of \autoref{thm:hardness-treedepth} is remarkable, as it shows that  (induced) subgraph hitting problems behave differently than minor hitting problems. Indeed, as mentioned before, it is known~\cite{JansenP20} that, for {\sl every} finite family $\M$ of connected graphs, the $\M$-{\sc Minor Deletion} problem admits a polynomial kernel parameterized by the vertex-deletion distance to a graph of constant treedepth.

Dekker and Jansen~\cite{DekkerJ22} asked if for every finite set of graphs $\M$, $\M$-{\sc Minor Deletion}  admits a polynomial kernel parameterized by the vertex-deletion distance to a graph with constant ${\sf exc(\M)}$-elimination distance, where ${\sf exc(\M)}$ is the class of graphs that excludes all the graphs in $\M$ as a minor. \autoref{thm:hardness-treedepth} shows that, by \autoref{eq:parameters}, for the problems of excluding subgraphs or induced subgraphs, the answer to this question is negative.

Finally, let us mention another consequence of our results. Agrawal et al.~\cite{AgrawalKLPRSZ22delet} proved, among other results, that for every hereditary target graph class $\C$ satisfying some mild assumptions,
parameterizing by the vertex-deletion distance to $\C$ and by the $\C$-elimination distance are equivalent from the point of view of the existence of fixed-parameter tractable algorithms. Our results imply, in particular, that the same equivalence does {\sl not} hold with respect to the existence of polynomial kernels in this ``distance from triviality'' setting, namely for problems defined by the exclusion of (induced) subgraphs.


\subparagraph*{Organization of the paper.}  In \autoref{sec:overview} we provide an overview of the main ideas of the kernelization algorithm, which is our main technical contribution. The full version of the kernel, which is quite lengthy, is presented from \autoref{sec:prelim} to  \autoref{sec:computeHittingSet}. More precisely, in \autoref{sec:prelim} we provide some preliminaries for the analysis of blocking sets (defined later) and the kernel. In \autoref{sec:blocking} we bound the size of minimal blocking sets, which is a crucial ingredient for the analysis of the kernelization algorithm, whose actual description can be found in \autoref{sec:kernel}, namely in \autoref{algo:kernel}. Its analysis also needs the results proved in \autoref{sec:computeHittingSet} (which are somehow independent)  to compute the decomposition associated with the parameter $\bedclique$ along with optimal solutions in graphs with bounded $\bedclique$. In \autoref{sec:hardness} we present our hardness results, and in \autoref{sec:conclusions} we discuss some directions for further research.

\section{Overview of the kernelization algorithm}
\label{sec:overview}

In this section we sketch the main ideas of the kernelization algorithm stated in \autoref{thm:mainkernel}, along with intuitive explanations and the required definitions.

\subparagraph*{Preliminaries.} Given a graph $G$ and $C \subseteq V(G)$, we denote by $N(C)=\bigcup_{v \in C}N(v)\setminus C$ and $G-C=G[V(G)\setminus C]$.
Given a graph $H$, and a subgraph (resp. induced subgraph) $F$ of $G$, we say that $F$ is a \textit{copy} (resp. \textit{induced copy}) of $H$ if $F$ is isomorphic to $H$. We say that $G$ is $H$-free (resp. $H$-induced free) if there is no copy of $H$ (resp. induced copy) in $G$.
Given two disjoint subsets $A,B \subseteq V(G)$, we say that there is an edge between $A$ and $B$ if there exists $e \in E(G)$ such that $|e \cap A|=|e \cap B|=1$. Otherwise, $A$ and $B$ are said to be \textit{anticomplete}.
When $\P=\{V_1, \ldots, V_m\}$ is a set of subsets of $V(G)$, we let $V(\P)=\bigcup_{V_i \in \P} V_i$.
A \textit{$t$-clique} is a set $K \subseteq V(G)$ such that $G[K]$ is a clique and $|K|=t$.
For any integer $n \in \mathbb{N}$, we denote by $[n]=\{1,\dots,n\}$. We study the following problem(s) for a fixed graph $H$.

\begin{center}
\fbox{
\begin{minipage}{13.5cm}
\noindent{{\sc \HH} (for {\sc $H$-{\sc Subgraph Hitting}})}\\
\noindent\textbf{Input}:~~A graph $G$.\\
\textbf{Objective}:~~Find a set $S\subseteq V(G)$ of minimum size such that $G-S$ is $H$-free.
\end{minipage}
}
\end{center}
We denote by $\HiH$ (for {\sc $H$-{\sc Induced Subgraph Hitting}}) the variant of the above problem where we impose that $G-S$ is $H$-\emph{induced} free.
We denote by $\opt^H(G)$ the optimal value of the considered problem for $G$, or simply $\opt(G)$ when $H$ is clear from the context.

Let us now introduce the main graph measures used in this paper.

\begin{definition} \label{def:vedH}
Let $H$ be a fixed graph. For a graph~$G$, define~$\veds(G)$ as
\begin{equation*}
    \begin{cases}
    0 & \mbox{if~$V(G) = \emptyset$,} \\
    \veds(G-v) & \mbox{if~$v$ is a vertex that is not in any copy of $H$,} \\
    \max_{C_i} \veds(C_i) & \mbox{if~$G$ has connected components~$C_1,\ldots, C_c$ with $c \ge 2$,} \\
    1+\min_{v \in V(G)} \veds(G-v) & \mbox{otherwise.}
    \end{cases}
\end{equation*}
\end{definition}
We define $\vedis$ in the same way as $\veds$, except that we replace ``in any copy of $H$'' by ``in any induced copy of $H$''. Note that the notation $\veds$ is motivated by the fact that it corresponds to \emph{vertex} elimination distance, with additional power of removing ``free'' vertices not in any copy of $H$. Note also that even though there could be multiple vertices~$v$ which satisfy the second criterion, the value is well-defined since it does not matter which one is picked; the second case will apply until all such vertices have been removed.
As the case where $H=K_t$ plays an important role in this paper, for the sake of shorter notation we use the shortcut $\vedt$ to denote the parameter $\vedclique$ (or $\vediclique$, which is the same).



To define the parameters $\beds$ and $\bedis$, it is convenient to introduce the following definitions (see \autoref{fig:root}).
\begin{definition}[root and pending component]\label{def:pending}
Given a fixed graph $H$ and a connected graph $G$, we say that a set $T \subseteq V(G)$ is a \emph{root of $G$} if
\begin{itemize}
    \item $T \neq \emptyset$, $G[T]$ is connected and $H$-free, and
    \item for any connected component $C$ of $G-T$, $|N(C) \cap T|=1$.
\end{itemize}
We extend to notion of root to any graph $G$ as follows. For any graph $G$ with connected components $\mathcal{C}$, we say that a set $\T=\{T_C \mid C \in \mathcal{C}\}$ is a \emph{root of $G$} if for any $C \in \mathcal{C}$,
$T_C$ is a root of $G[C]$.

Given a graph $G$ and a root $\T$ of $G$, we define the \emph{pending component of $v$ relatively to $\T$}, denoted by $C^{\T}(v)$, as the connected component of $v$ in the graph obtained from $G$ by removing all edges $e \subseteq E(\T)$.
We extend the notation to any subset $Z \subseteq V(\T)$ with
$C^{\T}(Z) = \bigcup_{v \in Z}C^{\T}(v)$. When the root is clear from context, we use $C(v)$ instead of $C^{\T}(v)$.
\end{definition}

\begin{figure}
\begin{center}
    \includegraphics[scale=.75]{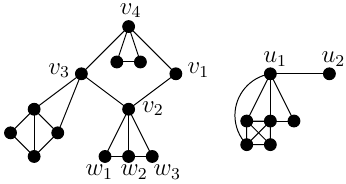}
    \caption{In this example we consider that $H=K_4$, and we denote by $C_1$ and $C_2$ the two connected components of $G$ (where $v_1 \in C_1$). Observe that $\T=(T_1,T_2)$ is a root of the depicted graph with $T_1 = \{v_1,v_2,v_3,v_4\}$ and $T_2 = \{u_1,u_2\}$. We have $C(v_1)=\{v_1\}$ and $C(v_2)=\{v_2,w_1,w_2,w_3\}$. Finally, taking $T'_1 = T_1 \cup \{w_1\}$ and $\T' = \{T_1',T_2\}$ would not be a root as $T'_1$ is not a root of $G[C_1]$.    }
    \label{fig:root}
\end{center}
  \end{figure}

We define an \textit{induced root} in the same way, except that we replace ``$G[T]$ is connected and $H$-free'' by ``$G[T]$ is connected and $H$-induced free''. Note that any graph admits a root, by taking for example a single vertex (to play the role of $T_C$) in each connected component.

\begin{observation}\label{obs:lambda4}
Let $G$ be a graph and $\T$ be a root of $G$.
For any $v \in V(\T)$, there is no edge between $C(v)\setminus \{v\}$ and $V(G)  \setminus C(v)$.
\end{observation}

\begin{definition} \label{def:bedHplus}
Let $H$ be a fixed graph. For a graph~$G$,  we define~$\beds(G)$ as
\begin{equation*}
    \begin{cases}
    0 & \mbox{if~$V(G) = \emptyset$,} \\
    \beds(G-v) & \mbox{if~$v$ is a vertex that is not in any copy of $H$} \\
    \max_{C_i} \beds(C_i) & \mbox{if~$G$ has connected components~$C_1,\ldots, C_c$ with $c \ge 2$,} \\
    1+\min_{T \subset V(G)} \beds(G-T) & \mbox{otherwise, where~$T$ ranges over  all roots of $G$.}
    \end{cases}
\end{equation*}
\end{definition}
We define $\bedis$ in the same way as $\beds$, except that we replace ``where $T$ ranges over all roots of $G$'' by  ``where $T$ ranges over all induced roots of $G$''.
Again, as the case where $H=K_t$ plays an important role in this paper,
 for the sake of shorter notation we use the shortcut $\vedt$ to denote the parameter $\bedclique$ (or $\bediclique$, which is the same).
We point out that, to make the definition of $\beds$ as simple as possible, we allowed $T$ to range over all roots of $G$.
However, as shown in \autoref{lemma:OPTforboundedlambda}, as soon as $H$ is biconnected, there always exists a root that is the connected union of $H$-(induced-)free blocks of $G$, hence our choice of notation to differentiate $\beds$ from $\veds$.


Given $\lambda \in \mathbb{N}$, let us now define the following variant of the considered problem, where we suppose that we are given as an additional input a modulator (corresponding to set $X$) to a ``simple'' graph $G-X$, where the simplicity is captured by $\bedt$ being at most $\lambda$.

\begin{center}
\fbox{
\begin{minipage}{13.6cm}
\noindent{{\sc $\KtHM^\lambda$} (for {\sc $K_t$-Subgraph Hitting given a modulator to $\bedt$ at most $\lambda$})}\\
\noindent\textbf{Input}:~~A graph $G$ and a set $X \subseteq V(G)$ such that $\bedt(G[R]) \le \lambda$, where $R = V(G)\setminus X$.\\
\textbf{Objective}:~~Find a set $S\subseteq V(G)$ of minimum size such that for any $t$-clique $Z$ of $G$, $S \cap Z \neq \emptyset$.
\end{minipage}
}
\end{center}

 We denote by $\KtHMdecbis{\lambda}$ the associated parameterized decision problem with an additional $k$ in the input, where the goal to decide whether $\opt(G) \le k$, and the parameter is $|X|$.

In \autoref{sec:kernel} we prove our main positive result that we restate here with less details, and which is a reformulation of \autoref{thm:mainkernel} with the notation introduced in this section:
\begin{theorem}\label{thm:kernellight}
  There is a polynomial kernel for $\KtHMdecbis{\lambda}$ of size $\O_{\lambda,t}(|X|^{\delta(\lambda,t)})$ for some function  $\delta(\lambda,t)$.
\end{theorem}
Let us now overview the techniques used to establish the above result.

\subparagraph*{Warming up with {\sc Vertex Cover}.}
As an extreme simplification of our set up, let us consider the case where $t=2$, corresponding to  {\sc Vertex Cover}, and assume that $X$ is a modulator to the simplest graph class, namely an independent set.
Our kernel uses a marking procedure (\autoref{def:mark}) that corresponds to the following marking algorithm for  {\sc Vertex Cover}.
For any $u \in X$, mark up to $|X|+1$ vertices $v \in R$ such that $\{u,v\} \in E(G)$. Let $M \subseteq R$ be the set of marked vertices.
Observe the following ``packing property'' of the marking algorithm: if there exists $v \in R\setminus M$ and $u \in X$ with $\{u,v\} \in E(G)$,
then there exists a ``packing'' $\P \subseteq M$ of $|X|+1$ vertices such that for any $v' \in \P$, $\{v',u\} \in E(G)$ (the term ``packing'' may seem inappropriate
here, but becomes natural for $t > 2$ as the marking algorithm will mark disjoint sets of vertices instead of distinct vertices).
Now, if $R=M$ then $|R| \le \O(|X|^2)$ and the instance is kernelized. Otherwise, if there exists $v \in R\setminus M$, then define a \textit{reduced instance} as $G'=G-v$ and $k'=k$.
Let us sketch why this removing step is safe, as the arguments also correspond to a very simplified version of \autoref{lemma:step2safe}. The only non-trivial direction is that if $(G',k')$ is a \yes-instance, then  $(G,k)$ is also a \yes-instance.
Given a solution $Z'$ of $(G',k')$, if there exists $u\in X \setminus Z'$ such that $\{u,v\} \in E(G)$,
then the packing property implies the existence of the above set $\P$. Thus, we get that $Z'$ \emph{overpays for $u$}: it contains one extra vertex (in this case, one instead of zero) for each $v' \in \P$ as we must have $\P \subseteq Z'$. This implies that we can restructure $Z'$ into $\tilde{Z}=X$, while ensuring that $|\tilde{Z}| \le |Z'|$. Now, $\tilde{Z}$ can be easily completed to a solution of $G$ of size $k$ (in this case, by doing nothing).
By repeating this reduction rule, we get the kernel of size~$\Oh(|X|^2)$.

\subparagraph*{Parts, chunks, and conflicts.}
Let us now point out some important ideas used to lift up the previous kernel for  {\sc Vertex Cover} to $\KtHMdecbis{\lambda}$.
In the previous setting, the key property when proving the safeness of the reduction rule, given a solution $Z'$ of $(G',k')$, is the following:
when ``adding back'' a non-marked vertex $v \in R \setminus M$ to $G'$,  either there exists $u \in X \setminus Z'$ such that $Z'$ overpays for $u$, or there is no edge $\{u,v\}$ for any $u \in X \setminus Z'$.

Let us now rephrase this key property in the setting of hitting $t$-cliques using the adapted concepts of part, chunk, and conflict, and we will formally define these terms later.
When ``adding back'' a non-marked \emph{part} $V' \subseteq R \setminus M$ to $G'$, we know that either there exists \emph{a chunk} $X' \subseteq X \setminus Z'$ such that $Z'$ overpays for $X'$, or there is no \emph{conflict between $X'$ and $V'$} for any chunk $X'$.
Observe first that, as $G-X$ is now more general than an independent set, we have to consider a packing of ``parts'' (subsets of vertices of $R$), meaning that if there is a non-marked part $V'$ that we remove, we now set $k'=k-\opt(G[V'])$. The second difference is the notion of ``conflict  between $X'$ and $V'$'' that plays the role of ``edge $\{u,v\}$''.
We say that there is no \emph{conflict} between $X'$ and $V'$ if $\conf^t_{X'}(V')=0$, the condition  $\conf^t_{X'}(V')=0$ being equivalent to the fact that we can pick only $\opt(G[V'])$ vertices in $V'$, while still hitting all $t$-cliques in $G[X' \cup V']$ (see \autoref{def:conflict} for the formal definition of $\conf^t$).
The third difference is the notion of chunk and blocking set. A good starting point when trying to complete a solution $Z'$ of $G'$ to a solution $Z$ of $G$ is that $\conf^t_{X \setminus Z'}(V')=0$.
Indeed, this condition implies that there exists a set $S^\star_{V'}$ of size $\opt(G[V'])$ such that $S^\star_{V'}$ hits all $t$-cliques in $G[V' \cup (X \setminus Z')]$. Thus, $S^\star_{V'}$ is a good candidate to build a solution $Z = Z' \cup S^\star_{V'}$ of $(G,k)$. Note that this only remains a good starting point, as $Z$ may not be a solution: it could miss cliques using $V'$, $(X \setminus Z')$, and other vertices in $R \setminus V'$.
This condition $\conf^t_{X \setminus Z'}(V')=0$ could be achieved by a marking algorithm that, for any $U \subseteq X$, marks up to $|X|+1$ parts $V'$ such that $\conf^t_{U}(V')>0$, which is  a generalization of the previous marking algorithm for  {\sc Vertex Cover}.
However, the running time and the number of marked parts by such an algorithm would not be polynomial in $|X|$, as there are too many subsets $U$ to consider.

To overcome this issue, the trick is the notion of \textit{maximum minimal blocking sets}, denoted by $\mmbs_t$ (\autoref{def:bsKt}), which is a graph parameter for which we skip the definition for the moment. What is important to state about $\mmbs_t$ here is that in \autoref{thm:mmbslambda4} in \autoref{sec:blocking}, we prove that there exists a function $\beta: \mathbb{N}^2 \to \mathbb{N}$ such that for every graph $G$, $\mmbs_t(G) \le \beta(\bedt(G),t)$. As in the kernelization algorithm we apply this to $G[R]$ and $\bedt(G[R]) \le \lambda$, we obtain $\mmbs_t(G[R]) \le \beta(\lambda,t)$. Moreover, \autoref{lemma:smallcertif} implies that the previous marking condition ``there exists $U \subseteq X$ such that $\conf^t_{U}(V')>0$'' is equivalent to ``there exists $X' \in \X$ such that $\conf^t_{X'}(V')>0$'', where $\X = \{X' \subseteq X \mid \mbox{ $|X'| \le (t-1)\beta(\lambda,t)$ and $X'$ does not contain a $t$-clique} \}$ is the set of \emph{chunks}.
Observe that, as the chunks have bounded size, the marking algorithm runs in time $\O^{\star}(|X|^{(t-1)\beta(\lambda,t)})$.
The conclusion is that the ``triviality'' of $G[R]$ ($\bedt(G[R]) \le \lambda$) implies that $G[R]$ has bounded $\mmbs_t$, which allows to certify the absence of conflict (for any $U \subseteq X$) in polynomial time.

These notions of conflict, chunk, and minimal blocking set were also critical in previous work on kernelization: the notion of conflict was introduced in \cite{JansenB13}, and bounds on $\mmbs_2$ (for  {\sc Vertex Cover}) have been proved for different triviality measures of $G[R]$~\cite{HolsKP22, BougeretS18, JansenB13, BougeretJS22,AraujoBCS23}.
 A first difference here is the study of $\mmbs_t$ for {\sc $K_t$-Subgraph Hitting}, whose behavior is more complex than  {\sc Vertex Cover}, as discussed below when mentioning the new challenges.

\subparagraph*{Decreasing $\bedt(G[R])$ and using recursion.}
Let us now informally describe the main steps of the kernel (see \autoref{fig:overallkernel}). Given a graph $G$, we denote by $N^t(G)=\{v \in V(G) \mid \nexists \mbox{ $t$-clique $K$ with $v \in K$}\}$ the set of vertices of $G$ that do not occur in any copy of $K_t$, called the \textit{non-$K_t$-vertices}.
Given an input $(G,X,k)$ of $\KtHMdecbis{\lambda}$, we first compute  $N^t(G[R])$ and, using the algorithm of \autoref{lemma:OPTforboundedlambda}, a $\bedt$-root $\T$
of $G[R]-N^t(G[R])$, where a \textit{$\bedt$-root} of $G$ is a root $\T$ such that $\bedt(G-V(\T)) = \bedt(G)-1$.
We point out that, unlike  the case for treedepth or bridge-depth, computing such a root is not straightforward, as one cannot try the a priori exponentially many possible roots to find one that decreases $\bedt$.
However, the algorithm of \autoref{lemma:OPTforboundedlambda} relies on the fact that it is possible to compute in polynomial time a set of size~$\Oh(n)$ that contains a $\bedt$-root.
Coming back to the kernel strategy, observe that there may be edges between some $C(v)$ and $N^t(G)$, but not between $C(v)$ and $C(u)$ for $u \neq v$, and that by definition of a $\bedt$-root, $\bedt(G[R']) < \lambda$, where $R'=R-V(\T)$.
      Then, we mark a small (polynomial in $|X|$) set $M(\T,N^t(G[R]),G,X)$ of vertices (\autoref{def:step1}) of $V(\T)$ using the {\tt mark} algorithm (\autoref{def:mark}).
      If there exists $v \in V(\T) \setminus M(\T,N^t(G[R]),G,X)$, then we can remove $C(v)$ and decrease $k$ by $\opt(G[C(v)])$ (\autoref{lemma:step2safe}).
      Otherwise, $|M(\T,N^t(G[R]),G,X)|=\O(|X|^{f(\lambda,t)})$ for some function $f$, and thus we can move $M(\T,N^t(G[R]),G,X)$ to the modulator and get a new modulator $X' = X \cup M(\T,N^t(G[R]),G,X)$ whose
      size is still polynomial in $|X|$.
      The key point is that $\bedt(G-X') = \bedt(G[R']) < \lambda$, and thus we use induction on $\lambda$ and make a recursive call to $(G,X')$,
      which is an input of $\KtHMdecbis{\lambda-1}$, leading to a kernel  polynomial in $|X'|$, and thus in $|X|$.

  \begin{figure}
\begin{center}
    \includegraphics[scale=0.65]{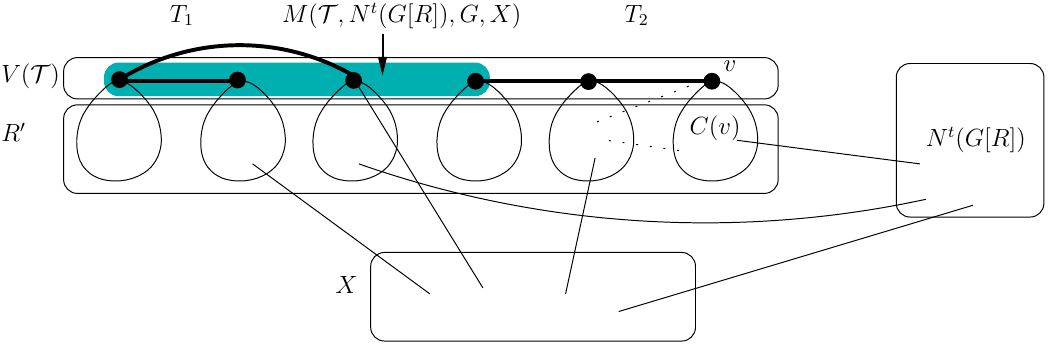}
    \caption{Main steps of the kernel. In this example $\T = \{T_1,T_2\}$ (edges inside $T_i$ are in bold, and dotted edges cannot exist), and there exists a non-marked vertex $v$, implying that the pending component $C(v)$ will be removed.}
    \label{fig:overallkernel}
\end{center}
  \end{figure}

This idea of shrinking the ``root'' of a decomposition of $G-X$ to decrease the ``triviality measure'' (here, $\bedt$) and recurse originates in~\cite{GajarskyHOORRVS13}, and was used in~\cite{BougeretS18} for treedepth.
It was subsequently generalized in~\cite{BougeretJS22}, where the triviality measure is a parameter called \textit{bridge-depth} and the equivalent of a root is a so-called \textit{tree of bridges} for each connected component of $G[R]$.

\subparagraph*{New challenges.}
With respect to the strategies followed in previous work on related topics~\cite{BougeretS18,BougeretJS22,HolsKP22,JansenP20,DekkerJ22,JansenK11,JansenB13}, in our setting we encounter (at least)
 the following three orthogonal difficulties, for which we have to develop new ideas: dealing with the non-$K_t$-vertices, dealing with cliques $K_t$ for arbitrary fixed $t$ instead of $t=2$, and proving that there exists a function $\beta$ such that for every graph $G$, $\mmbs_t(G) \le \beta(\bedt(G))$.

The first difficulty is handling vertices of $N^t(G[R])$, which are vertices not belonging to a $t$-clique in $G[R]$.
Indeed, observe that these non-$K_t$-vertices are ``free'' for $\bedt$, in the sense that $\bedt(G[R])=\bedt(G[R]-N^t(G[R]))$.
However, these vertices make the structure of $G[R]$ more complicated. Indeed, $\T$ being a root of $G[R]-N^t(G[R])$ implies that for any $T \in \T$ and $v \in T$,
there are no edges between $C(v) \setminus \{v\}$ and other vertices in $C(u)$ for $u \in V(\T) \setminus \{v\}$, but there could be edges between $C(v)$ and $N^t(G[R])$.
Thus, unlike in \cite{BougeretS18,BougeretJS22}, we cannot just bound the number of connected components of $G[R]$, and then assume that we have a single root $T$ with simple properties
(again, the root being a single vertex in~\cite{BougeretS18}, and a tree of bridges in \cite{BougeretJS22}).
Typically, here $G[R]$ could have only one connected component, but  the nice structure given by $\T$ could be ``polluted'' by vertices of $N^t(G[R])$.
We handle these vertices by considering a packing of ``bidimensional'' parts $(V_i,N_i)$, where in particular $V_i \subseteq V(\T)$ is a clique of size at most $t-1$
and $N_i \subseteq N^t(G[R])$, and we use a kind of generalized ``sunflower-like'' marking by first creating a maximal packing $\P$ of parts $(V_i,N_i)$, of size at most $|X|+1$,
and then recursively marking around each possible $g \in \bigcup N_i$ (see the last line of \autoref{def:mark}).

The second difficulty is to handle $t$-cliques instead of edges.
Indeed, assume that we just removed a pending component $C(v)$ for some $v \in V(\T)$ and defined $k'=k-\opt(G[C(v)])$.
Assume also that, given a solution $Z'$ of $(G',k')$, we have the good starting point $\conf^t_{(X \cup N^t(G[R])) \setminus Z'}(C(v)) = 0$, implying, by the definition of conflict, that there exists a locally optimal solution $S^\star_v$ of $G[C(v)]$ that intersects all $t$-cliques of $G[C(v) \cup ((X \cup N^t(G[R])) \setminus Z')]$.
However, there could also exist ``spread'' cliques $K$ containing $v$ and using vertices of $((X \cup N^t(G[R])) \setminus Z')$  and $M' \subseteq (V(\T)\setminus \{v\})$.
These cliques may be spread across several vertices of $V(\T)$, and by definition of a root they cannot use vertices in $C(u) \setminus \{u\}$ for any $u \in M' \cup \{v\}$ (according to \autoref{obs:lambda4}).
To take into account the potential conflicts generated by these spread cliques, we perform $t-1$ marking steps  (\autoref{def:step1}),
where informally at each step we guess all possible subsets $M'$, with $|M'| \le t-1$, corresponding to a guessed intersection of a spread clique with previously marked vertices.

The last difficulty is to bound $\mmbs_t(G)$ as a function of $\bedt(G)$ for any graph $G$. We first need to define the notion of blocking set adapted to our problem. Let $\EKtH$ (for \emph{Extended} $K_t$-{\sc Subgraph Hitting}) be the problem where given $(G,\F)$, where $\F$ is a set of subsets of $G$ such that for any $Z \in \F, 1 \le |Z| \le t-1$ and $G[Z]$ is a clique, a solution must intersect all $t$-cliques of $G$ and all $Z \in \F$. A \emph{blocking set} $\B$ of $G$ is a set of subsets of vertices of $G$ such that $\opt(G,\B)>\opt(G)$, where $\opt(G,\B)$ is the minimum size of a solution of the $\EKtH$ problem with input $(G,\B)$, meaning that any set hitting all $t$-cliques of $G$ and all $Z \in \B$ cannot be an optimal solution of $G$ for $\KtH$.
Then, $\mmbs_t(G)$ is the maximum size of an inclusion-wise minimal blocking set of $G$.
The ``max-min'' taste of this definition makes it difficult to handle, but fortunately we will use Property~\ref{prop:mmbs} stating that for any $\beta$, $\mmbs_t(G) \le \beta$ is equivalent to the fact that for any blocking set $\B$ of $G$, there exists $\overline{\B} \subseteq \B$ such that $\overline{\B}$ is still a blocking set of $G$ and $|\overline{\B}| \le \beta$.
We obtain the following upper bound, which requires a considerable amount of technical work.
\begin{theorem}\label{thm:mmbslambda4}
  For any graph $G$ and any integer $t \ge 3$, it holds that $\mmbs_t(G) \le \beta(\bedt(G),t)$, where $\beta(x,t)=
\underbrace{
  {{{{{^{2\vphantom{h}}}^{t2\vphantom{h}}}^{t2\vphantom{h}}}}^{\cdots\vphantom{h}}}^{t\vphantom{h}}
}_{\text{$x$ times}}
$
({{i.e.,}} $\beta(1)=2^t, \beta(2)=2^{t2^t}$, etc.)
\end{theorem}

To explain the difficulty of proving the above theorem, let us sketch how such a bound is obtained for  {\sc Vertex Cover} as a function of treedepth, for example in the proof in \cite{BougeretS18} that, for every graph $G$,
$\mmbs_2(G) \le 2^{\td(G)}$.

Observe first that for the  {\sc Vertex Cover} problem, a blocking set $\B$ is
a set of singletons, which we consider as a subset of vertices,
such that any vertex cover $S$ containing $\B$ is not optimal.
Let us now use Property~\ref{prop:mmbs} and consider a blocking set $\B$ of $G$, and let us show that there exists $\overline{\B} \subseteq \B$ such that $\overline{\B}$  is still a blocking set of $G$ and $|\overline{\B}| \le 2^{\td(G)}$.
Consider a graph $G$ and  a ``root'' $v$ of a treedepth decomposition of $G$, meaning that $\td(G-v) < \td(G)$.
Let us consider the most complex case where there exists an optimal solution using $v$, another avoiding $v$, and $v \notin \B$.
It is not difficult to prove that, as $\B$ is a blocking set of $G$, $\B_1=\B$ is a blocking set of $G_1:=G-v$,
and $\B_2=\B \setminus \{N(v)\}$ is a blocking set of $G_2 := G-(\{v\} \cup N(v))$. Thus, as for any $i \in [2]$, $\td(G_i) < \td(G)$, by induction
we get that there exists $\overline{\B}_i \subseteq \B_i$ such that $\overline{\B}_i$ is a blocking set of $G_i$ and $|\overline{\B}_i| \le 2^{\td(G)-1}$. As $\overline{\B}_1 \cup \overline{\B}_2$ is a blocking set of $G$, we get the desired bound.
The problem when lifting this idea to \textsc{$K_t$-Subgraph Hitting} instead of  {\sc Vertex Cover} is that, when considering an optimal solution $S$ that avoids a root $v$, we do not know which vertex
the solution $S$ will pick in $N(v)$. This is the reason for which we consider the more general version of the problem, namely $\EKtH$, to encode the fact that a solution of $(G,\emptyset)$ avoiding
a root $v$ must be a solution of $(G-v,\pr^t_v(V(G)\setminus \{v\}))$, where
given two disjoint sets  $A,B \subseteq V(G)$, $\pr^t_A(B)= \{K \cap B \mid K \mbox{ is  a $t$-clique in $G[A \cup B]$ and $K \cap A \neq \emptyset$ and $K \cap B \neq \emptyset$}\}$. We also need to define the corresponding generalized notion of blocking set of an instance $(G,\F)$ of $\EKtH$ (\autoref{def:bsKt}), and not only of a graph $G$.
Moreover, we have to keep track of the structure of $\F$, as there is no hope to bound $\mmbs_t(G,\F)$ as a function of $\bedt(G)$ for an arbitrary set $\F$.
Indeed, for example, let  $G_\ell$ be a chain of triangles of length $\ell$, as depicted in \autoref{fig-chains}.
We have $\mmbs_2(G_\ell) \ge \ell$, as if we let $\B$ be the set of top vertices of the $\ell$ triangles, then it can be easily seen that $\B$ is a minimal blocking set of $G_\ell$ with $|\B|=\ell$.
Now, take $\F = E(G_\ell)$, $t=4$, and observe that any solution of  instance $(G_\ell,\F)$ of $\EKtH$ is a vertex cover of $G_\ell$, and thus $\B$ is also a minimal blocking set of size $\ell$ for the input $(G_\ell,\F)$ with $t=4$,
while $\bedt(G_\ell)=0$ as $G_\ell$ is $K_4$-free.


We resolve this problem by proving bounds on $\mmbs_t(G,\F)$ only for a special type of instances that we call \emph{clean}, which are pairs $(G,\F)$ such that  $\opt(G,\F)=\opt(G)$.
The first main difficulty is that, when starting with a blocking set $\B$ of $(G,\F)$, reducing to a graph $G'$ with $\bedt(G') < \bedt(G)$ requires to remove the entire root $\T$ of $G-N^t(G)$, instead of just one vertex as in the treedepth case. As $|V(\T)|$ may be arbitrarily large, we need to prove (see \autoref{lemma:bszooming}) that it is enough to ``zoom'' on a small number of subgraphs (pending components here), allowing us to extract (by induction) a small blocking set only in each of these subgraphs.
The second main difficulty is to ensure that we can reduce via recursion to smaller \emph{clean} instances. Indeed, even if we initially consider a clean instance $(G,\emptyset)$, and even in the favorable case where $\T$ is just one vertex $v$ (as in treedepth), we have to consider the case where there exists an optimal solution of $G$ using $v$, another avoiding $v$, and $v \notin \B$, where $\B$ is fixed blocking set from which we try to extract a small one. However, observe that optimal solutions avoiding $v$ are optimal solutions of $(G-v,\pr^t_v(V(G)\setminus \{v\}))$, and that $\opt(G-v,\pr^t_v(V(G)\setminus \{v\}))=\opt(G)=1+\opt(G-v)$
(the last equality holds since there are optimal solutions taking $v$). Thus, we observe that this situation leads to a \emph{non}-clean instance $(G-v,\pr^t_v(V(G)\setminus \{v\}))$, but ``almost clean'' as $\opt(G-v,\pr^t_v(V(G)\setminus \{v\}))=\opt(G-v)+1$.
We treat these almost clean instances in \autoref{lemma:bsnonclean}, which is the main cause of the huge growth of function $\beta$ in the final bound $\mmbs_t(G) \le \beta(\bedt(G),t)$ given in \autoref{thm:mmbslambda4}. As this bound directly reverberates both in the running time and the size of the kernel (see \autoref{thm:kernel}, where $\delta(\lambda,t)$ is dominated by $\beta(\lambda,t)$ for an instance $(G,X)$ of $\KtHM$ where $\bedt(G-X) \le \lambda$), improving this bound is crucial in order to improve the kernel size. In this direction, we provide in \autoref{lemma:mmbstd} a significantly better upper bound for minimal blocking sets of the $\KtH$ problem as a function of $\td$ instead of $\bedt$. As the proof technique is also different, we believe that this result might be of independent interest.



 Finally, let us mention that in several earlier papers on kernelization using structural parameterizations, it was also crucial to understand the maximum size of an inclusion-minimal set with additional requirements on the solution to a vertex-deletion problem, for which no optimal solution can satisfy all additional requirements; these correspond to variations on the notion of blocking sets.  They were explored for the problems of hitting forbidden connected minors in graphs of bounded treedepth~\cite{JansenP20}, and for hitting cycles in graphs of bounded elimination distance to a forest~\cite{DekkerJ22}, both of which lead to super-exponential bounds in terms of the graph parameter.

\section{Preliminaries for the analysis of blocking sets and the kernel}
\label{sec:prelim}

In this section we provide the definitions and the preliminary results needed in the next two sections, namely in \autoref{sec:blocking} where we provide upper bounds on the sizes of minimal blocking set as a function of the considered parameter(s), and in \autoref{sec:kernel} where we formally describe our kernelization algorithm.

We start with several definitions. A graph $G$ is \textit{biconnected} if it remains connected after the removal of any vertex, and it is \textit{$2$-connected} if it is biconnected and has at least three vertices (hence, an edge is biconnected but not $2$-connected.


Given a graph $G$ and a set $\F$ of subsets of $V(G)$, we say that a set $S \subseteq V(G)$ is a \emph{hitting set of $\F$} if for any $Z \in \F$, $S \cap Z \neq \emptyset$.

\begin{definition}[projection]\label{def:proj}
  Given a graph $G$, two disjoint sets  $A,B \subseteq V(G)$, and an integer $t \ge 3$, we denote by $\pr^t_A(B)= \{K \cap B \mid K$ is  a $t$-clique in $G[A \cup B]$ and $K \cap A \neq \emptyset$ and $K \cap B \neq \emptyset \}$.
\end{definition}

\begin{definition}\label{def:nonKt}
Given a graph $G$ and an integer $t \ge 3$, we say that $N \subseteq V(G)$ is a set of \emph{non-$K_t$-vertices} if for any $t$-clique $K$ of $G$, $K \cap N = \emptyset$, and we denote by $N^t(G)$ the inclusion-wise maximal set of non-$K_t$-vertices of $G$.
\end{definition}

Let us introduce the following variant of $\KtH$.

\begin{center}
\fbox{
\begin{minipage}{13.5cm}
\noindent{{\sc $\EKtH$} (for {\sc Extended $K_t$-{\sc Hitting}})}\\
\noindent\textbf{Input}:~~A pair $(G,\F)$ where $G$ is a graph and $\F$ is a set of subsets of $V(G)$ such that for any $Z \in \F, 1 \le |Z| \le t-1$ and $G[Z]$ is a clique.\\
\textbf{Objective}:~~Find a set $S\subseteq V(G)$ of minimum size such that for any $t$-clique $Z$ of $G$ or any $Z \in \F$, $S \cap Z \neq \emptyset$.
\end{minipage}
}
\end{center}
We denote by $\opt^t(G,\F)$ the optimal value for $(G,\F)$, or simply $\opt(G,\F)$ when $t$ is clear from the context.
We say that $(G,\F)$ is \emph{clean} if $\opt(G,\F)=\opt(G)$.

\begin{definition}[additive partition]
  Given a graph $G$ and a partition $\P$ of $V(G)$, we say that $\P$ is an \emph{additive partition} of $G$ if $\opt(G)=\sum_{C \in \P}\opt(G[C])$ (where $\opt$ is the optimal value when considering the $\KtH$ problem).
\end{definition}

For any set of positive reals $c_1,\dots,c_\ell$, we use $\O_{c_1,\dots,c_\ell}$ to denote the usual the $\O$-notation where the hidden multiplicative constants may depend on $c_1,\dots,c_\ell$, and given two functions $f,g: \mathbb{N} \to \mathbb{N}$ we denote by $f=\O^\star(g)$ that there exists a polynomial $p: \mathbb{N} \to \mathbb{N}$ such that $f(n)=\O(p(n)\cdot g(n))$.
We denote $\mathbb{N^\star}=\mathbb{N} \setminus \{0\}$.

\paragraph*{Definitions and properties related to the parameter $\beds$.\\}

In the following observation, the sets $C(v)$ have been defined in \autoref{def:pending} in \autoref{sec:overview}.

\begin{observation}\label{obs:lambda4part}
Let $G$ be a graph and $\T$ be a root of $G$.
Then $V(G)$ is partitioned into $\{C(v) \mid v \in V(\T)\}$.
\end{observation}
\begin{proof}
Let $u \in V(G)$ and let  $P$ be a shortest path from $u$ to $V(\T)$. This implies that  $P=(u_1,\dots,u_{\ell})$ for some $\ell \in \mathbb{N}^\star$ such that $u_1 = u$, $u_\ell \in V(\T)$, and $u_{\ell'} \notin V(\T)$ for any $\ell' < \ell$. Then, we can prove by induction starting from $u_\ell$ that for any $u_i \in P$, $u_i \in C(u_\ell)$.
Thus, $u \in C(u_\ell)$, and moreover, by \autoref{obs:lambda4}, $u \notin C(v)$ for any $v \neq u_{\ell}$.
\end{proof}

\begin{observation}\label{obs:cliquelambda4}
Let $G$ be a graph, $N$ be a set of non-$K_t$-vertices of $G$, and $\T$ be a root of $G-N$.
For any $t$-clique $K$ of $G$, there exists $v \in V(\T)$ such that
$K \subseteq C(v)$.
\end{observation}
\begin{proof}
Let $G' = G - N$ and let $K$ be a $t$-clique of $G$. By definition of $N$, we know that $K \subseteq V(G')$. According to \autoref{obs:lambda4part}, $V(G')$ is partitioned into $\{C(v) \mid v \in V(\T)\}$, so let $v \in V(\T)$ such that $K \cap C(v) \neq \emptyset$. As by definition of a root $G'[V(\T)]$ is $K_t$-free, we cannot have $K \subseteq V(\T)$. Thus, there exists $u \in V(\T)$ such that $K \cap (C(u) \setminus \{u\}) \neq \emptyset$. If $u \neq v$, then there would be an edge between $(C(u) \setminus \{u\})$ and $C(v)$, contradicting \autoref{obs:lambda4}. Thus, we get that $K \cap (C(v) \setminus \{v\}) \neq \emptyset$. By \autoref{obs:lambda4} again, this implies that $K \subseteq C(v)$.\end{proof}

\begin{definition}\label{def:pendingpart}
Let $G$ be a graph and let $\T$ be a root of $G$.
For any partition $\P$ of $V(\T)$, we define the corresponding \emph{pending partition of $\P$} as $\P^p=\{C(Z), Z  \in \P\}$.
\end{definition}

\begin{lemma}\label{obs:pendingadditive}
Let $G$ be a graph, $N$ be a set of non-$K_t$-vertices of $G$, and $\T$ be a root of $G-N$.
For any partition $\P$ of $V(\T)$, $\P^p \cup \{N\}$ is an additive partition of $G$.
\end{lemma}
\begin{proof}
For any $v \in V(\T)$, let $S_v$ be an optimal solution of $G[C(v)]$, and let $S = \bigcup_{v \in V(\T)}S_v$.
According to \autoref{obs:cliquelambda4}, $S$ is a solution of $G$, leading
to $\opt(G) = \sum_{v \in V(\T)}\opt(C(v))=\sum_{Z \in \P}\sum_{v \in Z}\opt(C(v)) \le \sum_{Z \in \P}\opt(C(Z))=\sum_{Z \in \P}\opt(C(Z))+\opt(G[N])$.
As we also have $\opt(G) \ge \sum_{Z \in \P}\opt(C(Z))+\opt(G[N])$, we get the desired equality.
\end{proof}

It is worth pointing out that the following lemma applies to any graph $H$.

\begin{lemma}\label{obslambdamonotonesubgraph}
Let~$H$ be a fixed graph, not necessarily complete or connected. If~$G'$ is an induced subgraph of~$G$, then~$\beds(G') \leq \beds(G)$.
\end{lemma}
\begin{proof}
We prove the statement by induction on~$|V(G')| + |V(G)|$. Consider the defining case for the value of~$\beds(G')$ in \autoref{def:pending}.
\begin{enumerate}
    \item If~$V(G') = \emptyset$, then~$\beds(G') = 0 \leq \beds(G)$.
    \item If~$G'$ has a vertex~$v$ that is not in any copy of~$H$, then~$\beds(G') = \beds(G'-v)$. Since~$G' - v$ is an induced subgraph of~$G$ on strictly fewer vertices, by induction we find that~$\beds(G'-v) \leq \beds(G)$.
    \item If~$G'$ has multiple connected components~$C_1, \ldots, C_m$, then each is an induced subgraph of~$G$ with strictly fewer vertices. By induction, each component~$C_i$ therefore has~$\beds(G[C_i]) \leq \beds(G)$. Since the value for~$G'$ is the maximum over its components, it is at most~$\beds(G)$.
    \item If none of the cases above hold, then~$G'$ is connected and all its vertices belong to a copy of~$H$. If~$G$ has a vertex~$v$ that does not belong to any copy of~$H$, then~$v \notin G'$ and therefore~$G'$ is an induced subgraph of~$G-v$, from which the claim follows by induction. Similarly, if~$G$ is disconnected then~$G'$ is an induced subgraph of a connected component of~$G$ and again we may apply induction. In the remainder, we may assume that both~$G'$ and~$G'$ are connected and all their vertices occur in some copy of~$H$. Let~$T$ be a root of~$G$ with~$\beds(G - T) < \beds(G)$. Define~$T' := T \cap V(G')$. We consider three cases.
    \begin{itemize}
        \item If~$T' = \emptyset$, then~$G'$ is an induced subgraph of~$G - T$. Hence by induction, we have~$\beds(G') \leq \beds(G - T) < \beds(G)$. In the remainder, we assume~$T' \neq \emptyset$.
        \item If~$G'[T']$ is connected, then we claim~$T'$ is a root of~$G'$. For this we have to show that each component of~$G'-T'$ is adjacent to exactly one vertex of~$T'$. Consider the connected components of~$G' - T'$. We can think of these components as having been obtained from~$G$ by removing~$T$ and then removing the vertices of~$V(G) \setminus V(G')$. Hence each component is an induced subgraph of a component of~$G-T$. Since the latter components have at most one neighbor in~$T$, the former also have at most one neighbor in~$T$. Since~$G'$ is connected by assumption, each component of~$G'-T'$ has exactly one neighbor in~$T'$.

        So~$T'$ is a root and we find~$\beds(G') \leq 1 + \beds(G' - T') \leq 1 + \beds(G - T)$, where the last step follows from induction using the fact that~$G' - T'$ is an induced subgraph of~$G-T$.
        \item If~$G'[T']$ is disconnected, we derive a contradiction as follows. Consider two distinct vertices~$u,v \in T'$ which are in different connected components~$C_u, C_v$ of~$G[T']$. Since~$G'$ is connected by assumption, there is a path from~$u$ to~$v$ in~$G'$. Consider a minimal path~$P$ in~$G'$ which starts in~$C_u$, ends in~$C_v$, and has all its interior vertices outside~$G'[T']$. Then~$P$ has at least one interior vertex~$w$. The interior vertices of~$P$ belong to a single connected component of~$G' - T'$, and therefore also belong to a single connected component of~$G-T$. But this component is adjacent to at least two distinct vertices of~$T' \subseteq T$, one in~$C_v$ and one in~$C_u$. This contradicts the fact that~$T$ is a root of~$G$.
    \end{itemize}
\end{enumerate}
As the cases are exhaustive, this concludes the proof.
\end{proof}




\begin{definition}[$\bedt$-root of $G$]\label{definition:lambdaroot}
  Let $G$ be a graph with $\bedt(G) > 0$.
  We say  $\T$ is a \emph{$\bedt$-root of $G$} if $\T$ is a root of $G$  and $\bedt(G - V(\T)) = \bedt(G)-1$.
\end{definition}

\begin{observation}\label{lemma:bedtcv}
Let $G$ be a graph such that $N^t(G)=\emptyset$, $\T$ be a \emph{$\bedt$-root of $G$}. Then, for any $v \in V(\T)$, $\bedt(G[C(v)]-v) < \bedt(G)$.
\end{observation}
\begin{proof}
As for any $v \in V(\T)$, $G[C(v)]-v$ is a subgraph of $G-V(\T)$,
\autoref{obslambdamonotonesubgraph} implies the claimed inequality.
\end{proof}


\section{Blocking sets and their properties}
\label{sec:blocking}

In this section we provide upper bounds on the size of minimal blocking sets as a function of the parameters that we consider. In \autoref{sec:prelminaries-blocking-sets} we start by defining (minimal) blocking sets and prove several useful properties. In \autoref{sec:bounding-mmbs} we prove the main result of this section, namely \autoref{thm:mmbslambda4}, which provides an upper bound on minimal blocking sets as a function of the parameter  $\bedt$. In \autoref{ap:improved-mmbs-treedepth} we use a completely different technique to obtain upper bounds on the size of minimal blocking sets in graphs of bounded treedepth, which substantially improve those given in \autoref{thm:mmbslambda4} in graphs of bounded  $\bedt(G),t)$.

\subsection{Preliminaries}
\label{sec:prelminaries-blocking-sets}

We start with the definitions of (minimal) blocking set that we need for our purposes.

\begin{definition}\label{def:bsKt}
  Given an instance $(G,\F)$ of $\EKtH$, a \emph{blocking set} $\B$ of $(G,\F)$ is a set of subsets of $V(G)$ such that

  \begin{itemize}
  \item for any $B \in \B$, $1  \le |B| \le t-1$, $G[B]$ is a clique, and
  \item $\opt(G,\F \cup \B) > \opt(G,\F)$.
  \end{itemize}

  We define $\mmbs_t(G,\F)=\max\{|\B|$ such that $\B$ is a blocking set of $(G,\F)$, and for any $B \in \B$, $\B \setminus B$ is not a blocking set of $(G,\F)\}$, which can be rephrased as the maximum size of an inclusion-wise minimal blocking set of $(G,\F)$.
  Given a graph $G$, we denote by $\mmbs_t(G)=\max\{\mmbs_t(G,\F) \mbox{ such that $(G,\F)$ is a clean instance} \}$.
  We denote by $\mmbs_t^{\star}(G)$ the maximum value, over any subgraph  $G'$ of $G$,  of $\mmbs_t(G')$.
Finally, given an instance $G$ of $\KtH$, a \emph{blocking set  of $G$}  is a blocking set of $(G,\emptyset)$.
\end{definition}

Note that the above definition also defines blocking sets for $\KtH$ by setting $\F = \emptyset$, and that $\mmbs_t^{\star}(G)$ may be strictly greater than $\mmbs_t(G)$. Informally, $\mmbs_t(G)$ is the function that we are interested in,
and $\mmbs_t(G,\F)$ and $\mmbs_t^{\star}(G)$ are defined for technical reasons.
Note also that, even if the notion of blocking set is defined for any (not only clean) instance $(G,\F)$ of $\EKtH$, $\mmbs_t(G)$ is defined by restricting to clean instances. The reason is that our objective (see \autoref{thm:mmbslambda4})
is to prove that there exists a function $\beta$ such that $\mmbs_t(G) \le \beta(\bedt(G),t)$, and as discussed in \autoref{sec:overview} (see \autoref{fig-chains}),
such a function does not exists is $\mmbs_t$ were defined without restricting to clean instances.

The following property (already used in previous results involving blocking sets, for example in~\cite{BougeretJS22}) states that small $\mmbs_t$ is equivalent to the property of extracting a small blocking set from a given large one.

\begin{property}\label{prop:mmbs}
 Let $(G,\F)$ be an instance of $\EKtH$ and $\beta \in \mathbb{N}$.
  Then, the two following items are equivalent:
  \begin{itemize}
  \item For any $\B$ blocking set of $(G,\F)$, there exists $\overline{\B} \subseteq \B$ such that $\overline{\B}$ is a blocking set of $(G,\F)$ and $|\overline{\B}| \le \beta$.
  \item $\mmbs_t(G,\F) \le \beta$.
  \end{itemize}
\end{property}
\begin{proof}
  To prove the forward implication,
  let $\B$ be an inclusion-wise minimal blocking set of $(G,\F)$. If by contradiction $|\B| > \beta$, there would exist $\overline{\B} \subseteq \B$ such that $\overline{\B}$ is a blocking set of $(G,\F)$ and $|\overline{\B}| \le \beta < |\B|$, contradicting the inclusion-wise minimality of $\B$.

  TO prove the backward implication,
  Let $\B$ be a blocking set of $(G,\F)$. We simply define $\overline{\B} \subseteq \B$ as an inclusion-wise minimal blocking set of $(G,\F)$.
\end{proof}

\begin{definition}
  Given an instance $(G,\F)$ of $\EKtH$, $v \in V(G)$, and denoting $\optsol(G,\F)$ the set of optimal solutions of $(G,\F)$, we define the \emph{type of $v$} as $t^{(G,\F)}(v)=$
  \begin{equation*}
    \begin{cases}
      \{v,\bar{v}\} & \mbox{if $\exists$   $S \in\optsol(G,\F)$ such that $v \in S$, and $\exists$  $S'\in\optsol(G,\F)$ with $v \notin S'$}, \\
      \{v\} & \mbox{if $\exists$  $S\in\optsol(G,\F)$  such that $v \in S$, and $\nexists$  $S'\in\optsol(G,\F)$ with $v \notin S'$}, \\
      \{\bar{v}\} & \mbox{if $\nexists$  $S\in\optsol(G,\F)$ such that $v \in S$, and $\exists$  $S'\in\optsol(G,\F)$ with $v \notin S'$}.
    \end{cases}
\end{equation*}
\end{definition}

\noindent Given a graph $G$, a set $\C$ of subsets of $V(G)$, and $V' \subseteq V(G)$, we denote

\begin{itemize}
\item $\C \cap V' = \{B \in \C \mid B \subseteq V'\}$,
\item $\C_{|V'} = \{B \cap V'\mid B \in \C\}$,
\item $\C - V' = \C \cap (V(G) \setminus V')$, and
\item $V(\C)=\bigcup_{B \in \C}B$.
\end{itemize}

Note that given a set $\C$ of subsets of $V(G)$, and a partition $\P=\{P_i \mid i \in [|\P|]\}$ of $V(G)$, we may not have $\C = \bigcup_{i \in [|\P|]}(\C \cap P_i)$
as any $Z \in \C$ such that $Z \cap P_i \neq \emptyset$ and $Z \cap P_j \neq \emptyset$ is not contained in $\bigcup_{i \in [|\P|]}(\C \cap P_i)$.

\subsection{Graphs with bounded $\bedt$ have bounded minimal blocking sets}
\label{sec:bounding-mmbs}


We start with three technical lemma whose objective is indicated beside their statement.

\begin{lemma}[blocking set vs blocking set of subinstances]\label{lemma:bsbigbssmall}

  Let  $(G,\F)$ be an instance of $\EKtH$ and let $\B$ be a blocking set of $(G,\F)$.
  \begin{enumerate}[(i)]
  \item \label{lemma:bsbigbssmall:i} Let $X^{\star} \subseteq V(G)$ such that there exists an optimal solution $S^{\star}$ of $(G,\F)$ with $X^{\star} \subseteq S^\star$. Then
    \begin{enumerate}
    \item $\opt(G,\F)=|X^{\star}|+\opt(G- X^{\star}, \F - X^{\star})$, and
    \item $\B - X^{\star}$ is a blocking set of $(G- X^{\star}, \F - X^{\star})$.
    \end{enumerate}
  \item \label{lemma:bsbigbssmall:ii} Let $v \in V(G)$ such that $\bar{v} \in t^{(G,\F)}$ and $\{v\} \notin \B$. Then
 \begin{enumerate}
    \item $\opt(G,\F)=\opt(G-v, \F_{|V(G)\setminus v} \cup \pr^t_v(V(G)\setminus v))$, and
    \item $\B_{V(G)\setminus v}$ is a blocking set of $(G-v, \F_{|V(G)\setminus v} \cup \pr^t_v(V(G)\setminus v))$.
 \end{enumerate}
  \end{enumerate}

Let  $(G,\F)$ be an instance of $\EKtH$.
 \begin{enumerate}
   \item \label{lemma:bsbigbssmall:ibis}For any $X \subseteq V(G)$, any blocking set $\overline{\B}$ of $(G-X, \F-X)$ and any optimal solution $S^{\star}$ of $(G,\F)$
      such that $X \subseteq S^{\star}$, there exists $B \in \overline{\B}$ such that $S^{\star} \cap B=\emptyset$, and
     \item \label{lemma:bsbigbssmall:iibis}For any $v \in V(G)$ such that $\{v\} \notin \F$, blocking set $\overline{\B}$ of $(G-v, \F_{|V(G)\setminus v} \cup \pr^t_v(V(G)\setminus v))$, and any optimal solution $S^{\star}$ of $(G,\F)$  such that $v \notin S^{\star}$, there exists $B \in \overline{\B}$ such that $S^{\star} \cap B=\emptyset$.
 \end{enumerate}
\end{lemma}
\begin{proof}
  Proof of  \autoref{lemma:bsbigbssmall:i}.
  By the existence of $S^{\star}$ we get $\opt(G,\F)=|X^{\star}|+|S \setminus X^{\star}|$. As $S \setminus X^{\star}$ has to intersect any $t$-clique of $G-X^{\star}$ and any set of $\F - X^{\star}$,
  it is a solution of $(G- X^{\star}, \F - X^{\star})$, implying $|S \setminus X^{\star}| \ge \opt(G- X^{\star}, \F - X^{\star})$ and $\opt(G,\F) \ge |X^{\star}|+  \opt(G- X^{\star}, \F - X^{\star})$.
  As the other inequality is true by definition, we get the first property of \autoref{lemma:bsbigbssmall:i}.
  Assume now by contradiction that $\B - X^{\star}$ is a not a blocking set of $(G- X^{\star}, \F - X^{\star})$, implying the existence of an optimal solution $S^{\star}$ of
  $(G- X^{\star}, \F - X^{\star})$ intersecting any $B \in \B - X^{\star}$. Let $S=X^{\star} \cup S^{\star}$. We get that $S$ is a solution of $(G,\F)$ intersecting $\B$, and
  $|S|=|X^{\star}|+\opt(G- X^{\star}, \F - X^{\star})=\opt(G,\F)$, a contradiction.

  In the remaining of the proof, we denote by $I=(G-v, \F_{|V(G)\setminus v} \cup  \pr^t_v(V(G)\setminus v))$.

  Proof of \autoref{lemma:bsbigbssmall:ii}.
  As $\bar{v} \in t^{(G,\F)}$, there exists an optimal $S^{\star}$ of $(G,\F)$ with $v \notin S^{\star}$. This implies in particular that there is no set $Z \in \F$ with $Z=\{v\}$,
  and thus that $\F_{|V(G)\setminus v}$ does not contain $\emptyset$ (as required in the definition of $\EKtH$).
  We get $\opt(G,\F)=|S^{\star} \setminus v|$. As $S^{\star} \setminus v$ has to intersect any $t$-clique of $G-v$, any set of $\F_{|V(G)\setminus v}$ (as $S^{\star}$ intersects any $Z \in \F$
  and $v \notin S^{\star}$), and any set of $\pr^t_v(V(G)\setminus v)$ (as $S^{\star}$ intersects any $t$-clique of $G$ and $v \notin S^{\star}$), we get that $S^{\star}$ is a solution of
  $I$, leading to $\opt(G) \ge \opt(I)$.
   As the other inequality is true by definition, we get the first property of \autoref{lemma:bsbigbssmall:ii}.
   Assume now by contradiction that $\B_{V(G)\setminus v}$ is a not a blocking set of $I$, implying the existence of an optimal solution $S^{\star}$ of
   $I$ intersecting any $B \in \B_{V(G)\setminus v}$. As $\{v\} \notin \B$, $S^{\star}$ intersects all $B \in \B$,
   and as there is no set $Z \in \F$ with $Z=\{v\}$, $S^{\star}$ also intersects all $Z \in \F$. Finally, $S^{\star}$ intersects all $t$-cliques $K \subseteq G-v$ by definition,
   and any $t$-clique $K$ containing $v$, as $K \setminus v \in \pr^t_v(V(G)\setminus v)$. Thus, $S^{\star}$ is a solution of $(G,\F)$ intersecting $\B$, and $|S^{\star}|=\opt(I)=\opt(G,\F)$, a contradiction.

  Proof of \autoref{lemma:bsbigbssmall:ibis}.
  We have $\opt(G,\F)=|X|+|S^{\star} \setminus X|$, and $S^{\star} \setminus X$ is a solution of $(G-X,\F-X)$, implying $|S^{\star} \setminus X| \ge \opt(G-X,\F-X)$ and $\opt(G,\F) \ge |X|+\opt(G-X,\F-X)$.
  As taking an optimal solution $S^{\star}$ of $(G-X,\F-X)$ and defining $X \cup S^{\star}$ is a solution of $(G,\F)$, we get $\opt(G,\F) \le |X|+\opt(G-X,\F-X)$, and thus $\opt(G,\F)=|X|+\opt(G-X,\F-X)$
  and $|S^{\star} \setminus X|=\opt(G-X,\F-X)$. As $\overline{\B}$ is a blocking set of $(G-X,\F-X)$,
there exists $B \in \overline{\B}$ such that $(S^{\star} \setminus X) \cap B = \emptyset$. As $B \subseteq V(G)\setminus X$, this implies $S^{\star} \cap B = \emptyset$, and thus $\overline{\B}$ is indeed a blocking set.

  Proof of \autoref{lemma:bsbigbssmall:iibis}.
  We have $\opt(G,\F)=|S^{\star}|$, and as $v \notin S^{\star}$, $S^{\star}$ is a solution of $I$, implying $|S^{\star}| \ge \opt(I)$ and $\opt(G,\F) \ge \opt(I)$.
  As an optimal solution $S^{\star}$ of $I$ is also a solution of $(G,\F)$ (as $\{v\} \notin \F$), we get $\opt(G,\F) \le \opt(I)$, so $\opt(G,\F)=\opt(I)$,
  and thus $|S^{\star}|=\opt(I)$. As $\overline{\B}$ is a blocking set of $I$, there exists $B \in \overline{\B}$ such that $S^{\star} \cap B = \emptyset$.
\end{proof}

\begin{lemma}[conditions for clean instances]\label{lemma:bsclean}
  Let $(G,\F)$ be a clean instance of $\EKtH$.
\begin{enumerate}[(i)]
  \item \label{lemma:bsclean:i} Let $\P$ be an additive partition of $G$. Then, for any $C \in \P$, $(G[C],\F \cap C)$ is a clean instance.
  \item \label{lemma:bsclean:ii} Let $v \in V(G)$ such that $v \in t^{(G,\F)}$. Then, $(G-v,\F-v)$ is a clean instance.
\end{enumerate}
\end{lemma}
\begin{proof}
  Proof of \autoref{lemma:bsclean:i}.
  Suppose by contradiction that there exists $C_0 \in \P$, such that $\opt(G[C_0],\F \cap C_0)>\opt(G[C_0])$.
  Let $S^{\star}$ be an optimal solution of $(G,\F)$. For any $C \in \P$, $S^{\star} \cap C$ must a a solution of $(G[C],\F \cap C)$, implying
  $|S^{\star} \cap C| \ge \opt(G[C],\F \cap C) \ge \opt(G[C])$ for $C \neq C_0$, and $|S \cap C_0|>\opt(G[C_0])$. Thus, we get $|S^{\star}|=\sum_{C \in \P}|S \cap C|>\sum_{C \in \P}\opt(G[C])$.
  As $(G,\F)$ is clean and $\P$ is additive, we get $\opt(G,\F)=\opt(G)=\sum_{C \in \P}\opt(G[C])$, implying $|S^{\star}|>\opt(G)$, a contradiction.

  Proof of \autoref{lemma:bsclean:ii}.
  As $v \in t^{(G,\F)}$, by \autoref{lemma:bsbigbssmall:i} of \autoref{lemma:bsbigbssmall}  with $X^{\star}=\{v\}$, we get $\opt(G,\F)=1+\opt(G-v,\F-v)$.
  Moreover, we have $\opt(G) \le 1+\opt(G-v)$. As $(G,\F)$ is clean, we get $\opt(G-v,\F-v) \le \opt(G-v)$, leading to the claim as the other inequality is always true.
\end{proof}

\begin{lemma}[extracting a blocking set in a non-clean instance]\label{lemma:bsnonclean}
  Let $(G,\F)$ be an instance of $\EKtH$ such that $\opt(G,\F)=\opt(G)+1$ and $\B$ be a blocking set of $(G,\F)$.
  There exists $\overline{\B} \subseteq \B$ such that $\overline{\B}$ is a blocking set of $(G,\F)$ and $|\overline{\B}| \le 2^{(t-1)\mmbs_t(G)} \mmbs_t^{\star}(G)$.
\end{lemma}
\begin{proof}
To prove the lemma, we will prove that there exists $c \le 2^{(t-1)\mmbs_t(G)}$, $\B_i \subseteq \B$, $G_i \subseteq G$, $\F_i \subseteq \F$ for $i \in [c]$, such that
\begin{enumerate}[(i)]
  \item \label{lemma:bsnonclean:i} for any $i \in [c]$, $(G_i,\F_i)$ is a clean instance,
  \item \label{lemma:bsnonclean:ii} for any $i \in [c]$, $\B_i$ is a blocking set of $(G_i,\F_i)$, with $|\B_i| \le \mmbs_t(G_i)$, and
  \item \label{lemma:bsnonclean:iii} $\bigcup_{i \in [c]}\B_i$ is a blocking set of $(G,\F)$.
\end{enumerate}
Then, the lemma will immediately follow by setting $\overline{\B}=\bigcup_{i \in [c]}B_i$.

  As $\opt(G,\F)=\opt(G)+1$, we get that $\F$ is a blocking set of $(G,\emptyset)$, and thus by definition there exists $\overline{\F} \subseteq \F$ (recall that $\F$ is a set of subset of vertices, and thus each $Z \in \F$ is either entirely chosen and present in $\overline{\F}$, or not in $\overline{\F}$) such that
  $|\overline{\F}| \le \mmbs_t(G)$ and $\overline{\F}$ is a blocking set of $(G,\emptyset)$.
  Let $\F^{\star} = \{V(\overline{\F}) \cap S \mid S\mbox{ is an optimal solution of $(G,\F)$}\}$.
  Observe that for any $X^{\star} \in \F^{\star}$, $X^{\star}$ is a hitting set of $\overline{\F}$, and that  $|\F^{\star}| \le 2^{(t-1)\mmbs_t(G)}$ as any $B \in \overline{\F}$ has size at most $t-1$.

  Let us now establish the following property $p_1$: for any $X^{\star} \in \F^{\star}, \opt(G,\F)=\opt(G- X^{\star})+|X^{\star}|$.
  Let us start with $\opt(G,\F) \ge \opt(G- X^{\star})+|X^{\star}|$. Let $S$ be an optimal solution of $(G,\F)$ such that $X^{\star} = S \cap V(\overline{\F})$.
  We have $\opt(G,\F)=|S|=|X^{\star}|+|S\setminus X^{\star}|$. As $S \setminus X^{\star}$ is a solution of $(G-X^{\star})$, we get $|S\setminus X^{\star}| \ge \opt(G-X^{\star})$, and the desired inequality follows.
  Let us now prove that $\opt(G,\F) \le \opt(G- X^{\star})+|X^{\star}|$.
  Suppose by contradiction that $\opt(G,\F) \ge \opt(G- X^{\star})+|X^{\star}|+1$.
  As $\opt(G,\F)=\opt(G)+1$, we deduce $\opt(G) \ge \opt(G- X^{\star})+|X^{\star}|$ and thus $\opt(G) = \opt(G- X^{\star})+|X^{\star}|$.
  Let $S_1$ be an optimal solution of $(G-X^{\star})$ and $S = X^{\star} \cup S_1$. We get that $|S|=\opt(G)$ and $X^{\star} \subseteq S$, and thus $S$ is
  an optimal solution of $(G,\overline{\F})$, as $X^{\star}$ is a hitting set of $\overline{\F}$. This implies $\opt(G,\overline{\F})=\opt(G)$, contradicting that
  $\overline{\F}$ is a blocking set of $(G,\emptyset)$.

  Let us now establish the following property $p_2$: for any $X^{\star} \in \F^{\star}$, $(G-X^{\star},\F-X^{\star})$ is a clean instance, and there exists $\overline{\B}_{X^{\star}} \subseteq \B$ such that
  $\overline{\B}_{X^{\star}}$ is a blocking set of $(G-X^{\star},\F-X^{\star})$ with $|\overline{\B}_{X^{\star}}| \le \mmbs_t(G-X^{\star})$.
  Let $X^{\star} \in \F^{\star}$. By \autoref{lemma:bsbigbssmall:i} of \autoref{lemma:bsbigbssmall}, $\opt(G,\F)=|X^{\star}|+\opt(G-X^{\star},\F-X^{\star})$, and $\B_{X^{\star}}=\B-X^{\star}$ is a blocking set
  of $(G-X^{\star},\F-X^{\star})$.
  According to property $p_1$, $\opt(G,\F)=\opt(G- X^{\star})+|X^{\star}|$, leading to $\opt(G-X^{\star},\F-X^{\star})=\opt(G-X^{\star})$, implying that $(G-X^{\star},\F-X^{\star})$ is clean.
  By definition, there exists $\overline{\B}_{X^{\star}} \subseteq \B_{X^{\star}}$ as claimed.

  Let us now prove the lemma. We define $c= |\F^{\star}|$ and $\{B_i\}=\{\overline{\B}_{X^{\star}} \mid X^{\star} \in \F^{\star}\}$. Property~$p_2$ above corresponds to \autoref{lemma:bsnonclean:i} and \autoref{lemma:bsnonclean:ii} of the lemma.
  It remains to prove that $\overline{\B}=\bigcup_{X^{\star} \in \F^{\star}}\overline{\B}_{X^{\star}}$ is a blocking set of $(G,\F)$.
  Let $S$ be an optimal solution of $(G,\F)$, and $X^{\star}=S \cap V(\overline{\F})$. \autoref{lemma:bsbigbssmall} (\autoref{lemma:bsbigbssmall:i}) implies that $\opt(G,\F)=|X^{\star}|+\opt(G-X^{\star},\F-X^{\star})$,
  and as $\opt(G,\F)=|X^{\star}|+|S \setminus X^{\star}|$, we deduce that $|S \setminus X^{\star}|=\opt(G-X^{\star},\F-X^{\star})$. Since $\overline{\B}_{X^{\star}}$ is a blocking set of $(G-X^{\star},\F-X^{\star})$,
  there exists $B \in \overline{\B}_{X^{\star}}$ such that $(S \setminus X^{\star}) \cap B = \emptyset$. As $B \subseteq V(G)-X^{\star}$, it follows that $S \cap B = \emptyset$, and thus $\overline{\B}$ is indeed a blocking set.
\end{proof}

In the next two lemmas we study how blocking sets behave when we remove one vertex from a graph.
\begin{lemma}\label{lemma:bswithoutv}
  Let $(G,\F)$ be a clean instance of $\EKtH$, $v \in V(G)$, and $\B$ be a blocking set of $(G,\F)$ such that $\{v\} \notin \B$.
 %
%
There exists a blocking set $\overline{\B} \subseteq \B$ of $(G,\F)$ with $|\overline{\B}| \le (2^{(t-1)\mmbs_t(G-v)}+1) \cdot  \mmbs_t^{\star}(G-v)$.
\end{lemma}
\begin{proof}
Let $I=(G-v,\F_{V(G)\setminus v} \cup \pr^t_v(V(G)\setminus v))$.
We distinguish three cases depending on $t^{(G,\F)}(v)$.

Case 1: Suppose $t^{(G,\F)}(v)=\{v\}$.
According to \autoref{lemma:bsbigbssmall} (\autoref{lemma:bsbigbssmall:i}) we get that $\opt(G,\F)=1+\opt(G-v,\F-v)$ and $\B-v$ is a blocking set of
$(G-v,\F-v)$.
According to \autoref{lemma:bsclean} (\autoref{lemma:bsclean:ii}), $(G-v,\F-v)$ is a clean instance, and thus there exists $\overline{\B} \subseteq \B-v$
such that $\overline{\B}$ is a blocking set of $(G-v,\F-v)$ and $|\overline{\B}| \le \mmbs_t(G-v)$.
According to \autoref{lemma:bsbigbssmall} (\autoref{lemma:bsbigbssmall:ibis}), $\overline{\B}$ is a blocking set of $(G,\F)$.

Case 2: Suppose $t^{(G,\F)}(v)=\{\bar{v}\}$. Note first that this implies $\{v\} \notin \F$.

Case 2a) Suppose $I$ is a clean instance. By \autoref{lemma:bsbigbssmall} (\autoref{lemma:bsbigbssmall:ii}), $\B_{|V(G)\setminus v}$ is a blocking set of $I$. As $I$ is clean, and by definition of $\mmbs_t$, there exists $\overline{\B}_0 \subseteq \B_{|V(G)\setminus v}$ such that $\overline{\B}_0$ is a blocking set of $I$ and $|\overline{\B}_0| \le \mmbs_t(G-v)$.
Note that we cannot define $\overline{\B}=\overline{\B}_0$, as we have to output a set $\overline{\B} \subseteq \B$, but $\overline{\B}_0$ may not be a subset of $\B$ (as a set $B \in \B$ containing $v$ becomes $B\setminus v$ in $\B_{|V(G)\setminus v}$). Thus, we rather define  $\overline{\B}=\add(\overline{\B}_0,v,\B)$, where $\add(\overline{\B}_0,v,\B)=\bigcup_{B \in \overline{\B}_0}f(B)$, where $f(B)$ is any $B' \in \B$ such that $B' \setminus \{v\} = B$.
Let us prove that $\overline{\B}$ is a blocking set of $(G,\F)$.
Let $S$ be an optimal solution of $(G,\F)$. As $t^{(G,\F)}(v)=\{\bar{v}\}$, $v \notin S$.
As $v \notin \F$, according to \autoref{lemma:bsbigbssmall} (\autoref{lemma:bsbigbssmall:iibis}), there exists $B \in \overline{\B}_0$ such that $S \cap B = \emptyset$.
As $S$ intersects all the $K_t$'s containing $v$, we know that $S \cap \tilde{B} \neq \emptyset$ for any
$\tilde{B} \in  \pr^t_v(V(G)\setminus v)$. This implies that $B \notin  \pr^t_v(V(G)\setminus v)$, and thus either $B'=B$ or $B'=B\cup v$ is in $\overline{\B}$. As $v \notin S$, we get $S \cap B' = \emptyset$, implying that $\overline{\B}$ is indeed a blocking set. Moreover, $|\overline{\B}|=|\overline{\B}_0|\le \mmbs_t(G-v) \le \mmbs_t^{\star}(G-v)$ and thus is smaller that the required size.

Case 2b) Suppose $I$ is not a clean instance.
As $\opt(G)\le 1+\opt(G-v)$ and $\opt(G)=\opt(G,\F)=\opt(I)$, the last equality coming from \autoref{lemma:bsbigbssmall} (\autoref{lemma:bsbigbssmall:ii}), we get that $\opt(I) \le 1+\opt(G-v)$, and thus $\opt(I)=1+\opt(G-v)$ as $I$ is not clean.
Thus, we can apply \autoref{lemma:bsnonclean} on $I$ and $\B_{|V(G)\setminus v}$, and thus there exists
$\overline{\B}_0 \subseteq \B$ such that $\overline{\B}_0$ is a blocking set of $I$, and $|\overline{\B}_0| \le 2^{(t-1)\mmbs_t(G-v)} \mmbs_t^{\star}(G-v)$.
In this case, we define $\overline{\B}=\add(\overline{\B}_0,v,\B)$. For the same arguments than in case 2a) we get that $\overline{\B}$ is a blocking set, whose size is upper-bounded by the claimed bound.

Case 3: Suppose $t^{(G,\F)}(v)=\{v,\bar{v}\}$. Note first that this implies $\{v\} \notin \F$.
In this case we use both small blocking sets defined in cases 1) and 2b).
As in case 1), according to \autoref{lemma:bsbigbssmall} (\autoref{lemma:bsbigbssmall:i}) we get that $\opt(G,\F)=1+\opt(G-v,\F-v)$ and $\B-v$ is a blocking set of
$(G-v,\F-v)$.
According to \autoref{lemma:bsclean} (\autoref{lemma:bsclean:ii}), $(G-v,\F-v)$ is a clean instance, and thus there exists $\overline{\B}_1 \subseteq \B-v$
such that $\overline{\B}_1$ is a blocking set of $(G-v,\F-v)$ and $|\overline{\B}_1| \le \mmbs_t(G-v)$.

According to \autoref{lemma:bsbigbssmall} (\autoref{lemma:bsbigbssmall:ii}),  we get that $\opt(G,\F)=\opt(I)$ and $\B_{|V(G)\setminus v}$ is a blocking set of $I$.
As $v\in t^{(G,\F)}(v)$ and $(G,\F)$ is clean we get $\opt(G)=\opt(G,\F)=1+\opt(G-v,\F-v) \ge 1+\opt(G-v)$, leading to $\opt(G)= 1+\opt(G-v)$.
On the other hand, we also have $\opt(G,\F)=\opt(I)$, and we get that $\opt(I) = 1+\opt(G-v)$.
Thus, we can apply \autoref{lemma:bsnonclean} on $I$ and $\B_{|V(G)\setminus v}$, and thus there exists
$\overline{\B}_0 \subseteq \B$ such that $\overline{\B}_0$ is a blocking set of $I$, and $|\overline{\B}_0| \le 2^{(t-1)\mmbs_t(G-v)} \mmbs_t^{\star}(G-v)$.
Let us now conclude the proof of the lemma by showing that $\overline{\B}=\add(\overline{\B}_0,v,\B) \cup \overline{\B}_1$ satisfies the claim.
The bound of $|\overline{\B}|$ is immediate, and thus it remains to prove that $\overline{\B}$ is a blocking set of $(G,\F)$.
Let $S$ be an optimal solution of $(G,\F)$.
If $v \in S$, then according to \autoref{lemma:bsbigbssmall} (\autoref{lemma:bsbigbssmall:ibis}), there exists $B \in \overline{\B}_1$ such that $S \cap B = \emptyset$ and we are done.
If $v \notin S$, then as $v \notin \F$, according to \autoref{lemma:bsbigbssmall} (\autoref{lemma:bsbigbssmall:iibis}),  there exists $B \in \overline{\B}_0$ such that $S \cap B = \emptyset$.
By the same arguments than in case 2a) we get that there exists $B' \in \add(\overline{\B}_0,v,\B)$ such that $S \cap B' = \emptyset$, and thus that $\overline{\B}$ is a blocking set of $(G,\F)$.
\end{proof}

\begin{lemma}\label{lemma:bswithv}
  Let $(G,\F)$ be a clean instance of $\EKtH$, $v \in V(G)$,
  and $\B$ a blocking set of $(G,\F)$ such that $\{v\} \in \B$.
  Let $\B'$ such that $\B = \B' \cup \{v\}$.
There exists $\overline{\B} \subseteq \B'$ such that $\overline{\B} \cup \{v\}$ is a blocking set of $(G,\F)$ with $|\overline{\B}| \le \mmbs_t(G-v)$.
\end{lemma}
\begin{proof}
If $t^{(G,\F)}(v)=\{\bar{v}\}$, then $\overline{\B}=\emptyset$ satisfies the claim.
Otherwise, as $v \in t^{(G,\F)}(v)$, according to \autoref{lemma:bsbigbssmall} (\autoref{lemma:bsbigbssmall:i}), we get that $\opt(G,\F)=1+\opt(G-v,\F-v)$ and $\B'$ is a blocking set of
$(G-v,\F-v)$.
According to \autoref{lemma:bsclean} (\autoref{lemma:bsclean:ii}), $(G-v,\F-v)$ is a clean instance, and thus there exists $\overline{\B} \subseteq \B-v$
such that $\overline{\B}$ is a blocking set of $(G-v,\F-v)$ and $|\overline{\B}| \le \mmbs_t(G-v)$.
Let us prove that $\overline{\B} \cup \{v\}$ is a blocking set of $(G,\F)$.
Let $S$ be an optimal solution of $(G,\F)$. If $v \notin S$, we are done. Otherwise, according to \autoref{lemma:bsbigbssmall} (\autoref{lemma:bsbigbssmall:ibis}),
there exists $B\in\overline{\B}$ such that $S \cap B = \emptyset$, concluding the proof.
\end{proof}

The next lemma shows how, given a blocking set, we can ``zoom'' in a few pending components. It relies on \autoref{lemma:computeRoot} that will be proved in \autoref{sec:computeHittingSet}.
\begin{lemma}\label{lemma:bszooming}
  Let $(G,\F)$ be a clean instance of $\EKtH$. Let $\T$ be a $\bedt$ root of $G-N^t(G)$ (which exists by \autoref{lemma:computeRoot}).
  Suppose that $N^t(G)=\emptyset$ and that $G$ is connected, and let $\T = \{T\}$.
  \begin{enumerate}
  \item \label{lemma:bszooming1} Let $\B$ be a blocking set of $(G,\F)$. Then, either
  \begin{itemize}
  \item there exists $v \in T$ such that $\B \cap C(v)$ is a blocking set of $(G[C(v)],\F \cap C(v))$, or
  \item there exists $Z \in \B \cup \F$ such that $Z \subseteq T$, and for any $v \in Z, (\B \cap C(v)) \cup  \{v\}$ is a blocking set of $(G[C(v)],\F \cap C(v))$.
  \end{itemize}
\item \label{lemma:bszooming2} For any $v \in T$ and  blocking set $\overline{\B}$ of $(G[C(v)],\F \cap C(v))$, $\overline{\B}$ is a blocking set of $(G,\F)$.
\item \label{lemma:bszooming3} For any $Z \in \B \cup \F$ such that $Z \subseteq T$, and for any blocking sets $\overline{\B}_v$ of $(G[C(v)],\F \cap C(v))$ (with $\overline{\B}_v=\overline{\B}'_v \cup \{v\}$) for any $v \in Z$, it holds that
  \begin{itemize}
  \item if $Z \in \B$, then $Z \cup \bigcup_{v \in Z}\overline{\B}'_v$ is a blocking set of $(G,\F)$, and
  \item if $Z \in \F$, then $\bigcup_{v \in Z}\overline{\B}'_v$ is a blocking set of $(G,\F)$.
  \end{itemize}
  \end{enumerate}
\end{lemma}
\begin{proof}
Note first that using the trivial partition $\P=\{v \mid v \in T\}$, \autoref{obs:pendingadditive} implies that the pending partition $\P^P=\{C(v) \mid v \in T\}$ is additive,
  and \autoref{lemma:bsclean} implies that for any $v \in T$, $(G[C(v)],\F\cap C(v))$ is clean.
  Thus, we get property $p_1$: $\opt(G,\F)=\opt(G)=\sum_{v \in T}\opt(G[C(v)])=\sum_{v \in T}\opt(G[C(v)],\F \cap C(v))$.

  Proof of~\autoref{lemma:bszooming1}.
  Let us proceed by contradiction. For any $Z \in \B \cup \F$ such that $Z \subseteq T$, let $v_Z \in Z$ and $S_{v_Z}$ such that $S_Z$ is an optimal solution of $(G[C(v_Z)],\F \cap C(v_Z))$
  intersecting all sets of $(\B \cap C(v_Z)) \cup  \{v_Z\}$ (implying that $v_Z \in S_Z$).
  Let $V_1 = \{v_Z \mid Z \in \B \cup \F\}$ and $V_2 = T \setminus V_1$.
  For any $v \in V_2$, let $S_v$ be an optimal solution of $(G[C(v_Z)],\F \cap C(v_Z))$ intersecting all sets of $\B \cap C(v)$.
  Let $S = \bigcup_{v \in T}S_v$. Let us now prove that $S$ is an optimal solution of $(G,\F)$ intersecting all sets of $\B$, which is a contradiction to the hypothesis that $\B$ is a blocking set of $(G,\F)$.

  First, we get $|S|=\sum_{v \in T}\opt(G[C(v)],\F \cap C(v))=\opt(G) = \opt(G,\F)$ by property~$p_1$.
  Moreover, for any $t$-clique $K$ in $G$, \autoref{obs:lambda4} implies that there exists $v \in T$ such that $K \subseteq C(v)$, and thus $S_v \cap K \neq \emptyset$.
  Let us now verify that for any $Z \in \B \cup \F$, $S \cap Z \neq \emptyset$.
  Let $Z \in \B \cup \F$. If there exists $v \in T$ such that $Z \subseteq C(v)$, then $Z$ belongs to $\F \cap C(v)$ or $\B \cap C(v)$, and thus $S_v \cap Z \neq \emptyset$.
  Otherwise, \autoref{obs:lambda4} implies that $Z \subseteq T$ as $G[Z]$ is a clique. Thus, as $v_Z \in S_{v_Z}$, we get $S_{v_Z} \cap Z \neq \emptyset$.

  Before proving the last two items, let us prove property $p_2$: for any  optimal solution $S$ of $(G,\F)$, it holds that for any $v \in T$,  $|S \cap C(v)|$ is an optimal solution of $(G[C(v)],\F \cap C(v))$.
  Let $S$ be an optimal solution of $(G,\F)$. As by property~$p_1$ we have $\opt(G,\F)=\sum_{v \in T}\opt(G[C(v)],\F \cap C(v))$, and as for any $v \in T$, $S \cap C(v)$ is a solution of $(G[C(v)],\F \cap C(v))$,
  and thus $|S \cap C(v)| \ge \opt(G[C(v)],\F \cap C(v))$, we get that for any $v \in T$,  $S \cap C(v)$ is an optimal solution of $(G[C(v)],\F \cap C(v))$.

  Proof of~\autoref{lemma:bszooming2}.
  Let $S$ be an optimal solution of $(G,\F)$. By property~$p_2$, $S \cap C(v)$ is an optimal solution of $(G[C(v)],\F \cap C(v))$, and thus, there exists $B \in \overline{\B}$ such that $(S \cap C(v)) \cap B = \emptyset$, implying $S \cap B = \emptyset$.

  Proof of~\autoref{lemma:bszooming3}.
  Let $\B' = \bigcup_{v \in Z}\overline{\B}'_v$.
  Let $S$ be an optimal solution of $(G,\F)$. Suppose first that $Z \in \B$. If $S \cap Z = \emptyset$ then we are done, otherwise let $v \in S \cap Z$.
  By property~$p_2$, $S \cap C(v)$ is an optimal solution of $(G[C(v)],\F \cap C(v))$, and thus, as $\overline{\B}_v$ is a blocking set of $(G[C(v)],\F \cap C(v))$, there exists $B \in \overline{\B}_v$
  such that $(S \cap C(v)) \cap B = S \cap B = \emptyset$. As $v \in S$, we obtain that $B \in \overline{\B}'_v$, as required.
  Suppose now that $Z \in \F$. As $S$ must intersect $\F$, we again consider $v \in S \cap Z$, and we can follow the same arguments.
\end{proof}

The three next ``cleaning'' lemmas will help us in the main result of this section to reduce the proof to the most interesting case.
\begin{lemma}\label{lemma:bscleanF}
  Let $(G,\F)$ a clean instance of $\EKtH$.
  Let $\F' = \{Z \setminus N^t(G) \mid Z \in \F\}$. Then, $(G,\F')$  is well-defined (that is, $\emptyset \notin \F'$) and clean, and $\mmbs_t(G,\F) \le \mmbs_t(G,\F')$.
\end{lemma}
\begin{proof}
  Let $S$ be an optimal solution of $(G,\F)$. As $(G,\F)$ is clean, we get $|S|=\opt(G,\F)=\opt(G)$.
  By definition of $N^t(G)$, no optimal solution of $G$ contains a vertex in $N^t(G)$, and thus $S \cap N^t(G) = \emptyset$.
  As for any $Z \in \F$, $S \cap Z \neq \emptyset$, this implies that $\emptyset \notin \F'$, and that for any $Z \in \F'$, $S \cap Z \neq \emptyset$.
  Thus, $S$ is also a solution of $(G,\F')$, leading to $\opt(G,\F') \le |S|=\opt(G,\F)$. As the other inequality is true by definition, we get that $\opt(G,\F')=\opt(G,\F)=\opt(G)$,
  and thus $(G,\F')$ is clean.

  Let us now prove $\mmbs_t(G,\F) \le \mmbs_t(G,\F')$ using Property~\ref{prop:mmbs}.
  Let $\B$ be a blocking set of $(G,\F)$. Note that $\B$ is also a blocking set of $(G,\F')$. Indeed, if there existed an optimal solution $S$ of $(G,\F')$ intersecting all sets of $\B$,
  then $S$ is also a solution of $(G,\F)$, and $|S|=\opt(G,\F)$ by the previous equalities, leading to a contradiction.
  Thus, as $\B$ is a blocking set of $(G,\F')$ and $(G,\F')$ is clean, there exists $\overline{\B} \subseteq \B$ such that $|\overline{\B}| \le \mmbs_t(G,\F')$ and $\overline{\B}$ is a blocking set of $(G,\F)$.
\end{proof}

\begin{lemma}\label{lemma:bscleanN}
  Let $(G,\F)$ a clean instance of $\EKtH$ and  $\F$ such that for any $Z \in \F$, $Z \subseteq V(G) \setminus N^t(G)$.
  Then, $(G-N^t(G),\F)$ is clean, and $\mmbs_t(G,\F) \le \mmbs_t(G-N^t(G),\F)$.
\end{lemma}
\begin{proof}
  Note that, as for any $Z \in \F$, $Z \subseteq V(G) \setminus N^t(G)$, $(G-N^t(G),\F)$ is well-defined. By definition of $N^t(G)$ and as $(G,\F)$ is clean,
  we have property $p_1$: $\opt(G-N^t(G),\F)=\opt(G,\F)=\opt(G)=\opt(G-N^t(G))$.
  This property implies that $(G-N^t(G),\F)$ is clean.
  We also have property $p_2$: for any $S$, $S$ is an optimal solution of $(G,\F)$ if and only if $S$ is an optimal solution of $(G-N^t(G),\F)$.

  Let us now prove $\mmbs_t(G,\F) \le \mmbs_t(G-N^t(G),\F)$ using Property~\ref{prop:mmbs}.
  Let $\B$ be a blocking set of $(G,\F)$.

  If there exists $B \in \B$ such that $B \subseteq N^t(G)$, then let us prove that $\overline{\B}=\{B\}$ is a blocking set.
  Let $S$ be an optimal solution of $(G,\F)$. By definition of $N^t(G)$, and as for any $Z \in \F$, $Z \subseteq V(G) \setminus N^t(G)$, we get $S \cap N^t(G) = \emptyset$, implying $S \cap B = \emptyset$.

  Otherwise, we can define $\B' = \{B \setminus N^t(G) \mid B \in \B\}$ while ensuring that $\emptyset \notin \B'$.
  Let us show that $\B'$ is a blocking set of $(G-N^t(G),\F)$. Let $S$ be an optimal solution of $(G-N^t(G),\F)$. By property~$p_2$, $S$ is also an optimal solution of $(G,\F)$,
  and thus there exists $B \in \B$ such that $S \cap B = \emptyset$, implying $(B \setminus N^t(G)) \cap S = \emptyset$, and thus that $\B'$ is a blocking set of
  $(G-N^t(G),\F)$. As $(G-N^t(G),\F)$ is clean, there exists a blocking set $\overline{\B}' \subseteq \B'$ of $(G-N^t(G),\F)$ such that $|\overline{\B}'| \le \mmbs_t(G-N^t(G),\F)$.
  We define $\overline{\B}=\add(\overline{\B}',N^t(G),\B)$, where $\add(\overline{\B}',N^t(G),\B)=\bigcup_{B \in \overline{\B}'}f(B)$, where $f(B)$ is any $B' \in \B$ such that $B' \setminus N^t(G) = B$.
  Let us finally check that $\overline{\B}$ is a blocking set of $(G,\F)$. Let $S$ be an optimal solution of $(G,\F)$. By property~$p_2$, $S$ is an optimal solution of $(G-N^t(G),\F)$, and
  thus there exists $B \in \overline{\B}'$ such that $S \cap B = \emptyset$. As $S \cap N^t(G) = \emptyset$, we also have $S \cap f(B)= \emptyset$.
\end{proof}

\begin{lemma}\label{lemma:bscleanCC}
  Let $(G,\F)$ be a clean instance of $\EKtH$.
  Suppose $N^t(G)=\emptyset$, and let $\{C_i \mid i \in [c]\}$ be the connected components of $G$.
  Then, $\mmbs_t(G,\F) \le \max_{i \in [c]}\mmbs_t(G[C_i],\F \cap C_i)$.
\end{lemma}
\begin{proof}
Let us prove the inequality using Property~\ref{prop:mmbs}, and let us consider  a blocking set $\B$ of $(G,\F)$.
As the partition given by the connected components is additive, \autoref{lemma:bsclean} implies that for any $i \in [c]$, $(G[C_i], \F \cap C_i)$ is clean. Moreover,
as $\F$ are cliques, we get that $\F = \bigcup_{i \in [c]}(\F \cap C_i)$, and thus by additivity of the partition we get property $p_0$:
for any set of vertices $S$, $S$ is an optimal solution of $(G,\F)$ if and only if $S \cap C_i$ is an optimal solution of $(G[C_i],\F \cap C_i)$ for any $i \in [c]$.

Let us first prove that there exists $i \in [c]$ such that $\B \cap C_i$ is a blocking set of $(G[C_i], \F \cap C_i)$. Suppose by contradiction that there is no such $i$, and for any $i \in [c]$ let $S_i$ be an optimal solution of $(G[C_i],\F \cap C_i)$ intersecting all sets of $\B \cap C_i$, and $S = \bigcup_{i \in [c]}S_i$. By property $p_0$, we get that $S$ is an optimal solution of $(G,\F)$.
As all sets of $\B$ are cliques, we get that $\B = \bigcup_{i \in [c]}(\B \cap C_i)$, implying that $S$ also intersects all sets of $\B$. This contradicts the fact that $\B$ is a blocking set of $(G,\F)$.

Thus, there exists $i \in [c]$ such that $\B \cap C_i$ is a blocking set of $(G[C_i], \F \cap C_i)$, and as $(G[C_i], \F \cap C_i)$ is clean, there exists $\overline{\B} \subseteq \B \cap C_i$ such that $\overline{\B}$ is a blocking set of $(G[C_i], \F \cap C_i)$, and $|\overline{B}| \le \mmbs_t(G[C_i], \F \cap C_i)$. As $\overline{\B}$ has the claimed size, in only remains to check that $\overline{\B}$ is a blocking set of $(G,\F)$. Let $S$ be an optimal solution of $(G,\F)$.
By property $p_0$, $S \cap C_i$ is an optimal solution of $(G[C_i], \F \cap C_i)$,
implying that there exists $Z \in \overline{B}$ such that $(S \cap C_i) \cap Z = \emptyset$, implying also that $S \cap Z = \emptyset$, and thus that $\overline{B}$ is a blocking set of $(G,\F)$.
\end{proof}

We are now ready to prove the main theorem of this section.
\begin{theorem*}[\autoref{thm:mmbslambda4} restated]
  For any graph $G$  and any integer $t \ge 3$, we have $\mmbs_t(G) \le \beta(\bedt(G),t)$, where $\beta(x,t)=
\underbrace{
  {{{{{^{2\vphantom{h}}}^{t2\vphantom{h}}}^{t2\vphantom{h}}}}^{\cdots\vphantom{h}}}^{t\vphantom{h}}
}_{\text{$x$ times}}
$
({{i.e.,}} $\beta(1)=2^t, \beta(2)=2^{t2^t}$, etc.)
\end{theorem*}
\begin{proof}
  For the sake of readability, let us consider that $t$ is fixed, and let $f:\mathbb{N} \rightarrow \mathbb{N}$ such that $f(0)=1$, and for any integer $x \ge 1$, $f(x)=(2^{(t-1)f(x-1)}+1)f(x-1)$.
  The proof is by induction on $\bedt(G)$. If $\bedt(G)=0$, $G$ is a (possibly empty) collection of $K_t$-free components, implying $\opt(G)=0$,
  and thus that for any clean instance $(G,\F)$, $\F = \emptyset$, and thus that any inclusion-wise minimal blocking set of a clean instance $(G,\emptyset)$ has at most one element.

  Let us now turn to the induction case.
  Let $(G,\F)$ be a clean instance of $\EKtH$ and $\B$ be a blocking set of $(G,\F)$.
  Let $\T$ be a $\bedt$ root of $G-N^t(G)$ (which exists according to \autoref{lemma:computeRoot}). Let us prove that $\mmbs_t(G,\F) \le f(\bedt(G))$.
  By Lemmas \autoref{lemma:bscleanF}, \ref{lemma:bscleanN}, and \autoref{lemma:bscleanCC}, we can assume that $N^t(G)=\emptyset$ and that $G$ is connected.
  This implies that $\T=\{T\}$ for some subset $T \subseteq V(G)$.
  By \autoref{lemma:bszooming} (\autoref{lemma:bszooming1}), we get the following two cases:
 \begin{enumerate}
  \item There exists $v \in T$ such that $\B \cap C(v)$ is a blocking set of $(G[C(v)],\F \cap C(v))$.
  \item There exists $Z \in \B \cup \F$ such that $Z \subseteq T$, and for any $v \in Z, (\B \cap C(v)) \cup  \{v\}$ is a blocking set of $(G[C(v)],\F \cap C(v))$.
  \end{enumerate}

 Note first that using the trivial partition $\P=\{v \mid v \in T\}$, \autoref{obs:pendingadditive} implies that the pending partition $\P^P=\{C(v) \mid v \in T\}$ is additive,
 and \autoref{lemma:bsclean} implies that for any $v \in T$, $(G[C(v)],\F\cap C(v))$ is clean.
 Moreover, by \autoref{lemma:bedtcv}, for any $v \in T$, $\bedt(G[C(v)]-v) \le \bedt(G)-1$.

 \subparagraph*{Case 1.} If $v \notin \B$, then by \autoref{lemma:bswithoutv}, there exists a blocking set $\overline{\B}_v \subseteq \B \cap C(v)$ of $(G[C(v),\F \cap C(v))$
   satisfying that $|\overline{\B}_v| \le (2^{(t-1)\mmbs_t(G[C(v)]-v)}+1) \mmbs_t^{\star}(G[C(v)]-v) \le (2^{(t-1)f(\bedt(G)-1)}+1) f(\bedt(G)-1) = f(\bedt(G))$.
   Moreover, by \autoref{lemma:bszooming} (\autoref{lemma:bszooming2}),  $\overline{\B}_v$ is also a blocking set of $(G,\F)$.

 If $v \in \B$, then by \autoref{lemma:bswithv}, there exists a blocking set $\overline{\B}_v \cup \{v\} \subseteq \B \cap C(v)$ of $(G[C(v),\F \cap C(v))$
   with $|\overline{\B}_v \cup \{v\}| \le  \mmbs_t(G[C(v)]-v)+1 \le f(\bedt(G)-1)+1 \le f(\bedt(G))$.
   Again, by \autoref{lemma:bszooming} (\autoref{lemma:bszooming2}),  $\overline{\B}_v \cup \{v\}$ is also a blocking set of $(G,\F)$.

   \subparagraph*{Case 2.} Since for any  $v \in Z,  (\B \cap C(v)) \cup  \{v\}$ is a blocking set of $(G[C(v)],\F \cap C(v))$, by \autoref{lemma:bswithoutv},
   there exists a blocking set $\overline{\B}_v \cup \{v\}$ of $(G[C(v),\F \cap C(v))$ where $\overline{\B}_v \subseteq (\B \cap C(v))$ and $|\overline{\B}_v| \le  \mmbs_t(G[C(v)]-v) \le f(\bedt(G)-1)$.
     If $Z \in \B$, then by \autoref{lemma:bszooming} (\autoref{lemma:bszooming3}), $Z \cup \bigcup_{v \in Z}\overline{\B}_v$ is a blocking set of $(G,\F)$ of size at most $(t-1)f(\bedt(G)-1)+1 \le f(\bedt(G))$.
     Otherwise, $Z \in \F$, and by \autoref{lemma:bszooming} (\autoref{lemma:bszooming3}), $\bigcup_{v \in Z}\overline{\B}_v$ is a blocking set of $(G,\F)$ of size at most $(t-1)f(\bedt(G)-1) \le f(\bedt(G))$.

     Thus, we obtained that $\mmbs_t(G) \le f(\bedt(G))$.
     Let us prove that $f(x) \le \beta(x)$.
     The result is clear for $x=1$. 
      For $x \ge 2$, $f(x) \le (2^{(t-1)f(x-1)}+1)f(x-1) \le (2^{(t-1)f(x-1)+1})f(x-1) \le 2^{(t-1)f(x-1)+1+\log(f(x-1))}$, and as $1+\log(y) \le y$ for $y \ge 3$ and $f(x-1) \ge f(1) \ge 3$, we get $f(x) \le 2^{tf(x-1)} \le 2^{t\beta(x-1,t)}=\beta(x,t)$.
\end{proof}

\subsection{Minimal blocking sets in graphs of bounded treedepth}
\label{ap:improved-mmbs-treedepth}

In this subsection we obtain upper bounds on the size of minimal blocking sets in graphs of bounded treedepth, which substantially improve those given in \autoref{thm:mmbslambda4} in graphs of bounded $\bedt(G)$.

For the following argument, recall that an elimination forest of a graph~$G$ is a rooted tree~$T$ on vertex set~$V(G)$ such that for each edge~$\{u,v\} \in E(G)$, vertex~$u$ is an ancestor of vertex~$v$ in~$T$, or vice versa. For a node~$x$ in an elimination forest~$T$, we use~$\text{anc}_T(x)$ to denote the set of ancestors of~$x$, which includes~$x$ itself. It is well-known that the treedepth of a connected graph~$G$ is equal to the minimum, over all elimination forests~$T$ of~$G$, of the maximum number of nodes on a root-to-leaf path in~$T$. For a node~$x$ in an elimination forest~$T$ of a graph~$G$, we use~$G_x$ to denote the subgraph of~$G$ induced by the descendants of~$x$ in~$T$.

Observe that, unlike \autoref{thm:mmbslambda4}, the following bound does not depend on $t$.
\begin{lemma}\label{lemma:mmbstd}
For any graph $G$ and integer $t \ge 3$, $\mmbs_t(G) \leq \td(G)^{\td(G)} \cdot 2^{\td(G)^2}$.
\end{lemma}

\begin{proof}
Let \(G\) be a graph of treedepth \(\eta\) and let \(\B\) be a set of cliques in \(G\), each of size less than \(t\).
Let $\B$ be a blocking set of $G$, and let us show that there exists a subset \(\overline{\B}\subseteq \B\) of size at most \(\eta^\eta\cdot 2^{\eta^2}\) such that
$\overline{\B}$ is still a blocking set of $G$.

Let us recall that given a graph $G$ and $\B$ a set of cliques in $G$, we denote by $\opt(G)$ (resp. $\opt(G,\B)$) the minimum size of a subset of vertices intersecting all $t$-cliques (resp. and all sets in $\B$).

We may assume that \(G\) is connected, since a blocking set forms a blocking set for at least one connected component, and any blocking set for a component is a blocking set for \(G\). Let \(\overline{\B}\subseteq \B\) be an inclusion-minimal subset such that \(\opt (G,\overline{\B}) > \opt (G)\). Consider an optimal elimination forest \(T\) of \(G\), of depth \(\eta \). We now prove the following:

\begin{claim}
For each node \(x\in V(T)\), there are at most \(\eta\cdot 2^\eta\) distinct child subtrees of \(x\) that contain a vertex of \(V(\overline{\B})\).
\end{claim}
\begin{claimproof}
Assume for a contradiction that there exists \(x\in V(T)\) such that more than \(\eta\cdot 2^\eta\) distinct subtrees rooted at children \(c_1,\dots, c_m\) of \(x\) contain a vertex \(u_1,\dots, u_m\) of \(V(\overline{\B})\), and let \(B_1,\dots, B_m \in \overline{\B}\) such that \(u_i\in B_i\) for each \(i\). Note that \(B_i \neq B_j\) for \(i \neq j\) because the subgraphs in \(\B \supseteq \overline{\B}\) are cliques and a clique cannot contain vertices from two incomparable subtrees of an elimination forest
Since \(\overline{\B}\) is inclusion-minimal with respect to \(\opt (G,\overline{\B}) > \opt (G)\), we have \(\opt (G,\overline{\B}\setminus\{B_i\}) = \opt (G)\) for each \(i\), meaning there is a vertex set \(S_i\) of size \(\opt (G)\) which hits all $t$-cliques in \(G\) and all graphs \(\overline{\B}\setminus\{B_i\}\).

Since \(m > \eta\cdot 2^\eta\), there are more than \(\eta\) children \(c_i\) for which the corresponding solutions \(S_i\) make exactly the same selection from the ancestors~$\text{anc}_T(x)$ of \(x\) in~$T$. Let \(c_1,\dots, c_{\eta+1}\) be such children, for which there exists a set \(C\) of ancestors of \(x\) such that \(S_i \cap \text{anc}_T (x) = C\) for each \(i\). Consider the solution \(S_1\). For at least one of the \(\eta\) remaining involved children \(c_2,\dots, c_{\eta+1}\), we have \(|S_1 \cap V(G_{c_i})| = \opt (G_{c_i})\), where \(G_{c_i}\) is the subgraph of \(G\) induced by \(c_i\) and its descendants.

The reason for this is as follows: $S_1$ trivially contains at least $\text{opt}_t(G_{c_i})$ vertices from $G_{c_i}$ for each $2 \leq i \leq \eta+1$. If it picks strictly more in each of them, then (since $|\text{anc}_T(x)| < \eta$ since $x$ itself has a child and the number of vertices on a root-to-leaf path is $\eta$) there is a solution for $K_t$-{\sc Subgraph Hitting} which is smaller than $|S_1| = \opt(G)$: replace the solutions inside those $\eta$ children by locally optimal solutions, and then add $x$ and all its ancestors. Since all cliques of~$G$ appear on a single root-to-leaf path of any elimination forest, the replacement gives a valid solution whose size is strictly smaller than~$|S_1|$; but this contradicts $|S_1| = \opt(G)$. Hence there is indeed a child $c_i$ with $i>1$ for which $S_i \cap G_{c_j}$ has size $\opt(G_{c_j})$.

Now consider $S^* := (S_j \setminus V(G_{c_j})) \cup (S_1 \cap V(G_{c_j}))$. In other words, we take the solution $S_j$, but replace what it chooses from $G_{c_j}$ with how $S_1$ chooses from $G_{c_j}$. This cannot increase the size since $S_j$ pays at least $\opt(G_{c_j})$ inside $G_{c_j}$ and $S_1$ pays $\opt(G_{c_j})$. This also does not cause any $K_t$ to become unhit. Each $K_t$ is on a single root-to-leaf path of the decomposition. The new solution $S^*$ still hits all of them, since for each root-to-leaf path, the choices $S^*$ makes on that path agree with the choices made by some valid solution (either $S_1$ or $S_j$) on that path; here we exploit that $S_1$ and $S_j$ agree on their choices from the ancestors of $x$.

Now comes the key point. By assumption, $S_1$ hits all of $\overline{\B}$ except $B_1$, so in particular it hits $B_j$. We knew that $S_j$ hits all of $\overline{\B}$ except for $B_j$, and that $B_j$ contains a vertex in the child subtree rooted at $c_j$; hence the entire clique $B_j$ is contained in the subtree of $c_j$ together with its ancestors. From this, it is straightforward to verify that replacing what $S_j$ does inside $G_{c_j}$ with what $S_1$ does there, yields a solution that hits all of $\overline{\B}$; but that is a contradiction to the choice of $\overline{\B}$ since $|S^*| = \opt(G)$.
\end{claimproof}

The claim proves the lemma, because it effectively shows that the vertices of \(V(\overline{\B})\) and their ancestors form a subtree of \(T\) of depth \(\eta\) and branching factor \(\eta\cdot 2^\eta\), which proves the total number of vertices is bounded by \(\eta^\eta\cdot 2^{\eta^2}\); this clearly bounds the number of subgraphs in \(\overline{\B}\).
\end{proof}

\section{A polynomial kernel for $\KtHMdecbis{\lambda}$}\label{sec:kernel}

In this section we provide a polynomial kernel for $\KtHMdecbis{\lambda}$. 
 Recall that given any input $(G,\F)$ of the $\EKtH$ problem, we use the notation $\opt(G,\F)$ to denote the minimum size of a vertex subset of $G$ intersecting all $t$-cliques of $G$ and all the sets in $\F$.
Let us first define the crucial notion of chunk. Informally, a chunk is a small part of $X$ that is a small subset of $X$ which is candidate to be avoided by a solution.
\begin{definition}[chunk]\label{def:chunk}
  Given an input $(G,X)$ of $\KtHM^\lambda$, a \emph{chunk} is a set $X' \subseteq X$ such that $|X'| \le c(\lambda,t)$, where $c(\lambda,t)=(t-1)\beta(\lambda,t)$ and $\beta$ is defined in \autoref{thm:mmbslambda4}, and $X'$ contains no $t$-clique.
 The \emph{set of chunks} of $(G,X)$ is denoted by $\X$.
\end{definition}

We are now ready to define the notion of conflict.

\begin{definition}[conf]\label{def:conflict}
  Given a graph $G$ and two disjoint subsets of vertices $S_1,S_2$ such that $S_1$ does not contain a $t$-clique, we define the \emph{conflict} function as $\conf^t_{S_1}(S_2)=\opt(G[S_2],\pr^t_{S_1}(S_2))-\opt(G[S_2])$.
\end{definition}
Informally, $\conf^t_{S_1}(S_2)$ corresponds to the number of vertices we have to overpay in $S_2$ when we do not take any vertex from $S_1$ into the solution, compared to $\opt(G[S_2])$.
The following is a corollary of  \autoref{lemma:poly-opt-conf-lambda4} presented in \autoref{sec:computing-optimum-bounded-lambda}.
\begin{corollary}\label{obs:computingconflict}
Given a graph $G$, $S_2 \subseteq V(G)$ such that $\bedt(G[S_2]) \le \lambda$, and $S_1 \subseteq V(G)$ disjoint from $S_2$ such that $S_1$ does not contain a $t$-clique, we can compute if $\conf^t_{S_1}(S_2)>0$ in time $p_0^{(\lambda,t)}(n)$, where $n=|V(G)|$ and $p_0^{(\lambda,t)}$ is defined in \autoref{lemma:poly-opt-conf-lambda4}.
\end{corollary}

The following lemma, saying that in a graph with small $\mmbs_t$ it is possible, given any set $S_1$ inducing some conflict, to extract a small subset $\overline{S}_1$ that still ``certifies'' the conflict, has also been used in previous work, for example in \cite{BougeretS18}. However, as the problem considered here is different, we prefer to include the proof for the sake of completeness.

\begin{lemma}\label{lemma:smallcertif}
  Let $G$ be a graph, let $S_1,S_2$ be two disjoint subsets of vertices of $G$ such that $S_1$ does not contain a $t$-clique, and let $\beta$ be an integer such that $\mmbs_t(G[S_2]) \le \beta$.
  Then, $\conf^t_{S_1}(S_2)=0$ if and only if for any subset $\overline{S}_1 \subseteq S_1$ such that $|\overline{S}_1| \le (t-1)\beta$, $\conf^t_{\overline{S}_1}(S_2)=0$.
\end{lemma}
\begin{proof}
The forward implication holds as for any subset $\overline{S}_1 \subseteq S_1$, we have that $$\opt(G[S_2],\pr^t_{\overline{S}_1}(S_2)) \le \opt(G[S_2],\pr^t_{S_1}(S_2)) = \opt(G[S_2]).$$
  Let us prove the backward implication by contrapositive.
  Suppose that $\conf^t_{S_1}(S_2)>0$, implying that $\opt(G[S_2],\pr^t_{S_1}(S_2)) > \opt(G[S_2])$ and thus that $\pr^t_{S_1}(S_2)$ is a blocking set of $G[S_2]$.
  As $\mmbs_t(G[S_2]) \le \beta$, by Property~\ref{prop:mmbs} there exists $\overline{\B} \subseteq \pr^t_{S_1}(S_2)$ such that $\opt(G[S_2],\overline{\B}) >\opt(G[S_2])$ and $|\overline{\B}| \le \beta$.
  Note that for any $Z \in \overline{\B}$, there exists $f(Z) \subseteq S_1$ such that $Z \cup f(Z)$ is a $t$-clique.
  Thus, we define $\overline{S}_1 =\bigcup_{Z \in \overline{\B}}f(Z)$, and we have that $\overline{\B} \subseteq \pr^t_{\overline{S}_1}(S_2)$, implying
  $\opt(G[S_2],\pr^t_{\overline{S}_1}(S_2))>\opt(G[S_2],\overline{\B}) > \opt(G[S_2])$, and thus that $\conf^t_{\overline{S}_1}(S_2)>0$.
  Moreover, $|\overline{S}_1| \le |\overline{\B}|(t-1) \le (t-1)\beta$.
\end{proof}


 Let $(G,X)$ be an input of $\KtHM^\lambda$, let $N$ be a set of non-$K_t$-vertices of $G-X$,  and let $\T$ be a root of $G-X-N$.
For any $\P=\{(V_i,N_i) \mid i \in [|\P|]\}$, where for $i \in [|\P|]$, $V_i \subseteq V(\T)$ and $N_i \subseteq N$, we denote by $N^{\P}=\bigcup_{(V_iN_i) \in \P}N_i$
and $\T^{\P}=\bigcup_{(V_iN_i) \in \P}V_i$.

Let us start by defining the following marking algorithm, which will be used as a subroutine in the kernelization algorithm. See \autoref{fig:mark} for an illustration.

\begin{definition}[{\tt mark}]\label{def:mark}
  Let $(G,X)$ be an input of $\KtHM^\lambda$.
  We define an algorithm {\tt mark}$(\T,N,X',c,N',M')$ which returns a subset of vertices of $V(\T)$, and whose parameters are such that

  \begin{itemize}
    \item $N$ is a set of non-$K_t$-vertices of $G-X$ and $\T$ is a root of $G-X-N$,
  \item $X' \in \X$,
  \item $X' \cup N'$ does not contain a $t$-clique ,
     \item $N' \subseteq N$ such that $|N'| \le c(\lambda,t)-c$,
    \item $c$ is an integer with $-1 \le c \le c(\lambda,t)$, and
    \item $M' \subseteq V(\T)$, $|M'| \le t-2$, and $M'$ is a clique.
  \end{itemize}

  Given such a tuple $(\T,N,X',c,N',M')$, we say that a pair $(V_i,N_i)$ is a \emph{$(\T,N,X',c,N',M')$-part} if

  \begin{enumerate}
  \item \label{defpart1}$N_i \subseteq N \setminus N'$, $|N_i| \le c$, $X' \cup N' \cup N_i$ does not contain a $t$-clique, and
  \item \label{defpart2}$V_i \subseteq V(\T)$, $|V_i \cup M'| \le t-1$, and $V_i \cup M'$ is clique, and
  \item \label{defpart3}$\conf^t_{X'}(M' \cup C(V_i) \cup N' \cup N_i) > 0$
  \end{enumerate}

  We are now ready to define {\tt mark}$(\T,N,X',c,N',M')$:

  \begin{itemize}
  \item If $c=-1$, return $\emptyset$.
  \item Otherwise:
    \begin{itemize}
    \item Construct a packing $\P$ (meaning that $V_i's$ are pairwise disjoint and $N_i's$ are pairwise disjoint) of $(\T,N,X',c,N',M')$-parts $(V_i,N_i)$ until $|\P|=|X|+1$ or $\P$ is inclusion-wise maximal.
    \item If $|\P| = |X|+1$, return $\T^{\P}$.
    \item Otherwise, return $\bigcup_{g \in N^{\P}}${\tt mark}$(\T,N,X',c-1,N' \cup \{g\},M') \cup \T^{\P}$.
    \end{itemize}
  \end{itemize}

For any $X' \in \X$ and $M' \subseteq V(\T)$ such that $|M'| \le t-2$ and $M'$ is a clique, let \\ \mainMark$(\T,N,X',M')= ${\tt mark}$(\T,N,X',c(\lambda,t),\emptyset,M')$.
\end{definition}

  \begin{figure}
\begin{center}
    \includegraphics[scale=0.6]{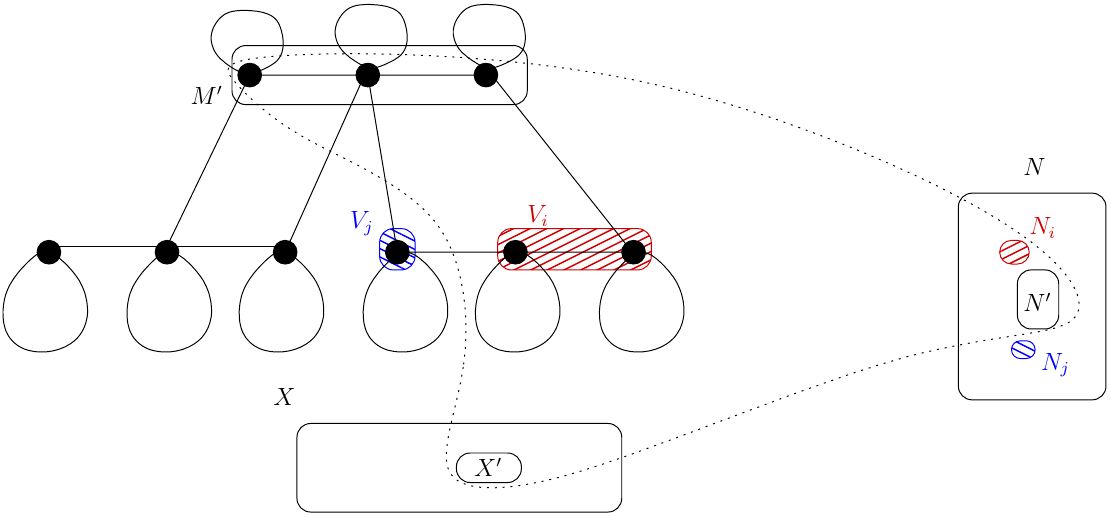}
    \caption{Example of a call to {\tt mark}$(\T,N,X',c,N',M')$ and of a $(\T,N,X',c,N',M')$-part $(V_i,N_i)$ (with in particular $\conf^t_{X'}(M' \cup C (V_i) \cup N' \cup N_i) > 0$).
      The nine vertices of $V(\T)$ are depicted as black filled circles.
    }
    \label{fig:mark}
\end{center}
  \end{figure}

\begin{observation}\label{obs:mark1}
Any recursive call of the above algorithm satisfies the required condition as, in particular, for any $g \in N^{\P}$, $X' \cup N' \cup \{g\} \subseteq X' \cup N' \cup N_i$ for some $(\T,N,X',c,N',M')$-part
$(V_i,N_i)$, and by definition of a part, $X' \cup N' \cup N_i$ does not contain a $t$-clique.
\end{observation}

In the following lemma we prove several properties of the above marking algorithm. Informally, \autoref{lemma:marking:2} of \autoref{lemma:marking} tells us that if a part $(V_i,N_i)$ has not been marked, then
we have marked $|X|+1$  other parts.
\begin{lemma}[properties of marking]\label{lemma:marking}
  Let $(G,X)$ be an input of $\KtHM^\lambda$, $N$ be a set of non-$K_t$-vertices of $G-X$, and $\T$ be a root of $G-X-N$.
  Let $X' \in \X$ and $M' \subseteq V(\T)$ such that $|M'| \le t-2$ and $M'$ is a clique.
 Then,
  \begin{enumerate}
  \item  \label{lemma:marking:0}  \mainMark$(\T,N,X',M')$ runs in time $p_1^{(\lambda,t)}(n)=\left(n^{t+c(\lambda,t)+\O(1)}c(\lambda,t)p_0^{(\lambda,t)}(n)\right)^{c(\lambda,t)+1}$, where $p_0^{(\lambda,t)}$ is defined in \autoref{lemma:poly-opt-conf-lambda4}.
    \item  \label{lemma:marking:1} $|$\mainMark$(\T,N,X',M')| \le p^\lambda(|X|)$ where $p_2^{(\lambda,t)}(|X|)=2t|X|(2c(\lambda,t)|X|)^{(c(\lambda,t)+1)}$.
    \item  \label{lemma:marking:2} If there is a $(\T,N,X',c(\lambda,t),\emptyset,M')$-part $(V_i,N_i)$ such that  $V_i \cap $\mainMark$(\T,N,X',M') = \emptyset$,
      then there exists $N' \subseteq N_i$, and a packing $\P^\star$ of $(\T,N,X',c(\lambda,t)-|N'|,N',M')$-parts such that $|\P^\star|=|X|+1$ and $\T^{\P^\star} \subseteq $\mainMark$(\T,N,X',M')$.

  \end{enumerate}
\end{lemma}

\begin{proof}
  Proof of~\autoref{lemma:marking:0}.
  For any $c \ge 0$, let $f_1(c)=\left(\alpha_2n^{t+c(\lambda,t)+\alpha_1+1}c(\lambda,t)p_0^{(\lambda,t)}(n)\right)^{c+1}$ where $\alpha_1$ and $\alpha_2$ are constants defined later.
  Let us prove by induction on $c \ge 0$ that for any input  $(\T,N,X',c,\tilde{N},M')$, {\tt mark}$(\T,N,X',c,\tilde{N},M')$ runs in at most $f_1(c)$ steps.
  Observe that for any input, a naive implementation of {\tt mark} consists in listing in $\O(n^{t+c(\lambda,t)})$ all possible subsets $V_i$ and $N_i$ with properties~\ref{defpart1} and~\ref{defpart2},
  checking for each of them in time $p_0^{(\lambda,t)}(n)$ using \autoref{obs:computingconflict} if \autoref{defpart3} holds, and then given the set of all parts, construct in time $\O(n^\alpha_1)$ a maximal packing (where $\alpha_1$ is a constant).
  Thus, property is true for $c=0$, for any $c > 0$, as {\tt mark} makes $|N^{\P}| \le |X|c \le nc(\lambda,t)$ recursive calls, we get that {\tt mark}$(\T,N,X',c,\tilde{N},M')$ runs in time $\O(n^{t+c(\lambda,t)+\alpha_1}p_0^{(\lambda,t)}(n)) \cdot nc(\lambda,t) \cdot f_1(c-1)$. Choosing $\alpha_2$ as the multiplicative factor of the previous $\O$-notation, we get that {\tt mark}$(\T,N,X',c,\tilde{N},M')$ runs in time $f_1(c)$.

  Proof of~\autoref{lemma:marking:1}.
  For any $c \ge 0$, let $f_2(c)=t\sum_{\ell=0}^c c(\lambda,t)^{\ell}(|X|+1)^{\ell+1}$.
  Let us prove by induction on $c \ge 0$ that for any input  $(\T,N,X',c,\tilde{N},M')$ of {\tt mark}, $|${\tt mark}$(\T,N,X',c,\tilde{N},M')| \le f_2(c)$.
  For $c=0$, $|{\tt mark}(\T,N,X',0,\tilde{N},M')|=|\T^{\P}| \le (|X|+1)t$.
  Suppose now that the claim is true for $c-1$ and let us prove it for $c$.
  Let $(\T,N,X',c,\tilde{N},M')$ be an input of {\tt mark}. As $|N^{\P}| \le |X|c \le (|X|+1)c(\lambda,t)$, and using induction hypothesis, we have $|{\tt mark}(\T,N,X',c,\tilde{N},M')| \le (|X|+1)t+(|X|+1)c(\lambda,t)f_2(c-1)=f_2(c)$. As $f_2(c) \le t(|X|+1)(c(\lambda,t)(|X|+1))^{c+1}$, we get the claimed bound.

  Proof of~\autoref{lemma:marking:2}.
  Let us first prove the following property $p_1$: for any input $(\T,N,X',c,\tilde{N},M')$ of {\tt mark} (thus satisfying the requirements in \autoref{def:mark}) where $c \ge 0$,
  if there exists a $(\T,N,X',c,\tilde{N},M')$-part $(V_i,N_i)$ such that $V_i \cap {\tt mark}(\T,N,X',c,\tilde{N},M') = \emptyset$, then there exists $N^\star \subseteq N$ and $\P^\star$ such that

  \begin{itemize}
  \item  $(\T,N,X',c-|N^\star|,\tilde{N} \cup N^\star,M')$ satisfies the condition required by an input of {\tt mark}, and
  \item $\P^\star$ is a packing of $(\T,N,X',c-|N^\star|,\tilde{N} \cup N^\star,M')$-parts with $|\P^\star|=|X|+1$ and $\T^{\P^\star} \subseteq {\tt mark}(\T,N,X',c,\tilde{N},M')$.
  \end{itemize}

  We prove this property by induction on $c$.
  Consider first the case where $c=0$.
  Suppose there is a $(\T,N,X',0,\tilde{N},M')$-part $(V_i,N_i)$ such that $V_i \cap {\tt mark}(\T,N,X',0,\tilde{N},M') = \emptyset$.
  Let $\P$ be the packing of $(\T,N,X',0,\tilde{N},M')$-parts found by {\tt mark}$(\T,N,X',0,\tilde{N},M')$. Let us prove that $\P^\star=\P$ and $N^\star=\emptyset$ satisfy the required properties.
  The first one is trivial. Let us prove that $|\P|=|X|+1$, implying the second property. Suppose by contradiction that $|\P|\le |X|$. As $(V_i,N_i)$
  is a $(\T,N,X',0,\tilde{N},M')$-part such that $V_i \cap {\tt mark}(\T,N,X',0,\tilde{N},M') = \emptyset$, it implies that $V_i \cap \T^{\P} = \emptyset$, implying that for any $(V_j,N_j) \in \P$,
  $V_i \cap V_j = \emptyset$. However, as $c=0$, all the $N_j$ are empty, and thus for any $(V_j,N_j) \in \P$, we have $V_j \cap V_i = \emptyset$ and $N_j \cap N_i = \emptyset$,
  contradicting the maximality of $\P$.

  Suppose now that the claim is true for $c-1$ and let us prove it for $c$.
  Suppose there exists a $(\T,N,X',c,\tilde{N},M')$-part $(V_i,N_i)$ such that $V_i \cap {\tt mark}(\T,N,X',c,\tilde{N},M') = \emptyset$.
  Let $\P$ be the packing of $(\T,N,X',c,\tilde{N},M')$-parts found by {\tt mark}$(\T,N,X',c,\tilde{N},M')$.
  If $|\P|=|X|+1$, then $\P^\star=\P$ and $N^\star=\emptyset$ satisfy the claim.
  Let us now assume that $|\P| \le |X|$. As $V_i \cap {\tt mark}(\T,N,X',c,\tilde{N},M') = \emptyset$, this implies that for any $(V_j,N_j) \in \P$, $V_j \cap V_i = \emptyset$.
  Thus, as $(V_i,N_i)$ was not added in the maximal packing $\P$, it implies that there exists $(V_j,N_j) \in \P$ such that $N_j \cap N_i \neq \emptyset$.
  Let $g \in N_j \cap N_i$. Observe that {\tt mark}$(\T,N,X',c-1,\tilde{N} \cup \{g\},M') \subseteq {\tt mark}(\T,N,X',c,\tilde{N},M')$, and that $(V_i, N_i \setminus \{g\})$ is a $(\T,N,X',c-1,\tilde{N} \cup \{g\},M')$-part
  such that $V_i \cap {\tt mark}(\T,N,X',c-1,\tilde{N} \cup \{g\},M') = \emptyset$. Thus, by induction hypothesis, there exists $N^\star_r \subseteq N$ and  $\P^\star_r$ such that

  \begin{itemize}
  \item  $(\T,N,X',c-1-|N^\star_r|,\tilde{N} \cup \{g\} \cup N^\star_r,M')$ satisfies the condition required by a input of {\tt mark}, and
  \item $\P^\star_r$ is a packing of $(\T,N,X',c-1-|N^\star_r|,\tilde{N} \cup \{g\} \cup N^\star_r,M')$-parts with $|\P^\star_r|=|X|+1$ and $\T^{\P^\star_r} \subseteq {\tt mark}(\T,N,X',c-1,\tilde{N} \cup \{g\},M')$.
  \end{itemize}

  Let us check that $N^\star = N^\star_r \cup \{g\}$ and $\P^\star = \P^\star_r$ satisfy the two required conditions.
  For the first one, the only non-trivial one is to prove that $X' \cup \tilde{N} \cup N^\star$ does not contain a $t$-clique.
  As $(\T,N,X',c-1-|N^\star_r|,\tilde{N} \cup \{g\} \cup N^\star_r,M')$ satisfies the condition required by an input of {\tt mark}, we get that $X' \cup (\tilde{N} \cup \{g\} \cup N^\star_r)$ does not contain a $t$-clique,
  and $X' \cup (\tilde{N} \cup \{g\} \cup N^\star_r) = X' \cup \tilde{N} \cup N^\star$.
  The second one is direct as any $(\T,N,X',c-1-|N^\star_r|,\tilde{N} \cup \{g\} \cup N^\star_r,M')$-part is also a $(\T,N,X',c-|N^\star|,\tilde{N} \cup N^\star,M')$-part.
  This concludes the inductive proof of $p_1$.

  Now, the proof of~\autoref{lemma:marking:2} follows from property $p_1$ with $c=c(\lambda,t)$ and $\tilde{N}=\emptyset$.
\end{proof}

We are now ready to define the steps that will be performed by the kernel.
\begin{definition}[Step 1: marking]\label{def:step1}
   Let $(G,X)$ be an input of $\KtHM^\lambda$, $N$ be a set of non-$K_t$-vertices of $G-X$, and $\T$ be a root of $G-X-N$.
   Let $M^0 = \emptyset$.
   For $\ell$ from $1$ to $t-1$ compute the following objects:
   \begin{itemize}
   \item $\M^\ell = \{M' \mid (|M'| \le t-2 )\mbox{,  $G[M']$ is a clique, and $M' \cap M^x \neq \emptyset$ for any $1 \le x \le \ell-1$}\} \cup \emptyset$.
   \item $M^\ell = (\bigcup_{X' \in \X, M' \in \M^\ell}$\mainMark$(\T,N,X',M')) \setminus \bigcup_{\ell' \in [\ell-1]}M^{\ell'}$.
   \end{itemize}

  Let $M(\T,N,G,X) = \bigcup_{\ell \in [t-1]}M^\ell$ be the set of marked vertices.
   For any $J \subseteq V(\T)$, we let $J(M(\T,N,G,X))=J \cap (\bigcup_{\ell \in [\ell(J,M(\T,N,G,X))]}M^\ell)$, where $\ell(J,M(\T,N,G,X))=\max\{\ell \mid J \cap M^x \neq \emptyset$ for any $1 \le x \le \ell \}$ (if $J \cap M(\T,N,G,X) = \emptyset$ then $J(M(\T,N,G,X))=\emptyset$).
   When $(\T,N,G,X)$ is clear from the context, $M(\T,N,G,X)$ is simply denoted by $M$.
\end{definition}

\begin{lemma}[running time and size of marked vertices in Step 1]\label{lemma:step1}
  Let $(G,X)$ be an input of $\KtHM^\lambda$, $N$ be a set of non-$K_t$-vertices of $G-X$, and $\T$ be a root of $G-X-N$.
  \begin{itemize}
  \item $|M(\T,N,G,X)| \le p_3^{(\lambda,t)}(|X|)$, where $p_3^{(\lambda,t)}(|X|)=(t-1)\left((t-2)|\X|p_2^{(\lambda,t)}(|X|)\right)^{(t-1)^{(t-1)}}$.
  \item Step 1 can be done in time $\O(p_4^{(\lambda,t)}(n))$, where $p_4^{(\lambda,t)}(n)=(t-1)\left((t-2)|\X|p_1^{(\lambda,t)}(n)\right)^{(t-1)^{(t-1)}}$.
  \end{itemize}

  \end{lemma}
\begin{proof}
  Let $x=|\X|$ and $p=p_2^{(\lambda,t)}(|X|)$.
  Let us prove by induction on $\ell \in [t-1]$ that $|M^\ell|  \le ((t-2)xp)^{(t-1)^{(\ell-1)}}$.
  Observe that for any $\ell$, $|M^{\ell}| \le |\M^\ell|\cdot xp \le |M^{<\ell}|^{(t-2)}xp$ where $M^{<\ell}=\bigcup_{i < \ell}M^i$.
  The result is immediate for $\ell \le 2$, so let us consider $\ell \ge 3$.
  We have
  \begin{eqnarray*}
    |M^{\ell}| &\le & |M^{<\ell}|^{(t-2)}xp \\
    &\le & \left((t-2)((t-2)xp)^{(t-1)^{(\ell-2)}}\right)^{(t-2)}xp \\
    &= &  (t-2)^{((t-1)^{\ell-2}+1)(t-2)}(xp)^{(t-1)^{(\ell-2)}(t-2)+1} \\
    &\le& (t-2)^{(t-1)^{\ell-1}}(xp)^{(t-1)^{(\ell-1)}},
  \end{eqnarray*}
  concluding the inductive proof, and leading to the claimed bound for $|M(\T,N,G,X)|$.
 The running time can be bounded using the same analysis, using $p=p_1^{(\lambda,t)}(n)$.
\end{proof}

For the next lemma we refer the reader to \autoref{fig:markingalgo}.
Informally, in this lemma we consider a vertex $v \in V(\T) \setminus M(\T,N,G,X)$ that has \emph{not} been marked in Step 1, even if there exists a ``dangerous'' set $J$
such that $J \cup \{v\}$ is a clique, and $\conf^t_{X'}(J(M(\T,N,G,X)) \cup (C(J\setminus J(M(\T,N,G,X)) \cup v) \cup N') > 0$. We show that, in this case, the marking algorithm necessarily found
a large packing.

\begin{lemma}[existence of non-marked vertex $v$ and dangerous set $J$ implies existence of large packing]\label{lemma:markingalgo}
  Let $(G,X)$ be an input of $\KtHM^\lambda$, $N$ be a set of non-$K_t$-vertices of $G-X$, $\T$ be a root of $G-X-N$, and $M(\T,N,G,X)$ as computed in \autoref{def:step1}.
 If there exist

 \begin{itemize}
 \item $X' \in \X$, $N' \subseteq N$ with $|N'| \le c(\lambda,t)$ and $X' \cup N'$ does not contain a $t$-clique, $v \in V(\T) \setminus M(\T,N,G,X)$, and
   \item $J \subseteq V(\T)$ such that $|J| \le t-2$, $J \cup \{v\}$ is a clique, and $\conf^t_{X'}(J(M(\T,N,G,X)) \cup (C(J\setminus J(M(\T,N,G,X)) \cup v) \cup N') > 0$,
 \end{itemize}

then there exist $\tilde{N} \subseteq N'$ and a packing $\P^\star$ of $(\T,N,X',c(\lambda,t)-|\tilde{N}|,\tilde{N},J(M(\T,N,G,X)))$-parts such that $|\P^\star|=|X|+1$ and $\T^{\P^\star} \subseteq M(\T,N,G,X)$. 

\end{lemma}
\begin{proof}
Let us denote by $M=M(\T,N,G,X)$.
  Let $U = \bigcup_{\ell \in [\ell(J,M)]}M^\ell$ and $S = (J \setminus J(M)) \cup \{v\}$. 
  We have that:
  \begin{itemize}
  \item $(S, N')$ is a $(\T,N,X',c(\lambda,t),\emptyset,J(M))$-part. Note that \autoref{defpart1} is true as $X' \cup N'$ does not contain a $t$-clique,
    \autoref{defpart2} is true as $(S \cup J(M)) = J \cup v$, and thus it has size at most $t-1$ and it is a clique, and \autoref{defpart3} is true
    as  $\conf^t_{X'}(J(M) \cup (C(J\setminus J(M) \cup v) \cup N') > 0$.
  \item $S \cap $ \mainMark$(\T,N,X',J(M)) = \emptyset$. Indeed, as $(J\setminus J(M)) \cap M^{\ell(J,M)+1} = \emptyset$ (recall that $\ell(J,M)$ is defined in \autoref{def:step1}, with $\ell(J,M) \le |J| \le t-2$) and $v \notin M$,
    we have $S \cap M^{\ell(J,M)+1}=\emptyset$, and by definition $S \cap U = \emptyset$.
    Thus, if by contradiction there exists $w \in $ \mainMark$(\T,N,X',J(M)) \cap S$, then as $M^{\ell(J,M)+1} \supseteq ($\mainMark$(\T,N,X',J(M)) \setminus U)$ and $w \notin U$,
    we would have $w \in M^{\ell(J,M)+1} \cap S$, a contradiction.
  \end{itemize}

  Thus, by \autoref{lemma:marking},  there exist $\tilde{N'} \subseteq N'$ and a packing $\P^\star$ of $(\T,N,X',c(\lambda,t)-|\tilde{N'}|,\tilde{N'},J(M))$-parts such that $|\P^\star|=|X|+1$ and $\T^{\P^\star} \subseteq $ \mainMark$(\T,N,X',J(M))$.
    \end{proof}

  \begin{figure}
\begin{center}
    \includegraphics[scale=0.6]{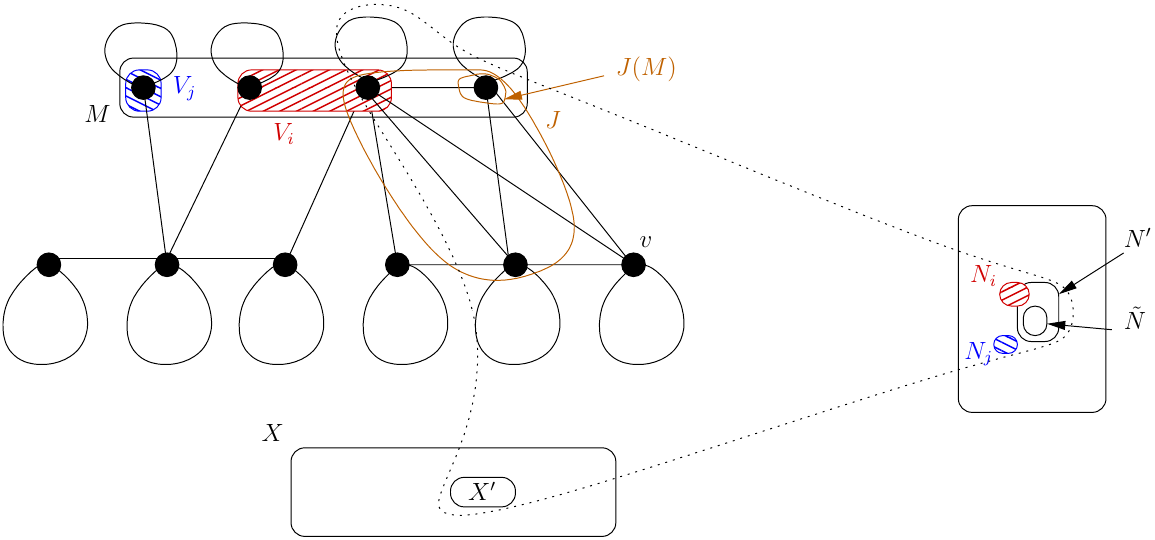}
    \caption{Setting of \autoref{lemma:markingalgo}. The set $M(\T,N,G,X)$ is denoted by $M$. The ten vertices of $V(\T)$ are depicted as black filled circles and $\P^\star=\{(V_i,N_i),(V_j,N_j)\}$.
      }
      \label{fig:markingalgo}
\end{center}
  \end{figure}

\begin{definition}[Step 2: removing a non-marked pending component]\label{def:step2}
  Let $(G,X,k)$ be an input of $\KtHMdecbis{\lambda}$, $N$ be a set of non-$K_t$-vertices of $G-X$, and $\T$ be a root of $G-X-N$.
  Let $M(\T,N,G,X)$ as computed in Step 1.
  If there exists $v \in V(\T) \setminus M(\T,N,G,X)$, let $G' = G-C(v)$ and $k'-\opt(G[C(v)])$.
\end{definition}

\begin{lemma}[running time of Step 2]\label{lemma:runinngTime}
Step 2 can be done in $\O^\star(p_0^{(\lambda,t)}(n))$, where
$p_0^{(\lambda,t)}$ is defined in \autoref{lemma:poly-opt-conf-lambda4}.
\end{lemma}
\begin{proof}
  Step 2 can be done in polynomial time, as in particular we can compute $\opt(G[C(v)])$ as according to
  \autoref{obslambdamonotonesubgraph} $\bedt(G[C(v)]) \le \bedt(G-X) \le \lambda$, and using  \autoref{lemma:poly-opt-conf-lambda4}.
\end{proof}

Before proving the safeness of Step 2, we need the following technical lemma.

\begin{lemma}\label{lemma:confXN}
  Let $(G,X)$ be an input of $\KtHM^\lambda$ (recall that $R = V(G)\setminus X$) and $N$ be a set of non-$K_t$-vertices of $G[R]$.
  For any $X' \subseteq X$, $N' \subseteq N$, $R' \subseteq R$ such that $R' \cap N' = \emptyset$ and $X' \cup N'$ does not contain a $t$-clique,
  $$\conf^t_{X' \cup N'}(R') = 0 \ \Leftrightarrow\ \conf^t_{X'}(R' \cup N') = 0.$$
\end{lemma}
\begin{proof}
  Note first that, by definition of $N'$ and $N$, $\opt(G[R' \cup N'])=\opt(G[R'])$.

We first prove the forward implication.  As $\opt(G[R'],\pr^t_{X' \cup N'}(R'))=\opt(G[R'])$, there exists an optimal solution $S$ of $G[R]$ that intersects any $Z \in \pr^t_{X' \cup N'}(R'))$
  Let us prove that $S$ is an optimal solution of $G[R' \cup N']$ that intersects any $Z \in \pr^t_{X'}(R' \cup N'))$.
  Firstly, $|S|=\opt(G[R'])=\opt(G[R' \cup N'])$.
  Secondly, as any $t$-clique $K$ of $G[R' \cup N']$ is such that $K \subseteq R'$, we deduce that $S$ intersects all $t$-cliques of $G[R' \cup N']$.
  Finally, let $Z \in \pr^t_{X'}(R' \cup N'))$. As $X' \cup N'$ does not contain a $t$-clique, we deduce that $Z \cap R' \neq \emptyset$, implying that $Z \setminus N' \in \pr^t_{X' \cup N'}(R')$
  and that $S \cap (Z \setminus N') \neq \emptyset$.

  We now turn to the proof of the backward implication.  As $\opt(G[R' \cup N'],\pr^t_{X'}(R' \cup N'))=\opt(G[R' \cup N'])$, there exists an optimal solution $S$ of $G[R' \cup N']$ that intersects any $Z \in \pr^t_{X'}(R' \cup N'))$.
  Let us prove that $S$ is an optimal solution of $G[R']$ that intersects any $Z \in \pr^t_{X' \cup N'}(R'))$.
  Firstly, $|S|=\opt(G[R' \cup N'])=\opt(G[R'])$.
  Secondly, $S$ intersects any $t$-clique of $G[R' \cup N']$, and thus any $t$-clique of $G[R']$.
  Finally, let $Z \in \pr^t_{X' \cup N'}(R'))$. By definition of $\pr^t$, there exists a $t$-clique $K$ such that $K \cap (X' \cup N') \neq \emptyset$,
  and $K \cap R' = Z$. By definition of $N$, we know that $K \cap X' \neq \emptyset$. This implies that $Z \cup (K \cap N') \in \pr^t_{X'}(R' \cup N'))$, and thus
  $S \cap (Z \cup (K \cap N')) \neq \emptyset$. Moreover, as $S$ is an optimal solution of $G[R' \cup N']$, and no vertex of $N'$ belongs to a $t$-clique in $G[R' \cup N']$,
  we get that $S \cap N' = \emptyset$. Thus, we obtain $S \cap Z \neq \emptyset$.
  \end{proof}

\begin{lemma}[safeness of Step 2]\label{lemma:step2safe}
    Let $(G,X,k)$ be an input of $\KtHMdecbis{\lambda}$, $N$ be a set of non-$K_t$-vertices of $G-X$,  $\T$ be a root of $G-X-N$, $v$ and $(G',k')$ as computed in \autoref{def:step2}.
   Then, the instances $(G,X,k)$ and $(G',X,k')$ of $\KtHMdecbis{\lambda}$ are equivalent.
 \end{lemma}
\begin{proof}
  Let us denote by $M=M(\T,N,G,X)$ and recall that  $R=V(G)-X$.
  Let us assume that $(G',X,k')$ is a \yes-instance and prove that $(G,X,k)$ is a \yes-instance as the other implication is direct.
  Let $Z'$ be a solution of $(G',X,k')$. Note that $(X \cup N)\setminus Z'$ does not contain any $t$-clique. Let $R'=R-C(v)$. Let us distinguish the following two cases.

 \paragraph*{Case 1: there exists $J \subseteq V(\T)$ such that $|J| \le t-2$, $J \cup \{v\}$ is a clique, $\conf^t_{(X \cup N)\setminus Z'}\left(J(M) \cup C((J \setminus J(M)) \cup \{v\})\right)>0$  and $Z' \cap J(M) = \emptyset$.}

\medskip

Let us prove the following ``overpay property''  $p_1: |Z' \cap R' | > \opt(G[R'])+|X|$.

By \autoref{obs:lambda4}, we have $\bedt(G[J(M) \cup C((J \setminus J(M)) \cup \{v\})]) \le \bedt(G[R]) \le \lambda$, implying by \autoref{thm:mmbslambda4} that $\mmbs_t(G[J(M) \cup C((J \setminus J(M)) \cup \{v\})]) \le \beta(\lambda,t)$.
By \autoref{lemma:smallcertif}, there exists $\overline{S} \subseteq (X \cup N)\setminus Z'$ such that  $\conf^t_{\overline{S}}(J(M) \cup C((J \setminus J(M)) \cup \{v\}))>0$ and $|\overline{S}| \le (t-1)\beta(\lambda,t)=c(\lambda,t)$. Let $X' = \overline{S} \cap X$ and $N' = \overline{S} \cap N$. By \autoref{lemma:confXN}, we get $\conf^t_{X'}(J(M) \cup C((J \setminus J(M)) \cup \{v\}) \cup N')>0$.

By \autoref{lemma:markingalgo}, there exist $\tilde{N} \subseteq N'$ and a packing $\P^\star$ of $(\T,N,X',c(\lambda,t)-|\tilde{N}|,\tilde{N},J(M))$-parts such that $|\P^\star|=|X|+1$ and $\T^{\P^\star} \subseteq M$.
Informally, the idea is that $Z'$ will overpay at least one in each part of $\P^\star$, implying $p_1$. Let us formalize this idea.

  By definition, $\P^\star=\{(V_i,N_i)\}$ is such the $V_i's$ and $N_i's$ are pairwise disjoint, and for any $i$, $\conf^t_{X'}(J(M) \cup C(V_i) \cup \tilde{N} \cup N_i)>0$.
  As $Z' \cap X' = \emptyset$, this implies that for any $i$, $|Z' \cap (J(M) \cup C(V_i) \cup \tilde{N} \cup N_i)| > \opt(G[J(M)\cup C(V_i) \cup \tilde{N} \cup N_i])$.
  Moreover, as $\tilde{N} \cup N_i \subseteq N$ and $N$ is the set of non-$K_t$-vertices of $G[R]$, we get that $\tilde{N} \cup N_i$ are still non-$K_t$-vertices of
  $G[J(M)\cup C(V_i) \cup \tilde{N} \cup N_i]$, implying $ \opt(G[J(M)\cup C(V_i) \cup \tilde{N} \cup N_i])= \opt(G[J(M)\cup C(V_i)])$.
  In addition, by definition of a part, $|V_i|+|J(M)| \le  t-1$, and by \autoref{obs:lambda4}, $J(M)$ also belongs to the set of non-$K_t$-vertices of $G[J(M)\cup C(V_i)]$,
  implying $\opt(G[J(M)\cup C(V_i)])=\opt(G[C(V_i)])$.
  As $Z' \cap \tilde{N} = \emptyset$ (because $Z' \cap N' = \emptyset$ and $\tilde{N} \subseteq N'$) and $Z' \cap J(M) = \emptyset$, $Z' \cap (J(M) \cup C(V_i) \cup \tilde{N} \cup N_i)=Z' \cap (C(V_i) \cup N_i)$.
  Thus, we get property $p_2$: for any $(V_i,N_i) \in \P^\star$, $|Z' \cap (C(V_i) \cup N_i)| > \opt(G[C(V_i)])$.

  Let us now define an appropriate partition of $R'$.
  Note that $G[(V(\T) \setminus \{v\}) \cup N]$ is also $K_t$-free in $R'$ and that for any $(V_i,N_i) \in \P^\star$, we have $V_i \subseteq V(\T) \setminus \{v\}$.
  As there is a partition $R'= \{(V(\T) \setminus \{v\}) \setminus \T^{\P^\star},  N\} \cup \{C(V_i), (V_i,N_i) \in \P^\star\}$, \autoref{obs:pendingadditive} implies that the corresponding  pending partition is additive, leading to $\opt(G[R']) = \opt(G[C((V(\T) \setminus \{v\}) \setminus \T^{\P^\star})]) + \sum_{(V_i,N_i) \in \P^\star} \opt(G[C(V_i)])$.


  On the other hand, we have:
  \begin{eqnarray*}
    |Z' \cap R'| & \ge & |Z' \cap (C((V(\T) \setminus \{v\}) \setminus \T^{\P^\star}))| + |Z' \cap (C(\T^{\P^\star} \cup N^{\P^\star}))|\\
             & \ge &  \opt(G[C((V(\T)\setminus \{v\}) \setminus \T^{\P^\star})]) + \sum_{(V_i,N_i) \in \P^\star} \opt(G[C(V_i)])+|\P^\star| \\
    & \ge & \opt(G[C((V(\T)\setminus \{v\}) \setminus \T^{\P^\star})])+ \sum_{(V_i,N_i) \in \P^\star} \opt(G[C(V_i)])+|X|+1 \\
    & \ge & \opt(G[R'])+|X|+1.
  \end{eqnarray*}
  This leads to the overpay property $p_1$.

  This property will now help us to define an appropriate solution $\tilde{Z'}$ of $G'$ which, unlike $Z'$, is easy to complete into a solution in $G$ of size at most $k$.
  Now, let us define $\tilde{Z'}=X \cup Z^\star_{R'}$, where $Z^\star_{R'}$ is an optimal solution of $G[R']$.
  We have $|\tilde{Z'}|=|X|+\opt(G[R']) \le |Z' \cap R'| \le |Z'| \le k'$. Let $Z = \tilde{Z'} \cup Z^\star_i$, where $Z^\star_i$ is an optimal solution of $G[C(V_i)]$.
  We get hat $|Z| \le k$. Moreover, $Z$ is a solution of $G$ as, in particular, for any $t$-clique $K$ of $G$ such that $K \cap C(V_i) \neq \emptyset$,
  either $K \cap X \neq \emptyset$ or $K \subseteq C(v_i)$ (according to \autoref{obs:lambda4}), and in both cases $Z \cap K \neq \emptyset$.

  \paragraph*{Case 2: for any $J \subseteq V(\T)$ such that $|J| \le t-2$, $J \cup \{v\}$ is a clique, $\conf^t_{(X \cup N)\setminus Z'}(J(M) \cup C((J \setminus J(M)) \cup \{v\}))=0$  or $Z' \cap J(M) \neq \emptyset$.}

\medskip

  In order to construct a solution $Z$ of $(G,k)$, let us prove the following properties $p_3$ and $p_4$ that assert the existence of ``good'' local solutions.
   ~\\

   Property $p_3$: there exists $Z^{\star}_{v} \subseteq C(v)$ such that
  \begin{enumerate}
  \item $|Z^{\star}_v|=\opt(G[C(v)])$ and
  \item for any $t$-clique $K' \subseteq C(v) \cup ((X \cup N) \setminus Z')$, $K' \cap Z^{\star}_v \neq \emptyset$.
  \end{enumerate}

\textit{Proof of property~$p_3$.}
  The hypothesis of Case 2 applied with $J = \emptyset$ gives us  $\conf^t_{(X \cup N)\setminus Z'}(C(v))=0$.
  By the definition of $\conf^t$, this is equivalent to
  $\opt(G[C((v)])=\opt(G[C(v)],pr^t_{(X \cup N)\setminus Z'}(C(v))$.
  Thus, we define $Z^\star_v$ as an optimal solution of $(G[C(v)],pr^t_{(X \cup N)\setminus Z'}(C(V))$.
  We have $|Z^\star_v|=\opt(G[C(v)])$. Let $K'$ be a $t$-clique with $K' \subseteq C(v) \cup ((X \cup N) \setminus Z')$. If $K' \subseteq C(u_K)$, then as $Z^{\star}_v$ is a solution of $G[C(v)]$,
  we get $K' \cap Z^{\star}_v \neq \emptyset$. Otherwise, as $K' \setminus (X \cup N) \in pr^t_{(X \cup N)\setminus Z'}(C(v))$, we know that $Z^\star_v \cap (K' \setminus (X \cup N)) \neq \emptyset$. This concludes the proof of property~$p_3$.\\

 Property $p_4$: for any $t$-clique $K$ of $G$ such that $K \cap Z' = \emptyset$ and $v \in K$, there exists $u_K \in V(\T)$  and $Z^{\star r}_{u_K}$ such that
  \begin{enumerate}
  \item $u_K \in Z^{\star r}_{u_K} \cap K$,
  \item $|Z^{\star r}_{u_K}|=\opt(G[C(u_K)])$, and
  \item for any $t$-clique $K' \subseteq C(u_K) \cup ((X \cup N) \setminus Z')$, $K' \cap Z^{\star r}_{u_K} \neq \emptyset$.
  \end{enumerate}

\textit{Proof of Property $p_4$.}
  Let $K$ be a $t$-clique as stated in $p_4$.
  Let $J^K = K \cap (V(\T) \setminus \{v\})$.


  As $G[V(\T)]$ is $K_t$-free, $|J^K| \le t-2$. Moreover, as $J^K \cup \{v\}$ is a clique,
  and $Z' \cap J^K(M) \subseteq Z' \cap K = \emptyset$, property of Case 2 applied with $J = J^K$  gives us  $\conf^t_{(X \cup N)\setminus Z'}(\tilde{V})=0$ where $\tilde{V}=J^K(M) \cup C((J^K \setminus J^K(M)) \cup \{v\})$.
  By the definition of $\conf^t$, this is equivalent to
  $\opt(G[\tilde{V}])=\opt(G[\tilde{V}],pr^t_{(X \cup N)\setminus Z'}(\tilde{V}))$.
  Let $Z^\star $ be an optimal solution of $(G[\tilde{V}],pr^t_{(X \cup N)\setminus Z'}(\tilde{V}))$.
  Note that according to \autoref{obs:lambda4}, vertices of $J^K(M)$ are non-$K_t$-vertices of $G[\tilde{V}]$, implying that $\opt(G[\tilde{V}])=\opt(G[C((J^K \setminus J^K(M)) \cup \{v\})])$,
 and by \autoref{obs:pendingadditive} that $\opt(G[\tilde{V}])=\sum_{v' \in (J^K \setminus J^K(M)) \cup \{v\}}\opt(G[C(v')])$.
  This equality implies that for any $v' \in (J^K \setminus J^K(M)) \cup \{v\}$, $|Z^\star  \cap C(v')|=\opt(G[C(v')])$, and that $Z^\star  \cap J^K(M) = \emptyset$.
  Moreover, as $K \cap R \in pr^t_{(X \cup N)\setminus Z'}(\tilde{V})$, we have that $Z^\star  \cap (K \cap R) \neq \emptyset$.
  As $K \cap R = J^K \cup \{v\} = J^K(M) \cup (J^K \setminus J^K(M)) \cup \{v\}$ and $Z^\star  \cap J^K(M) = \emptyset$, we deduce that there exists $u_K \in (J^K \setminus J^K(M)) \cup \{v\}$
  such that $u_K \in Z^\star  \cap K$. Now, define $Z^{\star r}_{u_K} = Z^\star  \cap C(u_K)$. The first two items of property $p_4$ are directly true, and let us prove the last one.
  Let $K'$ be a $t$-clique with $K' \subseteq C(u_K) \cup ((X \cup N) \setminus Z')$. If $u_K \in K'$ then we directly have $K' \cap Z^{\star r}_{u_K} \neq \emptyset$, so let us suppose that
  $u_K \notin K'$. If $K' \subseteq C(u_K)$, then as $Z^{\star r}_{u_K}$ is a solution of $G[C(u_K)]$,
  we get $K' \cap Z^{\star r}_{u_K} \neq \emptyset$. Otherwise, as $K' \setminus (X \cup N) \in pr^t_{(X \cup N)\setminus Z'}(\tilde{V})$, we know that $Z^\star  \cap (K' \setminus (X \cup N)) \neq \emptyset$.
  As $K' \setminus (X \cup N) \subseteq C(u_K)$, we get  $Z^\star  \cap (K' \setminus (X \cup N)) = Z^{\star r}_{u_K} \cap (K' \setminus (X \cup N))$, and thus the third item follows. This concludes the proof of property~$p_4$.\\

  We are now ready to define a solution $Z$ of $(G,k)$ be restructuring $Z'$ in the following way.
  Let $I = \{u_K \mid \mbox{ for $K$ a $t$-clique of $G$ such that $K \cap Z' = \emptyset$ and $v \in K$} \}$.
  We define $Z$ such that
  \begin{itemize}
  \item $Z \cap (X \cup N) = Z' \cap (X \cup N)$, and
  \item for any $u \in V(\T)$, $Z \cap C(u)$ is defined as follows:
    \begin{itemize}
    \item If $u \in I$, then $Z^{\star r}_u$.
    \item If $u \notin I$: if $u = v$ then $Z^\star _v$ and otherwise $Z' \cap C(u)$.
    \end{itemize}
  \end{itemize}

  Observe that we may have $v \in I$, and that we have property $p_5$: for any $u \in V(\T) \setminus \{v\}$, if $u \in Z'$, then  $u \in Z$.
  This property holds as either $u \in I$ and then $Z \cap C(u) = Z^{\star r}_u$, which contains $u_K$ by property $p_4$,
  or $u \notin I$ and then $Z \cap C(u) = Z' \cap C(u)$.

  Let us first bound $|Z|$. Note that for any $u \in V(\T)$, $u \neq v$, we always have $|Z \cap C(u)| \le |Z' \cap C(u)|$.
  Thus, $|Z|=|Z \cap (X \cup N)| + |Z \cup R| \le |Z' \cap (X \cup N)| + \sum_{u \in V(\T)\setminus \{v\}}|Z' \cap C(u)| + \opt(G[C(v)])=|Z'|+\opt(G[C(v)]) = k$.

  Let us now check that $Z$ is a solution of $G$.
  Let $K'$ be a $t$-clique. If $K' \cap (X \cup N) \cap Z' \neq \emptyset$, then $K' \cap Z \neq \emptyset$, so
  let us suppose that $K' \cap (X \cup N) \subseteq (X \cup N) \setminus Z'$. As $Z'$ is a solution of $G'$, we cannot have $K' \subseteq (X \cup N)$, implying $K' \cap R \neq \emptyset$.

  Suppose first that there exists $u \in V(\T)$ such that $K' \cap R \subseteq C(u)$.
  Then, if $u \in I$ then $u \in Z^{\star r}_u$ and we are done.
  Otherwise, if $u = v$ then by the second item of property~$p_3$ we get $Z^\star_v \cap K' \neq \emptyset$.
  Finally, if $u \in V(\T) \setminus (I \cup \{v\})$, then $Z \cap C(u) = Z' \cap C(u)$, and as $K' \subseteq V(G')$, we get $Z' \cap C(u) \cap K' \neq \emptyset$.

  Assume that there is no $u \in V(\T)$ such that $K' \cap R \subseteq C(u)$.
  Let $\tilde{J^K} = K' \cap R$. By \autoref{obs:lambda4}, we know that $\tilde{J^K} \subseteq V(\T)$.
  If $K' \cap Z' \neq \emptyset$, then there exists $u \in K' \cap Z'$ with $u \in V(\T)$, and by property~$p_5$ we get $K' \cap Z \neq \emptyset$.
  Otherwise, we know that $v \in \tilde{J^K}$ (as  otherwise $K' \subseteq V(G')$ and  we could not have $K' \cap Z' = \emptyset$), and thus let us write
  $K' \cap R = J^K \cup \{v\}$. Therefore, we get that $u_{K'} \in I$, and $Z \cap C(u_{K'})=Z^{\star r}_{u_{K'}}$ where $u_{K'} \in Z^{\star r}_{u_{K'}} \cap K'$ by property $p_4$.
  \end{proof}

\SetKwInput{KwData}{Input}
\SetKwInput{KwResult}{Output}


We are now ready to define the pseudo-code of the kernel.
In what follows we denote by {\tt kernel-size}$(G,k,t)$ the polynomial kernel for $K_t$-{\sc Hitting Set} parameterized by the solution size of size $\O(tk^{t-1})$ by Abu-Khzam~\cite{Abu-Khzam10}. We need the following observation.
\begin{observation}\label{obs:lambda0}
  Given an input $(G,X,k,\lambda)$ of $\KtHMdecbis{\lambda}$ where $\lambda=0$, as $\bedt(G-X)=0$ implies that $G-X$ is $K_t$-free,
  we can answer \yes if $k \ge |X|$. Thus, when $k < |X|$  {\tt kernel-size}$(G,k,t)$ of~\cite{Abu-Khzam10} provides a kernel of size $3t|X|^{t-1}$.
\end{observation}

\medskip
\begin{algorithm}[!ht]
\SetEndCharOfAlgoLine{}
 \KwData{an input $(G,X,k,\lambda)$ of $\KtHMdecbis{\lambda}$}
 \KwResult{$A(G,X,k,\lambda)$ returns an equivalent instance of size polynomial (for fixed $t$ and $\lambda$) in $|X|$ (see \autoref{thm:kernel})}

	\vspace{.2cm}

 \If{$\lambda=0$}{
return {\tt kernel-size}$(G,X,k,t)$.  \tcc*{use the kernel in the standard parameterization as in \autoref{obs:lambda0}}
}
 Compute the set $N=N^t(G-X)$ of non-$K_t$-vertices of $G-X$ and a $\bedt$-root $\T$ of $G-X-N$ using \autoref{lemma:computeRoot}.

 Let $(G',k'):=(G,k)$, $\T' = \T$.

 \While{there exists $v \in V(\T') \setminus M(\T',N,G',X)$}{\tcc*{$M(\T',N,G',X)$ as in \autoref{def:step1}}
   $G' := G'-C(v)$.

   $k' :=k'-\opt(G[C(v)])$.

   $\T' := \{T \setminus \{v\}, T \in \T'\} \setminus \emptyset$. \tcc*{remove $v$ from the unique $T \in \T'$ containing it, and if $T = \{v\}$ then remove $T$ from $\T$}

 }
 \tcc*{we have $|V(\T')|=|M(\T',N,G',X)| \le p_3^{(\lambda,t)}(|X|)$ by \autoref{lemma:step1}}

 Let $X' = X \cup V(\T')$. \label{algo:end} \tcc*{$|X'|=|X|^{\O_{\lambda,t}(1)}$ and $\bedt(G-X') \le \lambda-1$}

 Return $A(G',X',k',\lambda-1)$.

 \caption{A polynomial kernel for $\KtHMdecbis{\lambda}$}
\label{algo:kernel}
\end{algorithm}

We are now ready to prove the main theorem of this section, whose statement here is an extended version of \autoref{thm:kernellight} claimed in the introduction.
\begin{theorem}\label{thm:kernel}
  \autoref{algo:kernel} is a polynomial kernel for $\KtHMdecbis{\lambda}$ running in time $\O_{\lambda,t}(n^{\O_t(\lambda^2c(\lambda,t)^2)})$ and whose output has size $\O_{\lambda,t}(|X|^{\delta(\lambda,t)})$, where
  $\delta(\lambda,t)=(t-1)(2c(\lambda,t)+2)^{\lambda}$, and $c(\lambda,t)$ is the upper bound on the size of chunks coming from \autoref{def:chunk}.
\end{theorem}
\begin{proof}
  Let us start with the correctness.
  Let $\T$ be the $\bedt$-root of $G-X-N^t(G-X)$ computed by the algorithm.
  We claim that, during the while loop of \autoref{algo:kernel}, the following invariant remains true (where informally each loop adds one vertex $v$ to $V'$):
  for any $(G',\T',k')$ computed during the while loop, there exists $V' \subseteq V(\T)$ such that
  \begin{enumerate}
  \item\label{loop:inv1} $G'=G - C(V')$,
  \item\label{loop:inv2}  $k'=k-\sum_{v \in V'}\opt(G[C(v)])$,
   \item\label{loop:inv3} $(G,k)$ and $(G',k')$ are equivalent,
  \item\label{loop:inv4}   $V(\T') = V(\T) - V'$, and
   \item\label{loop:inv5} $N$ is a set of non-$K_t$-vertices of $G'-X$ and $\T'$ is a root of $G'-X-N$.
  \end{enumerate}

 Notice first that \autoref{loop:inv1}, \autoref{loop:inv2}, and \autoref{loop:inv4} trivially hold.
 Then, \autoref{loop:inv3} holds according to \autoref{lemma:step2safe} and the second part of \autoref{loop:inv5} holds by definition of a root.
Note that \autoref{loop:inv5} implies that $(\T',X,G',X)$ respects conditions required in \autoref{def:step1} to compute $M(\T',N,G',X)$.

Let us now prove that at the end of the while loop, $\bedt(G'-X') \le \lambda-1$.
(we use $\subseteq_i$ to denote induced subgraph in the next equation)

\begin{align*}
  G'-X' & =  G'-(V(\T') \cup X) \\
  & = G-\{C(V') \cup V(\T') \cup X\} && \text{ by~\autoref{loop:inv1}}\\
  & \subseteq_i  G-\{V' \cup V(\T') \cup X\} && \text{as $V' \subseteq C(V')$} \\
  & =  G-\{V(\T) \cup X\} &&\text{ by~\autoref{loop:inv4}.}
\end{align*}

Thus, using \autoref{obslambdamonotonesubgraph} we get $\bedt(G'-X') \le \bedt(G-\{V(\T) \cup X\})$.
Moreover,
\begin{align}
  \bedt(G-\{V(\T) \cup X\}) &\le \bedt(G-\{V(\T) \cup X \cup N^t(G-X)\}) \\
    & <   \bedt(G-\{X \cup N^t(G-X)\}) \\
  & =  \bedt(G-X) = \lambda.
\end{align}
Where inequality (2) holds as $N^t(G-X)$ is still a set of non-$K_t$-vertices in $G-\{V(\T) \cup X\}$,
and (3) as $\T$ is a $\bedt$-root of $G-X-N^t(G-X)$.

Thus, we get the claimed property that, at the end of the while loop, $\bedt(G'-X') \le \lambda-1$, and thus by using \autoref{loop:inv3} and a straightforward induction on $\lambda$ we get
that \autoref{algo:kernel} is correct.

Let us analyze the size of the kernel.
  For any $t \ge 3$, let $p_3^{(0,t)}(x)=3tx^{t-1}$ and
  $f^{(\lambda,t)}(x)=p_3^{(0,t)}(2p_3^{(1,t)}(\dots(2p_3^{(\lambda,t)}(x))\dots))$.
  Let us show by induction on $\lambda$ that for any input $(G,X,k,\lambda)$, $A(G,X,k,\lambda)$ returns a subset of vertices of $G$ of size $f^{(\lambda,t)}(|X|)$.
  The result is clear for $\lambda=0$ by \autoref{obs:lambda0}, and given $\lambda \ge 1$, by induction hypothesis $A(G,X,k,\lambda)$ returns a subset of vertices of size
  $f^{(\lambda-1,t)}(|X|+p_3^{(\lambda,t)}(|X|))\le f^{(\lambda-1,t)}(2p_3^{(\lambda,t)}(|X|)) = f^{(\lambda,t)}(|X|)$.

  Now, there exists constants $\alpha_2$ and $\alpha_3$,
  depending on $t$ and $\lambda$, such that
  \begin{itemize}
  \item $p_2^{(\lambda,t)}(x)=\alpha_2x^{\delta_2}$ where $\delta_2=c(\lambda,t)+2$, and
  \item $p_3^{(\lambda,t)}(x)=\alpha_3x^{\delta_3}$ where $\delta_3=\delta_2+c(\lambda,t)=2c(\lambda,t)+2$.
  \end{itemize}

  Thus, as $f^{(\lambda,t)}(x)\le p_3^{(0,t)}(2p_3^{(\lambda,t)}(\dots(2p_3^{(\lambda,t)}(x))\dots))$, we get that there exists a constant $\alpha_4$
  such that $f^{(\lambda,t)}(x)\le \alpha_4x^{(t-1)\delta_3^{\lambda}}$.

  Finally, let us analyze the running time.
  Let $r(n,\lambda)$ be the upper bound on complexity to any call of kernel algorithm with input $(G,X,k,\lambda)$, where $|V(G)|=n$.
  Let us first prove that there exists $f^{(\lambda,t)}(n)=\O_{\lambda,t}(n^{\O_t(\lambda^2c(\lambda,t)^2)})$ such that $r(n,\lambda)=f^{(\lambda,t)}(n)+r(n,\lambda-1)$, leading to the claimed complexity.

Before starting the analysis, notice that polynomials $p_0$, $p_1$ and $p_4$ defined respectively in \autoref{lemma:computeRoot}, \autoref{lemma:marking}, and \autoref{lemma:step1} are such that:
  \begin{itemize}
      \item $p_0^{(\lambda,t)}(n)=n^{\O_t(\lambda^2)}$,
      \item $p_1^{(\lambda,t)}(n)=\left(n^{t+c(\lambda,t)+\O(1)}c(\lambda,t)p_0^{(\lambda,t)}(n)\right)^{c(\lambda,t)+1}=\O_{\lambda,t}(n^{\O_t(\lambda^2c(\lambda,t)^2)})$, and
      \item $p_4^{(\lambda,t)}(n)=(t-1)\left((t-2)|\X|p_1^{(\lambda,t)}(n)\right)^{(t-1)^{(t-1)}}=\O_{\lambda,t}(n^{\O_t(\lambda^2c(\lambda,t)^2)})$.
  \end{itemize}

 Now, consider a call with input $(G,X,k,\lambda)$. Let $f^{(\lambda,t)}(n)$ be the complexity of all operations performed before the recursive call.
  One can observe that $f^{(\lambda,t)}(n)$ can be bounded by $\O^\star(n^t)$ (to compute $N^t(G-X)$), plus $n^{\O_t(\lambda)}$ (to compute a $\bedt$-root using \autoref{lemma:computeRoot}), plus $n(p_4^{(\lambda,t)}(n)+p_0^{(\lambda,t)}(n))$ (for the while loop).

 This leads to
 \begin{eqnarray*}
     f^{(\lambda,t)}(n) & = & \O^\star(n^t)+n^{\O_t(\lambda)}+n\left(\O_{\lambda,t}(n^{\O_t(\lambda^2c(\lambda,t)^2)})+n^{\O(\lambda^2)}\right) \\
      & = & \O_{\lambda,t}(n^{\O_t(\lambda^2c(\lambda,t)^2)}).
 \end{eqnarray*}
\end{proof}

To conclude this section, we would like to observe that, by our hardness results (\autoref{thm:hardness-lambda1} and \autoref{thm:hardness-treedepth}), the kernelization algorithm that we have presented \textit{needs} to exploit the fact that the graph $H$ that we want to hit is a clique. In particular, the following lemmas of the previous section about upper-bounding the size of minimal blocking sets do not hold anymore as soon as
$H$ is not a clique, and thus (see the discussion before \autoref{thm:hardness-lambda1}) when ``a set of two isolated vertices can be a part of blocking set". In particular,  \autoref{lemma:bszooming} fails as there is no reason than we can zoom in one $C(v)$, and \autoref{lemma:bscleanCC} fails as we no longer have $\B = \bigcup_{i \in [c]}(\B \cap C_i)$. Also, another place where we exploit that $H$ is a clique is in the proof of \autoref{lemma:mmbstd} to obtain the improved bound in graphs of bounded treedepth.

\section{Computing a $\bedt$-root and and an optimal solution in graphs with bounded $\bedt$}
\label{sec:computeHittingSet}

In this section we present some results for graphs where the parameter $\bedt$ is bounded by a constant, and which are needed in our kernelization algorithm. Namely, in \autoref{sec:computing-root} we show (\autoref{lemma:computeRoot}) how to compute a $\bedt$-root of a graph~$G$ with~$\bedt(G) \leq \lambda$ in polynomial time for fixed $t$ and $\lambda$. We need this ingredient in the kernel where we recursively call it with decreasing values of $\lambda$.  This algorithm is in fact a byproduct of an \XP-algorithm (\autoref{lem:compute:bedt}) to compute~$\bedt$ with the natural parameterization.

\smallskip

Then, in \autoref{sec:computing-optimum-bounded-lambda} we present an algorithm (\autoref{lemma:poly-opt-conf-lambda4}) that computes in polynomial time, again for fixed $t$ and $\lambda$, an optimum solution of $K_t$-{\sc Subgraph Hitting} in graphs with bounded $\bedt$, and that decides whether  an instance $(G,\F)$ of $\EKtH$ is clean. These ingredients are crucially used in the kernelization algorithm to compute optimal solutions in the pending components $C(v)$ considered by the algorithm, and for deciding whether a conflict is positive.

\medskip

For the sake of readability, in this section we use self-contained notation for the annotated problem that we consider.
Let~$\H$ be a collection of subgraphs of a graph~$G$. For an integer~$t$, define \[\sol_t(G, \H) = \{S \subseteq V(G) \mid G-S \text{ is $K_t$-free and $S \cap V(H) \neq \emptyset$ for each~$H \in \H$}\}.\]
For~$v \in V(G)$, we additionally define \[\optsolwith_t(G,\H,v) = \{S \in \optsol_t(G,\H) \mid v \in S\}.\]
Let~$\opt(G,\H) = \min _{S \in \sol_t(G,\H)} |S|$ and~$\optsol_t(G,\H) = \{S \in \sol_t(G,\H) \mid |S| = \opt(G,\H)\}$. We define~$\opt(G) = \opt(G,\emptyset)$, and similarly let~$\sol_t(G) = \sol_t(G,\emptyset)$ and~$\optsol_t(G) = \optsol_t(G,\emptyset)$.

For a collection~$\H$ of subgraphs of~$G$ and a subset of vertices~$V'$ of~$G$, we denote by~$\H \cap V'$ the collection of subgraphs~$H \in \H$ which are fully contained in~$V'$.

\subsection{Computing a $\bedt$-root}
\label{sec:computing-root}

In this section we show that for fixed~$\lambda$, we can compute a $\bedt$-root for a graph~$G$ with~$\bedt(G) \leq \lambda$ and~$N^t(G) = \emptyset$ in polynomial time (recall that $N^t(G)$ is the set of non-$K_t$-vertices of $G$). It is a byproduct of an \XP-algorithm to compute~$\bedt$ with the natural parameterization. Before presenting that algorithm, we give some structural lemmas about the behavior of roots.

\begin{lemma} \label{lemma:root:all:or:one}
Let~$G$ be a connected undirected graph and let~$T$ be a vertex set such that each connected component of~$G-T$ is adjacent to exactly one vertex of~$T$. For each biconnected component~$F$ of~$G$ on at least three vertices, the following holds: if~$|T \cap F| > 1$, then~$F \subseteq T$.
\end{lemma}
\begin{proof}
Suppose for a contradiction that~$T$ contains more than one, but not all, vertices of a biconnected component~$F$ on at least three vertices. Let~$u \in F \setminus T$ and consider the connected component~$C$ of~$G-T$ that contains~$u$. By the stated assumption on~$T$, the component~$C$ has exactly one neighbor~$t \in T$. Let~$r \in (F \cap T) \setminus \{t\}$ be another vertex from~$T$ inside the biconnected component~$F$. Since~$F$ is biconnected and has at least three vertices, it is 2-connected. Hence, by Menger's theorem, there are two internally vertex-disjoint paths~$P_1, P_2$ from~$u$ to~$r$. At least one of these paths, say~$P_1$, does not visit vertex~$t$. Now consider the first time that~$P_1$, which starts in~$u \notin T$ and ends in~$T$, visits a vertex~$t' \in T$. We have~$t' \neq t$ as~$t$ does not lie on~$P_1$. Since the prefix of~$P_1$ consisting of the subpath until reaching~$t'$ is disjoint from~$T$, the prefix is contained entirely inside one connected component of~$G-T$. That component also contains~$u$ and is therefore~$C$. But this shows that component~$C$ is adjacent to at least two distinct vertices of~$T$, namely~$t$ and~$t'$; a contradiction to our assumption on~$T$.
\end{proof}

\begin{lemma} \label{lemma:union:roots}
Let~$G$ be a connected undirected graph and let~$T_1, T_2$ be vertex sets such that for each~$i \in [2]$, the graph~$G[T_i]$ is connected and each connected component of~$G-T_i$ is adjacent to exactly one vertex of~$T_i$. If~$T_1 \cap T_2 \neq \emptyset$, then each connected component of~$G - (T_1 \cup T_2)$ is adjacent to exactly one vertex of~$(T_1 \cup T_2)$.
\end{lemma}
\begin{proof}
Note that since~$G$ is connected, each component of~$G - (T_1 \cup T_2)$ is adjacent to \emph{at least} one vertex of~$T_1 \cup T_2$. We argue that it cannot be adjacent to more. Assume for a contradiction that~$C$ is a connected component~$C$ of~$G - (T_1 \cup T_2)$ that is adjacent to two distinct vertices~$t_1, t_2 \in T_1 \cup T_2$. If~$t_1, t_2$ belong to a common set~$T_i$, then in~$G - T_i$ there is a connected component containing~$C$, and it would be adjacent to both~$t_1, t_2 \in T_i$; this would contradict our assumption on~$T_i$. Hence we may assume that~$t_1 \in T_1$ and~$t_2 \in T_2$. Since~$T_1 \cap T_2 \neq \emptyset$ and the subgraph~$G[T_i]$ is connected for each~$i$, there is a simple path~$P$ from~$t_1$ to~$t_2$ in~$G$ whose edges are contained in~$E(G[T_1]) \cup E(G[T_2])$; such a path can be obtained by combining a path from~$t_1$ to~$T_1 \cap T_2$ with a path from~$t_2$ to~$T_1 \cap T_2$ and trimming repetitions. Since~$t_1 \notin T_2$ and~$t_2 \notin T_1$, while the path is confined to edges that belong to some~$E(G[T_i])$, path~$P$ uses at least one third vertex~$t_{12} \in T_1 \cap T_2$. Since component~$C$ is connected and adjacent to both~$t_1$ and~$t_2$, there is a path~$Q$ in~$C$ from a neighbor of~$t_1$ to a neighbor of~$t_2$, and~$P \cup Q$ form a simple cycle~$D$ on at least three vertices. Let~$c$ be a vertex on the cycle~$D$ which is contained in component~$C$. The cycle witnesses that~$t_1, t_2, t_{12}$, and~$c$ are part of a common biconnected component~$F$ of~$G$, with~$F \supseteq V(D)$. We use it to derive a contradiction: the set~$T_1$ contains at least two vertices from the biconnected component~$F$, but it does not contain a third vertex~$c \in F$; a contradiction to \autoref{lemma:root:all:or:one}.
\end{proof}


\begin{lemma} \label{lem:root:candidates}
Let~$t \geq 3$ be an integer and let~$H := K_t$. Let~$G$ be a connected graph with~$N^t(G) = \emptyset$. Let~$G'$ be the subgraph of~$G$ consisting of the union of the biconnected components of~$G$ which are~$K_t$-free. Then for each connected component~$T_i$ of~$G'$, the set~$T_i$ is a root of~$G$. Furthermore, one of the following holds:
\begin{enumerate}
    \item There is a vertex~$v \in V(G)$ with~$\bedt(G - v) < \bedt(G)$.
    \item There is a connected component~$T_i$ of~$G'$ for which~$\bedt(G - T_i) < \bedt(G)$, so that~$T_i$ is a $\bedt$-root of~$G$.
\end{enumerate}
\end{lemma}
\begin{proof}
We first establish that for each biconnected component~$F$ of~$G$, if~$G[F]$ is~$H$-free then~$F$ is a root of~$G$. Since~$G[F]$ is trivially connected, it suffices to argue that each connected component of~$G-F$ is adjacent to exactly one vertex of~$F$. Connectivity of~$G$ ensures a component cannot be adjacent to fewer than one. If~$C$ is a connected component of~$G-F$ that is adjacent to two distinct vertices~$u,v$ of~$F$, then~$F$ is not a biconnected component (inclusion-maximal connected subgraph without a cutvertex): the subgraph obtained from~$F$ by extending it with an edge from~$u$ to a neighbor in~$C$, a path within~$C$ to a neighbor of~$v$, and an edge to~$v$, is again biconnected and a strict supergraph of~$F$; a contradiction.

The preceding paragraph shows that for each $H$-free biconnected component~$F$ of~$G$, the set~$F$ is a root of~$G$. If two roots~$T_1, T_2$ of~$G$ share a vertex, then~$T_1 \cup T_2$ is also a root of~$G$: the graph~$G[T_1 \cup T_2]$ is connected since each part is connected and they share a vertex; by \autoref{lemma:union:roots} the removal of~$T_1 \cup T_2$ leaves components that each see exactly one vertex of~$T_1 \cup T_2$. It remains to argue that~$G[T_1 \cup T_2]$ is $H$-free. This follows from the fact that~$H$ is biconnected and has at least three vertices. Hence any occurrence of~$H$ as a subgraph of~$G$ is contained within a single biconnected component of~$G$ on at least three vertices. If a copy of~$H$ is contained in a single biconnected component~$F$ of~$G[T_1 \cup T_2]$. Since~$H$ has at least three vertices, at least two vertices of~$H$ belong to a common set~$T_i$ for some~$i \in [2]$. The fact that~$F$ is a biconnected component of~$G[T_1 \cup T_2]$ implies that~$F$ is a subgraph of a biconnected component~$F'$ of~$G$. Since~$T_i$ contains at least two vertices of~$F'$, by \autoref{lemma:root:all:or:one}, the set~$T_i$ contains all vertices of~$F'$; but then~$G[T_i]$ contains a copy of~$H$, which is a contradiction.

We have now established that for each~$H$-free biconnected component~$F$ of~$G$, the set~$F$ is a root of~$G$; and that the union of two intersecting roots of~$G$ is itself a root. Note that for each connected component~$C$ of the graph~$G'$ defined in the lemma statement, the vertex set of~$C$ can obtained by repeatedly taking the union of two intersecting~$H$-free biconnected components of~$G$. If follows that for each connected component~$T_i$ of~$G'$, the set~$T_i$ is a root of~$G$.

To complete the proof of the lemma, we show that if no single vertex~$v \in V(G)$ has the property that~$\bedt(G-v) < \bedt(G)$, then there is a connected component~$T_i$ of~$G'$ for which~$\bedt(G-T_i) < \bedt(G)$. We argue as follows. Since~$G$ is connected and~$N^t(G) = \emptyset$, by \autoref{def:pending} there is a root~$T$ such that~$\bedt(G) = 1 + \bedt(G-T)$. Since we assumed no single vertex has this property, we have~$|T| \geq 2$.

\begin{claim}
There is a connected component~$T_i$ of~$G'$ such that~$T \subseteq T_i$.
\end{claim}
\begin{claimproof}
Recall that the biconnected components of a graph partition its edge set. Hence for each edge~$e$ of~$G[T]$, some biconnected component~$F_e$ of~$G$ contains~$e$. We argue that~$F_e \subseteq T$. If~$|F_e| = 2$ then this is trivial since~$e$ was an edge of~$G[T]$. If~$|F_e| \geq 3$, then the fact that~$T$ contains both endpoints of~$e$ implies by \autoref{lemma:root:all:or:one} that~$T$ contains all vertices of~$F_e$. Since~$G[T]$ is~$H$-free, the biconnected component~$F_e$ is~$H$-free. So for each~$e \in G[T]$,~$G[F_e]$ is a subgraph of~$G'$, with~$G'$ as defined in the lemma statement. Since~$G[T]$ is connected, this implies that there is a connected component~$T_i$ of~$G'$ that contains all edges~$e$ of~$G[T]$. Since~$|T| \geq 2$, this implies that~$T \subseteq T_i$.
\end{claimproof}

Using the claim, we now prove the last part of the lemma statement: there exists a component~$T_i$ such that~$\bedt(G) = 1 + \bedt(G - T) \leq 1 + \bedt(G - T_i)$, where the last inequality follows from \autoref{obslambdamonotonesubgraph}. It follows that~$\bedt(G - T_i) \leq \bedt(G) - 1$, which completes the proof.
\end{proof}

With these ingredients, we can now give an \XP-algorithm algorithm to compute~$\bedt(G)$.

\begin{lemma} \label{lem:compute:bedt}
Let~$t \geq 3$ be an integer and let~$H := K_t$. There is an algorithm that, given a graph~$G$ and integer~$\lambda \geq 0$, determines whether or not~$\bedt(G) \leq \lambda$ in time~$n^{\Oh_t(\lambda)}$.
\end{lemma}
\begin{proof}
We first present the algorithm and then analyze it. The algorithm is recursive and proceeds as follows.

\subparagraph*{Base case: $V(G) = \emptyset$.} In this case we have~$\bedt(G) = 0 \leq \lambda$ and therefore the algorithm outputs $\mathsf{true}$.

\subparagraph*{Base case: $\lambda = 0$.} If~$\lambda = 0$, then we output $\mathsf{true}$ if and only if~$G$ is~$K_t$-free, which can be determined in time~$n^{\Oh(t)} \in n^{\Oh_t(1)}$.

\subparagraph*{Step case: $V(G) \neq \emptyset$.} The algorithm considers the following cases.
\begin{enumerate}
    \item If~$G$ is disconnected, with connected components~$C_1, \ldots, C_m$, we recursively run the algorithm on~$(C_i, \lambda)$ for each~$i \in [m]$ and return the disjunction of the resulting values.
    \item If~$G$ is connected, then we compute the set~$N^t(G)$, which can be done in time~$n^{\Oh(t)} \in n^{\Oh_t(1)}$.
    \begin{enumerate}
        \item If~$N^t(G)$ is nonempty, then we return the result of the recursive call on~$(G - N^t(G), \lambda)$.
        \item If~$N^t(G)$ is empty, then we make a sequence of recursive calls.
        \begin{enumerate}
            \item For each vertex~$v \in V(G)$ we recurse on~$(G - v, \lambda-1)$.
            \item Let~$G'$ be the subgraph of~$G$ consisting of the union of the biconnected components of~$G$ which are~$K_t$-free. For each connected component~$T_i$ of~$G'$, we recurse on~$(G - T_i, \lambda - 1)$.
        \end{enumerate}
        If at least one recursive call returns $\mathsf{true}$, then we return $\mathsf{true}$; otherwise we return $\mathsf{false}$.
    \end{enumerate}
\end{enumerate}
This concludes the description of the algorithm.

\subparagraph*{Running time.} We prove the algorithm runs in time~$n^{\Oh_t(\lambda)}$. Towards that purpose, we analyze the recursion tree generated by the algorithm.

We first bound the depth of the recursion tree. The value of~$\lambda$ never decreases. After every three successive levels of recursion, the value of~$\lambda$ has strictly decreased. To see this, note that the only recursive calls which do not decrease~$\lambda$ are those where we split on connected components, or remove the vertices of~$N^t(G)$. Recursive calls in which we split in components cannot be consecutive, since each call it generates is for a connected graph. If we recurse by removing~$N^t(G)$, then the graph on which we recurse has all vertices in a copy of~$K_t$ (but may be disconnected). After possibly recursing one more level on the connected components (which preserves the fact that all vertices occur in a copy of~$K_t$), the graph under consideration is connected and has all its vertices in a copy of~$K_t$. Hence after those two recursive calls that do not decrease~$\lambda$, the next call either terminates the algorithm or leads to a decrease of~$\lambda$ by recursing on~$G-v$ or~$G-T_i$. Since the algorithm terminates once~$\lambda = 0$, the recursion depth is indeed bounded by~$\Oh(\lambda)$.

The branching factor of the algorithm is~$\Oh(n)$: it either recurses once for each component, or it recurses once for each vertex and once for each of the at most~$n$ choices for~$T_i$.

The total number of nodes in the recursion tree is bounded by the branching factor raised to the power depth, which is therefore~$n^{\Oh(\lambda)}$. Since the running time per call is~$n^{\Oh_t(1)}$, the algorithm runs in time~$n^{\Oh_t(\lambda)}$, as required.

\subparagraph*{Correctness.} Correctness in the base cases is trivial; the same goes for the step case in which~$G$ is connected or~$N^t(G) \neq \emptyset$. For the case that~$G$ is connected and~$N^t(G) = \emptyset$, if~$\bedt(G) \leq \lambda$ then by \autoref{lem:root:candidates} the algorithm makes a recursive call on an induced subgraph~$G^*$ with~$\bedt(G^*) < \lambda$. Using induction on the recursion depth, we may assume that the recursive call is correct and reports that~$\bedt(G^*) \leq \lambda-1$, which leads the algorithm to conclude that~$\bedt(G) \leq \lambda$. Since the algorithm only reports~$\bedt(G) \leq \lambda$ when it has found a connected~$K_t$-free vertex set whose removal yields a graph~$G^*$ with~$\bedt(G^*) \leq \lambda - 1$, the algorithm is correct when it concludes that~$\bedt(G) \leq \lambda$.

This concludes the proof.
\end{proof}

Using the previous algorithm to compute~$\bedt(G)$ and \autoref{lem:root:candidates}, we can now derive an algorithm to compute a~$\bedt$-root of a graph.

\begin{lemma}\label{lemma:computeRoot}
Let~$t \geq 3$ be fixed and~$H := K_t$. There is an algorithm that, given a graph~$G$ and integer~$\lambda$ such that~$\bedt(G) \leq \lambda$ and $N^t(G)=\emptyset$, outputs a $\bedt$-root~$\T$ of~$G$ in time~$n^{\Oh_t(\lambda)}$.
\end{lemma}
\begin{proof}
We first prove the statement for the case that~$G$ is connected. Using \autoref{lem:compute:bedt}, we can compute~$\bedt(G)$ in time~$n^{\Oh_t(\lambda)}$ and store it as~$\lambda^*$. For each~$v \in V(G)$, we invoke \autoref{lem:compute:bedt} again on~$G-v$. If we find a vertex for which~$\bedt(G-v) < \lambda^*$, we terminate with the $\bedt$-root~$v$ as output.

Suppose now that no choice of~$v$ causes the algorithm to terminate. We compute the graph~$G'$ consisting of the union of the $H$-free biconnected components of~$G$, in time~$n^{\Oh_t(1)}$. For each connected component~$T_i$ of~$G'$, we invoke \autoref{lem:compute:bedt} on the graph~$G - T_i$. The lemma guarantees that each set~$T_i$ is a root of~$G$. Since no single vertex formed a $\bedt$-root of~$G$, the lemma also guarantees that~$\bedt(G - T_i) < \bedt(G)$ for at least one such choice of~$T_i$; then we give~$T_i$ as the output of the algorithm. This concludes the proof for the case that~$G$ is connected. Note that the algorithm runs in time~$n^{\Oh_t(\lambda)}$ since it performs~$\Oh(n)$ calls to a subroutine with running time~$n^{\Oh_t(\lambda))}$, and additionally spends~$n^{\Oh_t(1)}$ time computing~$G'$.

Now suppose~$G$ has multiple connected components. Since the forbidden graph~$H$ is connected, a $\bedt$-root~$\T$ for a disconnected graph~$G$ can be obtained by selecting a $\bedt$-root for each connected component: since the removal of~$V(\T)$ decreases the $\bedt$ value of each component, by \autoref{def:pending} the removal of~$V(\T)$ also decreases $\bedt(G)$. If~$N^t(G) = \emptyset$, then as each copy of~$H$ is contained in a single connected component, we also have~$N^t(C_i) = \emptyset$ for each connected component~$C_i$ of~$G$. Hence it suffices to run the algorithm for connected graphs described above and take the union of the computed roots.
\end{proof}

\subsection{Computing an optimal solution}
\label{sec:computing-optimum-bounded-lambda}

Before presenting the algorithm, we give the following lemma which will be useful in its analysis. The lemma roughly says that if we have a connected graph~$G$ with a root~$T$ with respect to a forbidden subgraph~$K_t$, then in any optimal solution~$U \in \optsol_t(G,\H)$ we can replace the intersection of~$U$ with a pending component~$C(v)$ (\autoref{def:pending}) by a solution from~$\optsol_t(C(v), \H \cap C(v))$ to obtain a new optimal solution; we just need to ensure to replace it by a solution to the smaller problem that contains the attachment vertex~$v$, if there is an optimal solution doing so. We will use the lemma to later prove the correctness of an algorithmic strategy that solves the problem on~$G$ by separately solving it on pending components of the root~$T$, testing for each~$v \in T$ whether there is an optimal solution that contains the attachment point.

\begin{lemma} \label{lemma:restructure:root}
Let~$H = K_t$ for some constant~$t \geq 3$. Let~$T \subseteq V(G)$ be a root of a connected graph~$G$ and let~$\H$ be a collection of complete subgraphs of~$G$, each having less than~$t$ vertices. Let~$U \in \optsol_t(G,\H)$. For each~$v \in T$, the following holds.
\begin{enumerate}
    \item If~$\optsolwith_t(C(v), \H \cap C(v), v) \neq \emptyset$, then for each set~$U^+ \in \optsolwith_t(C(v), \H \cap C(v), v)$ we have~$(U \setminus V(C(v)) \cup U^+ \in \optsol_t(G,\H)$.
    \item If~$\optsolwith_t(C(v), \H \cap C(v), v) = \emptyset$, then for each set~$U^- \in \optsol_t(C(v), \H \cap C(v))$ (note that~$v \notin U^-$) we have~$(U \setminus (V(C(v)) \setminus \{v\})) \cup U^- \in \optsol_t(G,\H)$.
\end{enumerate}
\end{lemma}
\begin{proof}
Assume first that $U^+ \in \optsolwith_t(C(v), \H \cap C(v), v)$, so that the set of such optimal solutions is not empty. Define~$U^* := (U \setminus V(C(v)) \cup U^+$; our goal is to show~$U^* \in \optsol_t(G,\H)$.

We first argue that~$|U^*| \leq \opt(G,\H)$. This follows from the fact~$U^+ \in \optsol_t(C(v), \H \cap C(v))$, so~$U^+$ is an optimal choice of vertices to hit the copies of~$K_t$ in~$C(v)$ together with those graphs in~$\H$ which are fully contained in~$C(v)$, while the set~$U \cap C(v)$ also hits all copies of~$K_t$ in~$C(v)$ and all graphs of~$\H$ fully contained in~$C(v)$. Hence replacing the latter by the former does not yield an increase in size, and~$|U| = \opt(G,\H)$ by definition.

Now we argue that~$U^* \in \sol_t(G,\H)$. Since~$U \in \sol_t(G,\H)$ and we obtain~$U^*$ by replacing~$U \cap C(v)$ with~$U^+$, it suffices to argue that any copy of~$K_t$ or~$H \in \H$ that is hit by~$U \cap C(v)$, is also hit by~$U^+$. By definition of pending component, any complete subgraph of~$G$ is either contained fully in~$C(v)$, or disjoint from~$C(v) \setminus \{v\}$. Since~$v \in U^+$ and the set~$U^+$ hits all~$K_t$-subgraphs in~$C(v)$ and all subgraphs of~$\H \cap C(v)$, we conclude~$U^* \in \sol_t(G,\H)$. Together with the previous paragraph, this yields~$U^* \in \optsol_t(G,\H)$ as desired.

The proof of the second part is similar. So assume~$\optsolwith_t(C(v), \H \cap C(v), v) = \emptyset$ and let~$U^- \in \optsol_t(C(v), \H \cap C(v))$. Define~$U^{**} := (U \setminus (V(C(v)) \setminus \{v\})) \cup U^-$, so it is effectively obtained by replacing the choices of~$U$ in the interior of the pending component by the solution~$U^-$. By this definition, if~$v \in U$ then it remains in~$U$.

We argue that~$|U^{**}| \leq \opt(G,\H)$. If~$v \notin U$, then the argument is identical to the one given for the first case. So assume~$v \in U$. Since~$\optsolwith_t(C(v), \H \cap C(v), v) = \emptyset$, while it is easy to see that~$U \cap C(v) \in \sol_t(C(v), \H \cap C(v))$, we derive that~$|U \cap C(v)| > \opt(C(v), \H \cap C(v))$, which implies that~$|U \cap (C(v) \setminus \{v\})| \geq \opt(C(v), \H \cap C(v))$. Hence the number of vertices inserted into~$U^{**}$ during the replacement is at most as large as the number we remove, which proves that~$|U^{**}| \leq \opt(G,\H)$.

To complete the proof, it suffices to establish that~$U^{**} \in \sol_t(G,\H)$. If~$v \notin U$, then this follows analogously to the first case of the lemma. If~$v \in U$, then since it is not replaced we have~$v \in U^{**}$, which therefore hits all copies of~$K_t$ or~$H \in \H$ that contain~$v$. All remaining forbidden subgraphs either live fully in~$C(v)$ and are therefore hit by~$U^-$, or are disjoint from~$C(v)$ and are therefore hit by~$U \setminus C(v)$. It follows that~$U^{**} \in \sol_t(G, \H)$.
\end{proof}

The following technical lemma will also be useful when analyzing the algorithm. It says that the number of vertices~$v$ that a solution~$U \in \optsol_t(G,\H)$ picks from a root~$T$, but for which the subinstance~$(C(v), \H \cap C(v))$ does not have an optimal solution containing~$v$, is bounded in terms of the difference~$\opt(G,\H) - \opt(G)$.

\begin{lemma} \label{lemma:overpay:in:root}
Let~$H = K_t$ for some constant~$t \geq 3$. Let~$T \subseteq V(G)$ be a root of a connected graph~$G$ and let~$\H$ be a collection of complete subgraphs of~$G$, each having less than~$t$ vertices. Define~$T' := \{v \in T \mid \optsolwith_t(C(v), \H \cap C(v), v) \neq \emptyset\}$, and let~$U \in \optsol_t(G,\H)$. Then~$|U \cap (T \setminus T')| \leq \opt(G, \H) - \opt(G)$.
\end{lemma}
\begin{proof}
Consider a vertex~$v \in U \cap (T \setminus T')$. By definition of~$T'$, the set~$\optsolwith_t(C(v), \H \cap C(v), v)$ is empty: no solution that hits the $K_t$-subgraphs of~$C(v)$ and the subgraphs of~$\H$ fully contained in~$C(v)$, contains~$v$. Since the set~$U \cap C(v)$ hits both types of subgraphs and contains~$v$, we have~$|U \cap C(v)| > \opt(C(v), \H \cap C(v)) \geq \opt(C(v))$.

We can now analyze the global situation as follows. It is easy to see that for each~$v \in T$ we have~$U \cap C(v) \geq \opt(C(v))$; we just argued that for~$v \in U \cap (T \setminus T')$ we have a strict inequality. Hence we find: \[\opt(G,\H) = |U| \geq \left( \sum_{v \in T} \opt(C(v)) \right) + |U \cap (T \setminus T')|.\] On the other hand, by definition of root, no vertex of~$T$ belongs to a~$K_t$-subgraph in~$G$. Hence the edges within~$G[T]$ can be removed from the graph without changing the set of~$K_t$-subgraphs and thereby without changing~$\opt$. Since the connected components of the resulting graph are exactly the subgraphs~$C(v)$ for~$v \in T$, it follows that: \[\opt(G) = \sum_{t \in T} \opt(C(v)).\]

Subtracting the second equation from the first proves the lemma.
\end{proof}

\begin{lemma}\label{lemma:OPTforboundedlambda}
For each constant~$t \in \mathbb{N}$ with~$t \geq 3$, there is an algorithm with the following specifications.
\begin{itemize}
    \item The input consists of a graph~$G$, a collection~$\H$ of complete subgraphs of~$G$ having less than~$t$-vertices each, and integers~$\lambda, \kappa \geq 0$.
    \item The output is a vertex set~$S \subseteq V(G)$ such that~$S \in \sol_t(G,\H)$.
    \item When~$\bedt(G) \leq \lambda$ and~$\opt(G,\H) \leq \opt(G) + \kappa$, the algorithm guarantees that~$|S| = \opt(G,\H)$, so that~$S \in \optsol_t(G,\H)$.
\end{itemize}
The algorithm runs in time~$n^{\Oh_t((\lambda + \kappa)^2)}$.
\end{lemma}
\begin{proof}
We first present the algorithm and then its analysis. The algorithm is recursive, with a setup similar (but more complicated) than the one in \autoref{lem:compute:bedt}. Before the recursion starts, we run the algorithm from \autoref{lem:compute:bedt} to test whether~$\bedt(G) \leq \lambda$. If this is not the case, we may simply output~$V(G)$. Hence in the remainder we may assume~$\bedt(G) \leq \lambda$. We now describe the recursion.

\paragraph*{Base case: $V(G) = \emptyset$.} In this case we simply return the empty set as the unique optimal solution.

\paragraph*{Base case: $\lambda = 0$.} By definition, if~$\bedt(G) = \lambda = 0$ then~$G$ is $K_t$-free.
\begin{enumerate}
    \item If~$G$ contains a~$K_t$, then we can be sure that~$\bedt(G) > \lambda$ and return the trivial solution~$V(G)$.
    \item If~$G$ is~$K_t$-free and~$\H = \emptyset$, we return the empty set which belongs to~$\optsol_t(G,\H)$.
    \item If~$G$ is~$K_t$-free,~$\H \neq \emptyset$, and~$\kappa = 0$, then we conclude that~$\opt(G,\H) > \opt(G) + \kappa = 0 + 0$ and return the trivial solution~$V(G)$.
    \item If~$G$ is~$K_t$-free,~$\H \neq \emptyset$, and~$\kappa > 0$, then we consider one arbitrary~$H_1 \in \H$. For each vertex~$v_i \in V(H_1)$ we recurse on the instance~$(G-v_i, \H \cap (V(G)\setminus v_i), \lambda, \kappa-1)$. Let~$S_1, \ldots, S_{|V(H)|}$ be the results of the recursive calls. We output a minimum-size set among~$\{\{v_i\} \cup S_i \mid i \in [|V(H_1)|]\}$. (Steps of this type mimic the behavior of a bounded-depth search tree algorithm for $\H$-hitting set parameterized by solution size~$\kappa$.)
\end{enumerate}
This concludes the description of the base case, which is easily seen to be correct.

\paragraph*{Step case: $\lambda > 0$.} The step case consists of various subcases, depending on how the recursive definition of~$\bedt(G)$ (\autoref{def:bedHplus}) attains its minimum value. The algorithm considers the vertex set that would be eliminated by the recursive definition and guesses how an optimal solution intersects that set.

\begin{enumerate}
    \item If~$G$ is disconnected with connected components~$C_1,\ldots, C_m$, then we recurse on the instance~$(C_i, \H \cap C_i, \alpha, \kappa)$ for each~$i \in [m]$ and return the union~$S$ of the obtained vertex sets. Note that~$S \in \sol_t(G, \H)$ since each graph~$H_j \in \H$ is complete and is therefore contained fully within some component~$C_i$.
    \item If~$G$ is connected, then let~$Z \subseteq V(G)$ be the subset of vertices which are not contained in any~$K_t$-subgraph of~$G$. We distinguish two cases.
    \begin{enumerate}
        \item If~$Z \neq \emptyset$, then proceed as follows. For each subset~$Z'\subseteq Z$ of size at most~$\kappa$, which will be our guess for how the solution intersects~$Z$, we define:
        \[\H_{Z'} = \{ H \cap (G - Z) \mid H \in \H \wedge V(H) \cap Z' = \emptyset.\}\]
        In other words, we take those subgraphs of~$\H$ which are disjoint from~$Z'$ and project them onto the remaining graph~$G-Z$. If~$\H_{Z'}$ contains the empty subgraph on vertex set~$\emptyset$, then our guess of~$Z'$ was incorrect. Otherwise, we recurse on~$(G - Z, \H_{Z'}, \kappa - |Z'|, \lambda)$, resulting in a subset~$S_{Z'} \subseteq V(G - Z)$. If at least one recursive call was made, we output a minimum-size subset among~$\{S_{Z'} \cup Z' \mid Z' \subseteq Z \wedge |Z'| \leq \kappa \wedge \emptyset \notin \H_{Z'}\}$. If all choices of~$Z'$ lead to~$\emptyset \in \H_{Z'}$, we output~$V(G)$ as a trivial solution. (Below we will argue that~$\opt(G,\H) > \opt(G) + \kappa$ when this happens.) \label{xp:alg:branchz}
        \item Now suppose~$Z = \emptyset$. If the algorithm reached this point, then~$V(G) \neq \emptyset$, the graph~$G$ is connected, and all its vertices belong to some~$K_t$. Hence we have~$\bedt(G) = 1 + \min _{T \subset V(G)} \bedt(G-T)$, where~$T$ ranges over all roots of~$G$. Since have~$\bedt(G) \leq \lambda$ based on the opening step of the algorithm, we can invoke \autoref{lemma:computeRoot} to compute a $\bedt$-root~$T$ of~$G$ in time~$n^{\Oh_t(\lambda)}$.
        \label{xp:alg:emptyz}

        Let~$T = \{v_1, \ldots, v_m\}$. For each~$j \in [m]$, recall (\autoref{def:pending}) that~$C^T(v_j)$ coincides with the subgraph of~$G$ induced by the union of the connected components of~$G-T$ for which~$v_j$ is the unique vertex of~$T$ to which they are adjacent, together with~$v_j$ itself. Let~$\H_j := \H \cap C(v_j)$. We want to compute a good solution~$S^+_j \in \sol_t(C^T(v_j), \H_j)$ that contains~$v_j$. To do this, we proceed as follows. We recursively call the algorithm on~$(C^T(v_j) - v_j, \H \cap C(v_j)-v_j, \lambda - 1, \kappa)$ and define~$S^+_j$ as the union of~$\{v_j\}$ with the obtained result.

        Additionally, we want to compute a good solution~$S^-_j \in \sol_t(C^T(v_j), \H_j)$ that avoids~$v_j$. To that end, we let~$\H'_j$ contain the graph~$H \cap (C^T(v_j) - v_j)$ for each~$H \in \H$ with~$H \subseteq C^T(v_j)$, together with the graph~$F - v_j$ for each copy~$F$ of~$K_t$ in~$C^T(v_j)$ that contains~$v_j$; hence we project the annotated subgraphs contained in the extended component~$C^T(v_j)$ onto the component~$C^T(v_j) - v_j$, thereby omitting the vertex~$v_j$ from~$H$ if it was contained in~$H$, and in addition we ensure that the~$K_t$ subgraphs containing~$v_j$ are hit by inserting new annotations that enforce they are hit at a vertex other than~$v_j$. If no graph in~$\H'_j$ has an empty vertex set, we recursively call the algorithm on~$(C^T(v_j) - v_j, \H'_j, \lambda - 1, \kappa + 1)$, resulting in a set~$S^-_j \in \sol_t(C^T(v_j) - v_j, \H'_j)$; note that~$\kappa$ increases in the recursive call. If~$\H'_j$ contains a graph on the empty vertex set, then set~$S^-_j := V(C^T(v_j) - v_j)$.

        Note that~$S^+_j$ and~$S_j^-$ both belong to~$\sol_t(C^T(v_j), \H_j)$ and that~$S^+_j$ contains~$v_j$ but~$S^-_j$ does not (assuming there is any solution that does not contain~$v_j$). When~$|S^+_j| \leq |S_j^-|$, it effectively tells us we can find a solution to~$(C^T(v_j), \H_j)$ that contains~$v_j$ (thereby potentially hitting subgraphs of~$\H$ that intersect~$T_j$ in multiple vertices) without additional cost. We will therefore treat the vertices~$v_j$ for which~$|S^+_j| \leq |S_j^-|$ differently from the vertices where we get~$|S^+_j| > |S_j^-|$.

        We now proceed as follows. Intuitively, whenever it is possible for a solution in an extended component~$C^T(v_j)$ to contain the attachment point~$v_j$ without using more vertices than needed without being forced to contain the attachment point, we do so. Let~$T'$ be~$T$ minus all vertices~$v_j$ for which~$|S^+_j| \leq |S^-_j|$. Using a bounded-depth branching algorithm, find a minimum vertex set~$S_{T}$ of~$G[T']$ that hits all subgraphs~$\H \cap T'$ which were not yet hit by~$T \setminus T'$, if there is such a vertex set if size at most~$\kappa$. (To find such a set, it suffices to run a bounded-depth search tree algorithm with a recursion depth of~$\kappa$ and a branching factor of~$\max_{H \in \H} |V(H)| < t \in \Oh_t(1)$.)

        We now define the output of the algorithm as follows. If the bounded-depth search tree does not find any solution of size at most~$\kappa$, we output the trivial solution~$V(G)$. Otherwise, we output the union of~$S_{T}$, the sets~$S^+_j$ for the components~$C^T(v_j) - v_j$ where~$|S^+_j| \leq |S^-_j|$, and the sets~$S^-_j$ for the remaining components. Intuitively, the sets~$S^+_j$ and~$S^-_j$ together hit all the copies of~$K_t$ in~$G$ (since~$G[T]$ is~$K_t$-free) and also hit those subgraphs in~$\H$ which are fully contained in some~$C^T(v_j)$, while the set~$S_{T}$ hits the remaining subgraphs of~$\H$.
    \end{enumerate}
\end{enumerate}

This concludes the description of the algorithm.

\subparagraph*{Running time.} We argue that the algorithm runs in time~$n^{\Oh_t(\lambda + \kappa)^2}$ by analyzing the recursion tree generated by the algorithm.

To bound the depth of the recursion tree, first note that the value of~$\lambda + \kappa$ never increases in the algorithm: we decrease~$\lambda$ whenever we increase~$\kappa$. The value of of~$\lambda + \kappa$ strictly decreases after every three successive levels of recursion: if we have an iteration in which~$\lambda + \kappa$ does not decrease, it is either because we split on connected components or because we remove the vertices~$Z$ which do not belong to any~$K_t$-subgraph and make a guess of~$Z' = \emptyset$. There cannot be three such recursive calls in a row: after splitting on connected components, the graph is connected; when removing vertices~$Z$ which do not belong to any~$K_t$-subgraph and splitting on connected components, in the next call the graph will be connected with every vertex in a~$K_t$-subgraph so that we will trigger Case~\ref{xp:alg:emptyz} which decreases~$\lambda$ as it recurses. Since the value of~$\lambda + \kappa$ decreases every three successive levels and neither parameter can become negative, the depth of the recursion tree is bounded by~$\Oh(\lambda + \kappa)$.

The branching factor of the algorithm is bounded by~$n^{\Oh(\kappa)}$, since the sets~$Z' \subseteq Z$ we branch on in Case~\ref{xp:alg:branchz} have size at most~$\kappa$; trivially, the number of connected components for which we recurse cannot exceed~$n$. The number of nodes in the recursion tree is bounded by the branching factor~$n^{\Oh(\kappa)}$ raised to the power depth~$\Oh(\lambda + \kappa)$, and is therefore~$n^{\Oh((\lambda + \kappa)^2)}$.

The time per iteration of the described algorithm is bounded by~$t^{\lambda + \kappa} \cdot n^{\Oh(1)} + n^{\Oh_1(\lambda)} \subseteq n^{\Oh_t(\lambda + \kappa)}$: the expensive steps are running a bounded-depth search tree algorithm with a recursion depth of~$\kappa$ and branching factor less than~$t$, and invoking \autoref{lemma:computeRoot}. Multiplying the time spent per node with the total number of nodes, we get a total running time of~$n^{\Oh((\lambda + \kappa)^2)}$.

\subparagraph*{Correctness.} Correctness in the base cases is trivial. Below, we treat the various step cases and argue for their correctness. For this argument we use induction on the depth of the recursion tree generated by the recursive call; this is valid since the running time analysis establishes that this depth is finite. Hence we may assume inductively that all recursive calls performed by the algorithm result in an output consistent with the specification of the algorithm.

\begin{enumerate}
    \item When~$G$ is disconnected, the algorithm returns the union of recursive calls for the connected components~$C_1, \ldots, C_m$. Note that~$\bedt(C_i) \leq \bedt(G)$ by \autoref{def:bedHplus}. Since each forbidden subgraph is contained entirely within one connected component, we have~$\opt(G) = \sum_{i \in [m]} \opt(C_i)$ and~$\opt(G, \H) = \sum_{i \in [m]} \opt(G, \H_{C_i})$. From this, it follows that~$\opt(C_i, \H_{C_i}) - \opt(C_i) \leq \opt(G, \H) - \opt(G)$. What we have derived so far shows that whenever~$\bedt(G) \leq \lambda$ and~$\opt(G,\H) \leq \opt(G) + \kappa$ (which requires the algorithm to output an optimal solution), we will have~$\bedt(C_i) \leq \lambda$ and~$\opt(C_i, \H_{C_i}) \leq \opt(C_i) + \kappa$ (which ensures that the recursive call outputs an optimal solution). It follows that the union of the recursively obtained sets is indeed an optimal solution whenever the conditions require the algorithm to output an optimal solution, which proves correctness for this case.
    \item We consider the two subcases separately.
    \begin{enumerate}
        \item Suppose first that~$Z \neq \emptyset$. We claim that the vertex set~$S^*$ given as output satisfies~$S^* \in \sol_t(G,\H)$. This is trivial if~$S^* = V(G)$. If not, then it is the union of a solution~$S_{Z'} \in \sol_t(G - Z, \H_{Z'})$ with the set~$Z'$. Since no vertex of~$Z$ occurs in any~$K_t$-subgraph by definition, this union trivially hits all copies of~$K_t$. For each forbidden subgraph~$H \in \H$ that is not hit by~$Z'$, we add the subgraph~$H \cap (G-Z)$ to the collection~$\H_{Z'}$. Since~$S_{Z'}$ hits all subgraphs of~$\H_{Z'}$, it follows that~$S_{Z'} \cup Z' \in \sol_t(G,\H)$.

        We now argue for optimality in case~$\bedt(G) \leq \lambda$ and~$\opt(G,\H) \leq \opt(G) + \kappa$. The algorithm considers all subsets~$Z' \subseteq Z$ of size at most~$\kappa$. Under the stated condition, for any~$S \in \optsol_t(G, \H)$ we will have~$|S \cap Z| \leq \kappa$: we have~$S \setminus Z \in \sol_t(G)$ since no vertex of~$Z$ is contained in a~$K_t$-subgraph, so if~$|S \cap Z| > \kappa$ then~$\opt(G) < \opt(G, \H) - \kappa$, a contradiction. The algorithm therefore considers a choice of~$Z'$ that matches with the intersection of an optimal solution with~$Z$, which is easily seen to lead to a solution of size~$\optsol_t(G,\H)$.

        \item We now deal with the last case, that~$Z = \emptyset$. To prove correctness, we first establish that the output~$S^*$ satisfies~$S^* \in \sol_t(G,\H)$. This is trivial when~$S^* = V(G)$, so consider the case that it is not. Then~$S^*$ was defined as the union of some sets~$S^+_j \in \sol_t(C'_j, \H_j)$, some sets~$S^-_j \in \sol_t(C'_j, \H_j)$, and a set~$S_{T}$ that hits all subgraphs in~$\H$ which are disjoint from the chosen~$S^+_j, S^-_j$. It is easy to see that~$S^* \in \sol_t(G,\H)$ for each such union: each~$K_t$-subgraph is fully contained in some component~$C'_j$ since~$G[T]$ is~$K_t$-free, and hence each~$K_t$-subgraph is hit by the solution~$S^+_j$ or~$S^-_j$ to the corresponding extended component~$C'_j$. Each forbidden graph~$H \in \H$ is either fully contained in some~$C'_j$, or is fully contained in~$G[T]$: since each such~$H$ is a clique, it cannot intersect multiple connected components of~$G - T$. When~$H \in \H$ is fully contained in~$C'_j$ it is hit by either~$S^+_j$ or~$S^-_j$; when~$H$ is contained in~$G[T]$ and not hit by the chosen vertices, then it is hit by the solution~$S_{T}$ for~$\H_{G[T'_i]}$ that is included in the solution. Hence each candidate solution belongs to~$\sol_t(G,\H)$.

        As the last part of the argument, we prove that~$S^* \in \optsol_t(G,\H)$ whenever~$\bedt(G) \leq \lambda$ and~$\opt(G,\H) \leq \opt(G) + \kappa$. To do so, we first prove the following claim which says that the computation of~$S^+_j$ and~$S^-_j$ succeeds in a certain technical sense.

        \begin{claim} \label{claim:xp:alg}
        If~$\opt(G,\H) \leq \opt(G) + \kappa$, then the sets~$S^+_j$ and~$S^-_j$ computed for each vertex~$v_1, \ldots, v_m$ of~$T$ satisfy the following:
        \begin{enumerate}
            \item If~$\optsolwith(C(v_j), \H \cap C(v_j), v_j) \neq \emptyset$, then~$S^+_j \in \optsol_t(C(v_j), \H \cap C(v_j))$.
            \item If~$\optsolwith(C(v_j), \H \cap C(v_j)) = \emptyset$, then~$S^-_j \in \optsol(C(v_j), \H \cap C(v_j))$.
        \end{enumerate}
        \end{claim}
        \begin{claimproof}
        Suppose first that~$\optsolwith(C(v_j), \H \cap C(v_j), v_j) \neq \emptyset$, so that there is an optimal solution to the subproblem that contains the attachment vertex~$v_j$ in the root. This means that removing vertex~$v_j$ (along with all the subgraphs from~$\H$ that it hits) from the subinstance, decreases the value of the optimum by one. Hence~$\opt(C(v_j) - v_j, \H \cap (C(v_j)-v_j)) = \opt(C(v_j), \H \cap (C(v_j)-v_j)) - 1$. Since~$C(v_j) - v_j$ is a subgraph of~$G-T$ and we computed~$T$ as a $\bedt$-root of~$G$, we have~$\bedt(G[C(v_j)] - v_j) \leq \bedt(G-T) \leq \bedt(G) - 1 \leq \lambda - 1$. Hence the recursive call with parameters~$(C(v_j) - v_j, \H \cap (C(v_j)-v_j), \lambda-1, \kappa)$ and whose result is combined with~$\{v_j\}$ to obtain~$S^+_j$, has the property that the third parameter~$\lambda-1$ indeed bounds the~$\bedt$-value of the graph in the first parameter. Using the fact that~$\opt(G) = \sum_{v \in T} C(v)$, it is not difficult to verify that~$\opt(C(v), \H \cap C(v)) \leq \opt(C(v)) + \kappa$ using the assumption~$\opt(G,\H) \leq \opt(G) + \kappa$. Together, this means that the preconditions are satisfied for the invoked recursive call to output an optimal solution. By the induction hypothesis, this means that the result of the recursive call a set belonging to~$\opt(C(v_j) - v_j, \H \cap (C(v_j)-v_j))$; since~$S^+_j$ was obtained by taking the union with~$v_j$ we therefore have~$S^+_j \in \optsolwith_t(C(v_j), \H \cap C(v_j))$; the optimality of the computed solution follows from the fact that its size is one larger than an optimal solution for~$(C(v_j) - v_j, \H \cap (C(v_j)-v_j))$, whose optimal solutions are in turn one smaller than those in for~$(C(v), \H \cap C(v_j))$. This proves the first item in the statement of the claim.

        The argument for the second item uses similar ideas, but is somewhat more subtle: the reason for this is that the recursive call is for the graph~$C(v_j) - v_j$ from which vertex~$v_j$ has been removed, while to capture the fact that we look for a solution without~$v_j$, we have added the projections of all~$K_t$-subgraphs of~$C(v_j)$ intersecting~$v_j$ to the annotated set~$\H'_j$. The main observation to handle this complication is the following: in the second case of the claim, we assume that~$\optsolwith_t(C(v_j), \H \cap C(v_j))$ is empty. Under this assumption, we have~$\opt(C(v_j), \H \cap C(v_j)) = \opt(C(v_j) - v_j, \H'_j)$, with~$\H'_j$ as defined during the algorithm. Note that adding the mentioned projections to the annotated set may increase the gap between the costs of optimal solutions for hitting~$K_t$-subgraphs versus solutions that also hit all annotated graphs. However, this gap will increase by at most one, since removal of~$v_j$ decreases~$\opt(C(v_j))$ by at most one. Since we increase the fourth parameter to~$\kappa + 1$ in the recursive call, this ensures that the parameters of the recursive call are such that the specification of the algorithm guarantees it outputs an optimal solution to the instance~$(C(v_j) - v_j, \H'_j)$. Due to the projection, this solution also belongs to~$\sol_t(C(v_j), \H \cap C(v_j))$, and since there is no optimal solution to that problem containing~$v_j$, its size is~$\opt((C(v_j), \H \cap C(v_j))$. This completes the proof of the claim.
        \end{claimproof}

        Using \autoref{claim:xp:alg} and Lemmata~\ref{lemma:restructure:root} and~\ref{lemma:overpay:in:root} we now complete the correctness proof. Under the active assumption that~$\bedt(G) \leq \lambda$ and~$\opt(G,\H) \leq \opt(G) + \kappa$, we prove that the computed solution~$S^*$ belongs to~$\optsol_t(G,\H)$. \autoref{claim:xp:alg} shows that the computed set~$S^+_j$ belongs to~$\optsol_t(C^T(v_j), \H \cap C(v_j))$ when there is an optimal solution containing~$v_j$; otherwise,~$S^-_j$ belongs to~$\optsol_t(C^T(v_j), \H \cap C(v_j))$. From this, together with the fact that the recursively computed sets are valid solutions to the respective subproblems, it follows that if~$\optsolwith_t(C^T(v_j), \H \cap C(v_j), v_j) \neq \emptyset$, then~
        $S^+_j$ is such a solution and it has size at most~$|S^-_j|$; if no optimal solution contains~$v_j$, then~$S^-_j \in \optsol_t(C^T(v_j), \H \cap C(v_j))$.
        By applying \autoref{lemma:restructure:root} separately for each pending component of~$T$, we obtain that there is an optimal solution in~$\optsol_t(G,\H)$ that contains the union of all the sets~$S^+_j$ for~$j \in T'$, and the union of all sets~$S^-_j$ for~$j \in T \setminus T'$. By \autoref{lemma:overpay:in:root}, the remaining part of the solution consists of at most~$\kappa$ vertices. Since the only subgraphs which are not hit yet by the chosen~$S^+_j$ and~$S^-_j$ are those fully contained in~$G[T]$ and which are not hit by~$T'$, the remainder of an optimal  solution is formed by a hitting set of size at most~$\kappa$ for these unhit subgraphs. Since the algorithm computes an optimal solution of size at most~$\kappa$ for this hitting set instance, the union of the obtained sets is indeed an optimal solution.
    \end{enumerate}
\end{enumerate}
This concludes the proof.
\end{proof}

\begin{lemma}\label{lemma:poly-opt-conf-lambda4}
  Let $t \ge 3$. There is an algorithm that, given an input $(G,\F)$ of $\EKtH$ and integer~$\lambda$ such that~$\bedt(G) \leq \lambda$, runs in time~$p_0^{(\lambda,t)}(n)=n^{\Oh_t(\lambda^2)}$ and outputs the following:
  \begin{itemize}
  \item the value $\opt(G)$, and
  \item whether the instance $(G,\F)$ is clean.
  \end{itemize}
\end{lemma}
\begin{proof}
To compute~$\opt(G)$ it suffices to apply \autoref{lemma:OPTforboundedlambda} to the tuple~$(G, \H := \emptyset, \lambda, \kappa := 0)$ and take the size of the obtained solution. To determine whether the instance~$(G,\F)$ is clean, we run the algorithm once more on~$(G, \H := \F, \lambda, \kappa := 0)$. If the output given in the second run is equally large as that in the first run, the instance is clean; otherwise it is not.
\end{proof}

\section{Hardness results}
\label{sec:hardness}

 In this section we present two  reductions from \textsc{CNF-SAT}, and to transfer the the non-existence of polynomial kernels (under reasonable complexity assumptions), we use the notion of \emph{polynomial parameter transformation}, introduced by Bodlaender, Thomass{\'{e}}, and Yeo~\cite{BodlaenderTY11}.  A polynomial parameter transformation from a parameterized problem $P$ to a parameterized problem $Q$ is an algorithm that, given an instance $(x,k)$ of $P$, computes in polynomial time an equivalent instance $(x',k')$ of $Q$ such that $k'$ is bounded by a polynomial depending only on $k$. It follows easily from the definition that if $P$ does not admit a polynomial kernel under some complexity assumption, then the same holds for $Q$.
The complexity hypothesis in the following proposition builds on the results by  Fortnow and Santhanam~\cite{FortnowS11}.

\begin{proposition}[Dell and van Melkebeek~\cite{DellM14}]
\label{prop:hardness-SAT}
{\sc CNF-SAT}  does not admit a polynomial kernel parameterized by the number of variables of the input formula, unless ${\sf NP} \subseteq {\sf coNP}/{\sf poly}$.
\end{proposition}

We are ready to present our main hardness result, which is inspired by other reductions for related problems~\cite{JansenK11,FominS16,BougeretJS22,DekkerJ22,CyganLPPS14}. The crucial issue of this reduction, and its main conceptual novelty, is the following fact: when $H$ is not a clique, the intersection of an occurrence of $H$ with (a subgraph of) $G-X$  may be \textit{disconnected}. We exploit this fact by creating clause gadgets with large minimal blocking sets whose elements are disconnected (in \autoref{fig-reduction-lambda1}, each pair of consecutive non-adjacent vertices $u,v$ is an element of a blocking set), and this results in clause gadgets behaving as ``chains'' where the propagation of information (that is, the vertices picked by the solution) is done without needing edges connecting the elements of the chain (thus, in a ``wireless'' fashion), easily implying that the corresponding parameter in $G-X$ is bounded by a constant.

\begin{theorem}\label{thm:hardness-lambda1}
Let $H$ be a biconnected graph that is not a clique. The $H$-{\sc Subgraph Hitting} (resp. $H$-{\sc Induced Subgraph Hitting}) problem does not admit a polynomial kernel parameterized by the size of a given vertex set $X$ of the input graph $G$ such that $\veds(G-X) \leq 1$ (resp. $\vedis(G-X) \leq 1$), unless ${\sf NP} \subseteq {\sf coNP}/{\sf poly}$.
\end{theorem}
\begin{proof}
    We present a polynomial parameter transformation from the \textsc{CNF-SAT} problem parameterized by the number of variables, which does not admit a polynomial kernel by \autoref{prop:hardness-SAT}, unless ${\sf NP} \subseteq {\sf coNP}/{\sf poly}$. We present our reduction for the $H$-\textsc{Subgraph Hitting} problem, and at the end of the proof we observe that the same reduction applies to $H$-\textsc{Induced Subgraph Hitting} as well. 

    Given a \textsc{CNF-SAT} formula $\phi$ with $n$ variables $x_1,\ldots,x_n$ and $m$ clauses $C_1,\ldots,C_m$, we proceed to construct in polynomial time an instance $G$ of $H$-\textsc{Subgraph Hitting}, together with a set $X \subseteq V(G)$ with $\vedH(G-X) \leq 1$ and $|X| = |V(H)|\cdot n$, such that $\phi$ is satisfiable if and only if $G$ contains a solution of $H$-\textsc{Subgraph Hitting} of size at most $n - m + \sum_{j=1}^m c_j$, where $c_j$ denotes the number of literals in clause $C_j$. Since $|V(H)|$ is a constant, this would indeed define a polynomial parameter transformation from  \textsc{CNF-SAT} parameterized by the number of variables to  $H$-\textsc{Subgraph Hitting}  parameterized by the size of a given vertex set $X$ of the input graph $G$ such that $\vedH(G-X) \leq 1$.

    For each variable $x_i$, we add a disjoint copy of $H$ to $G$. We call such a copy of $H$ the \emph{$i$-variable-copy} of $H$. For each clause $C_j$, we add $c_j - 1$ disjoint copies of $H$ to $G$, and we order them arbitrarily from $1$ to $c_j - 1$. Moreover, we add two new vertices $u_j^0$ and $v_j^{c_j}$ to $G$.  We call each of these $c_j -1$ copies of $H$ a \emph{$j$-clause-copy} of $H$. Note that, so far, we have introduced $n - m + \sum_{j=1}^m c_j$ disjoint copies of $H$ in $G$.

    We now proceed to interconnect these copies of $H$ according to $\phi$. Since $H$ is a biconnected graph that is not a clique (hence, it is 2-connected), it follows that $|V(H)| \geq 4$. Thus, in particular there exist two {\sl non-adjacent} vertices $u,v \in V(H)$ and another vertex $w \in V(H)$ distinct from $u$ and $v$. Let $H' = H - \{u,v,w\}$. (Even if it is not critical for the proof, note that $|V(H')| \geq 1$.) Let also $z^+$ and $z^-$ be two distinct vertices of $H$ (not necessarily different from $u,v,w$). We will use the copies of these vertices in the variable-copies and clause-copies of $H$ to interconnect them in $G$. To this end, for three distinct vertices $a,b,c \in V(G)$ and a subgraph $F$ of $G$ isomorphic to $H'$ not containing any of $a,b,c$, by \emph{adding an $(a,b,c,V(F))$-copy of $H$ to $G$} we mean the operation of, starting from $G[\{a,b,c\} \cup V(F)]$, adding the missing edges to complete a copy of $H$, where vertex $a$ (resp. $b$, $c$) of $G$ plays the role of vertex $w$ (resp. $u$, $v$) of $H$, and $F$ plays the role of $H'$, with a fixed isomorphism that we suppose to have at hand.

    For each clause $C_j$ of $\phi$, consider an arbitrary ordering of its literals as $\ell_1, \ldots, \ell_{c_j}$, and recall that $G$ contains $c_j - 1$ ordered disjoint $j$-clause-copies of $H$ together with two extra vertices $u_j^0$ and $v_j^{c_j}$. For $i \in [c_j]$, we add a new copy of $H'$ to $G$, which we denote by $F_{j}^i$. For $i \in [c_j -1]$, let $u_j^i$ and $v_j^i$ be the copies of vertices $u$ and $v$ of $H$, respectively,  in the $i$-th $j$-clause-copy of $H$. For $i \in [c_j]$, if literal $\ell_i$ of clause $C_j$ corresponds to a positive (resp. negative) occurrence of a variable $x_p$, let $z$ be the copy of vertex $z^+ \in V(H)$ (resp. $z^- \in V(H)$) in the $p$-variable-copy of $H$. Then we add a $(z,u_j^{i-1},v_j^{i},V(F_{j}^i))$-copy of $H$ to $G$, and we call such a copy of $H$ a \emph{transversal-copy} of $H$, denoted by $H_j^i$. We define $X  \subseteq V(G)$ to be the union of the vertex sets of all the variable-copies of $H$, and note that $|X| = |V(H)| \cdot n$. This completes the construction of $G$ and $X$, which is illustrated in \autoref{fig-reduction-lambda1}.

\begin{figure}[tb]
 \begin{center}
  \includegraphics[scale=.9]{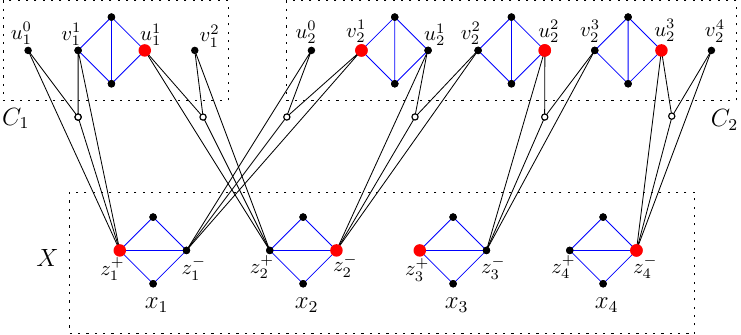}
\end{center}
\caption{Example of the construction of graph $G$ in the proof of \autoref{thm:hardness-lambda1} for $H$-\textsc{Subgraph Hitting}. In this example, $H$ is the diamond (that is, $K_4$ minus one edge), $u$ and $v$ are the only pair of non-adjacent vertices in $H$, and $w$ is any other vertex. The construction corresponds to a CNF-SAT formula $\phi$ consisting of two clauses $C_1=(x_1 \vee x_2)$ and $C_2 = (\bar{x}_1 \vee \bar{x}_2 \vee \bar{x}_3 \vee \bar{x}_4)$, and the satisfying assignment $\alpha(x_1)=1$, $\alpha(x_2)=0$, $\alpha(x_3)=1$, and $\alpha(x_4)=0$. The variable-copies and clause-copies of $H$ are depicted in blue,  the vertices in the copies of $H' = H - \{u,v,w\}$ (which is a single vertex) are the white ones, and the vertices in the solution $S$ are the large red ones. Note that clause $C_2$ is satisfied by both $\bar{x}_2$ and $\bar{x}_4$; in the example we have taken $\bar{x}_2$ as the satisfying literal.}
 \label{fig-reduction-lambda1}
 \end{figure}

    
    In the next claim we prove one of the properties claimed in the statement of the theorem.  

    \begin{claim}\label{claim:lambda1atmost1}
     $\veds(G-X) = 1$ and $\vedis(G-X) = 1$.
    \end{claim}
    \begin{proof}
    Note that each connected component of $G-X$ corresponds to the clause-copies of $H$ associated with a clause $C_j$ and isolated vertices, together with the copies of $H'$ between those $j$-clause-copies of $H$ and isolated vertices. By construction of $G$, each such a copy of $H'$, say $F_j^i$, has at most two neighbors in $G-X$, namely vertices $u_j^{i-1}$ and $v_j^{i}$. If an occurrence of $H$ as a subgraph in $G-X$, say $F$, contained a vertex of $F_j^i$, since $|V(H')| = |V(H)| -3$, necessarily $F$ contains at least one of $u_j^{i-1}$ and $v_j^{i}$, and at least one more vertex in the $(i-1)$-th or $i$-th $j$-clause-copies of $H$. Thus, $u_j^{i-1}$ or $v_j^{i}$ is a separator of size one of $F$, contradicting the hypothesis that $H$ is biconnected.

     That is, we have proved that no vertex in a copy $F_j^i$ of $H'$ in $G-X$ is contained in an occurrence of $H$ as a subgraph, hence neither as an induced subgraphs. Therefore, those vertices can be removed while preserving the value of $\veds(G-X)$. Formally,

     $$\veds(G-X) = \veds(G - X - \bigcup_{j=1}^{m}\bigcup_{i=1}^{c_j} V(F_j^i)),
     $$
     and the same holds for $\vedis(G-X)$. To conclude the proof of the claim, it suffices to note that $G - X - \bigcup_{j=1}^{m}\bigcup_{i=1}^{c_j} V(F_j^i)$ consists of a disjoint union of clause-copies of $H$ and isolated vertices, and using the fact that $\veds$ (resp. $\vedis$) of a disconnected graph is the maximum of $\veds$ (resp. $\vedis$) over its connected components, by removing one arbitrary vertex from each such a copy of $H$ we get that $\veds(G-X)=1$ and $\vedis(G-X)=1$.
    \end{proof}


We now claim that $\phi$ is satisfiable if and only if $G$ contains a solution $S \subseteq V(G)$ of $H$-\textsc{Subgraph Hitting} of size at most $n - m + \sum_{j=1}^m c_j$.

Suppose first that $\phi$ is satisfiable and let $\alpha: \{x_1,\ldots,x_n\} \to \{0,1\}$ be a satisfying assignment of the variables. We define a set $S \subseteq V(G)$ of size $n - m + \sum_{j=1}^m c_j$ as follows (cf. the red vertices in \autoref{fig-reduction-lambda1}). For each variable $x_i$, if $\alpha(x_i)=1$  (resp. $\alpha(x_i)=0$), we add to $S$ the copy of $z^+$ (resp. $z^-$) in the $i$-variable-copy of $H$, which we denote by $z_i^+$ (resp. $z_i^-$). For each clause $C_j$, let $\ell_{s_j}$ be a literal in $C_j$ that is satisfied by the assignment $\alpha$. Recall that the $c_j -1$ $j$-clause-copies of $H$ are ordered (arbitrarily) from $1$ to $c_j - 1$. We add to $S$ the vertex set
    $$\{ v_j^i \mid 1 \leq i \leq s_j -1  \} \cup \{u_j^i \mid s_j \leq i \leq c_j -1 \}.$$
    In words, we add to $S$ the  copy of vertex $v$ in all the $j$-clause-copies of $H$ from $1$ to $s_j -1 $, and the copy of vertex $u$ in all the $j$-clause-copies of $H$ from $s_j$ to $c_j -1$. Note that $|S| = n - m + \sum_{j=1}^m c_j$, and it remains to prove that $G - S$ does not contain $H$ as a subgraph.  Note that each variable-copy and clause-copy of $H$ contains exactly two vertices that have neighbors outside of that copy --let us call these vertices \emph{boundary vertices} of that copy--, and that $S$ contains exactly one of these two boundary vertices for each of these copies of $H$. Hence, since $H$ is biconnected, no occurrence of $H$ in $G-S$ can contain a non-boundary vertex in a variable-copy or clause-copy of $H$.

    Moreover, there do not exist two pairs of integers $(i_1,j_1)$ and $(i_2,j_2)$, with $i_1 \neq i_2$ or $j_1 \neq j_2$, such that there exists an occurrence $F$ of $H$ in $G-S$ with $F \cap (V(H^{i_1}_{j_1}) \setminus X) \neq \emptyset$ and $F \cap (V(H^{i_2}_{j_2} \setminus X) \neq \emptyset$.
    Indeed, if such $(i_1,j_1)$ and $(i_2,j_2)$ existed, then,
    as $|N(V(H^i_j)\setminus X) \cap X| = 1$ for any two indices $i,j$, and as $F$ cannot contain a non-boundary vertex of a variable-copy of $H$, there would exist $z \in X$ such that $F \cap X = \{z\}$, implying that $z$ is a separator of $F$, contradicting the $2$-connectivity of $H$.

 Thus, if an occurrence of $H$ in $G-S$ existed, say $F$, then the above discussion and the construction of $G$ imply that $F$ should be one of the transversal-copies of $H$. But such an $F$ cannot exist in $G-X$ by the choice of~$S$: either $S$ contains one of the boundary vertices in the two $j$-clause-copies intersected by $F$ for some $j \in [m]$ or, if it is not the case, then $S$ contains the vertex in a variable-copy of $H$ corresponding to the literal that satisfies clause~$C_j$.

\medskip

Conversely, let $S \subseteq V(G)$ be a solution of $H$-\textsc{Subgraph Hitting} of size at most $n - m + \sum_{j=1}^m c_j$. Since $G$ contains $|S|$ disjoint variable-copies and clause-copies of $H$, necessarily $S$ consists of exactly one vertex in each of these copies. Since the boundary vertices in each of the variable-copies and clause-copies of $H$ are the only vertices with neighbors outside of the corresponding copy, we may assume that all the vertices in $S$ are boundary vertices. We define from $S$ a satisfying assignment $\alpha$ of $\phi$ as follows. If $S$ contains $z_i^+$ (resp. $z_i^-$) we set $\alpha(x_i)=1$  (resp. $\alpha(x_i)=0$). Let us verify that $\alpha$ indeed satisfies all the clauses of $\phi$. Consider an arbitrary clause $C_j$ with $c_j$ literals, and note that $S$ contains $c_j - 1$ vertices in the $j$-clause-copies of $H$. Therefore, there exists $s_j \in [c_j]$ such that the $s_j$-th transversal-copy of $H$ associated with $C_j$, say $F$, is {\sl not} hit by a vertex in a clause-copy of $H$. Thus, since $S \cap V(F) \neq \emptyset$, necessarily there exists an index $i \in [n]$ such that $S \cap V(F)$ is equal to either $z_i^+$ or $z_i^-$, and thus the defined assignment of variable $x_i$ satisfies clause $C_j$.One direction of the equivalence is proved in the next claim (such a solution $S$ is illustrated by the red vertices in  \autoref{fig-reduction-lambda1}).

%

\medskip

    To conclude the proof, we claim that the same reduction presented above proves the hardness result for the $H$-\textsc{Induced Subgraph Hitting} problem. Indeed, in the proof of the equivalence between the satisfiability of $\phi$ and the existence of a solution $S$ of $H$-\textsc{Subgraph Hitting} with the appropriate size, all that is relevant to the proof are the variable-copies, clause-copies, and transversal-copies of $H$. As all these occurrences of $H$ in $G$ occur as induced subgraphs,  the same reduction implies the non-existence of polynomial kernels for  $H$-\textsc{Induced Subgraph Hitting}.
\end{proof}

In \autoref{thm:hardness-treedepth} we replace the condition ``$\vedH(G-X) \leq 1$'' of \autoref{thm:hardness-lambda1} with the condition that  $\td(G-X)$ is bounded by a constant.
However, in the proof of \autoref{thm:hardness-treedepth} we need an extra condition on $H$ stronger than biconnectivity, namely the non-existence of a stable cutset.  The reduction in the proof of \autoref{thm:hardness-treedepth}  follows essentially the same lines as the one described in \autoref{thm:hardness-lambda1}, but in order to guarantee that $\td(G-X)$ is bounded by a constant (depending on $H$), we need to be more careful. Namely, in the interconnection among the variable and clause gadgets, now we cannot afford to add a distinct gadget for each literal in a clause, as it was the case for the copies of $H'$ in the proof of \autoref{thm:hardness-lambda1} (cf. the white vertices in \autoref{fig-reduction-lambda1}). Indeed, these copies of $H'$ can be removed ``for free'' when dealing with $\vedH$, but it is not the case anymore when dealing with treedepth, as they may blow up the value of $\td(G-X)$. In a nutshell, we overcome this issue by ``reusing'' these copies of a (now, carefully chosen) subgraph $H' \subseteq H$ for all the literals of the same clause. However, having a single common $H'$ for each clause may create undesired occurrences of $H$ (other than the variable-copies, clause-copies, and transversal-copies, as we wish), and preventing the existence of these undesired copies is the reason why we need an assumption on $H$ stronger than  biconnectivity. We now restate \autoref{thm:hardness-treedepth} and provide its full proof. 

\begin{theorem}[Restatement of \autoref{thm:hardness-treedepth}]\label{thm:hardness-treedepth-restated}
Let $H$ be a graph on $h$ vertices that is not a clique and that has no stable cutset. None of the $H$-{\sc Subgraph Hitting} and $H$-{\sc Induced Subgraph Hitting} problems admits a polynomial kernel parameterized by the size of a given vertex set $X$ of the input graph $G$ such that $\td(G-X) = \O(h)$, unless ${\sf NP} \subseteq {\sf coNP}/{\sf poly}$.
\end{theorem}

\begin{proof}
We again present a polynomial parameter transformation from the \textsc{CNF-SAT} problem parameterized by the number of variables, which does not admit a polynomial kernel by \autoref{prop:hardness-SAT}, unless ${\sf NP} \subseteq {\sf coNP}/{\sf poly}$. We also present the reduction for the $H$-\textsc{Subgraph Hitting} problem, and at the end of the proof we observe that the same reduction applies to $H$-\textsc{Induced Subgraph Hitting} as well. Let $h = |V(H)|$.

    Given a \textsc{CNF-SAT} formula $\phi$ with $n$ variables $x_1,\ldots,x_n$ and $m$ clauses $C_1,\ldots,C_m$, we  construct in polynomial time an instance $G$ of the $H$-\textsc{Subgraph Hitting} problem, together with a set $X \subseteq V(G)$ with $\td(G-X) = \O(h)$ and $|X| = h \cdot n$, such that $\phi$ is satisfiable if and only if $G$ contains a solution of $H$-\textsc{Subgraph Hitting} of size at most~$t$, where $t$ will be defined later (cf. \autoref{eq:budget-t}).

    The variable gadgets remain unchanged. Namely, for each variable $x_i$, we add a disjoint copy of $H$ to $G$. We call such a copy of $H$ the \emph{$i$-variable-copy} of $H$. Let $z^+,z^-$ be two distinct vertices of $H$. For $i \in [n]$. we denote the copy of $z^+$ (resp. $z^-$) in the $i$-variable-copy of $H$ by $z_i^+$ (resp. $z_i^-$).

    Before defining the updated clause gadgets, we need an gadgets. Let $A,B \subseteq V(H)$ be two disjoint non-empty vertex subsets such that there is no edge between $A$ and $B$ and such that $|A| + |B|$ is maximized among all possible candidate pairs. Note that these sets $A,B$ exists because $H$ is not a clique, and let $a = |A|$ and $b = |B|$. Let $w \in V(G) \setminus (A \cup B)$, and note that such a vertex $w$ exists because $H$ is, in particular, connected. Let $H' = H - (A \cup B \cup \{v\})$. (The fact that $H$ has no stable cutset implies that $|V(H')| \geq 1$, but this is not relevant to the proof.)

    We now introduce a gadget whose objective is to mimic the behavior of the clause-copies of $H$ in the proof of \autoref{thm:hardness-lambda1}, but with the anticomplete sets $A,B \subseteq V(H)$ playing the role of the non-adjacent vertices $u,v \in V(H)$ (cf. \autoref{fig-reduction-lambda1}). Let $s,t$ be two arbitrary distinct vertices of $H$, and suppose  without loss of generality that $\max\{a,b\} = a$. We define the \emph{$(A,B)$-gadget} as the graph obtained from $2a$ disjoint copies of $H$ by identifying, in a circular way, vertex $s$ of a copy with vertex $t$ of the next one. We call these $2a$ vertices $s,t$ the \emph{attachment vertices} of the $(A,B)$-gadget, consider a circular ordering of them from $1$ to $2a$, following the order of the copies of $H$, and denote them $y_1, y_2, \ldots, y_{2a}$. We define the \emph{$A$-vertices} of the $(A,B)$-gadget as the set $\{y_i \mid i \in[2a], \text{ $i$ odd}\}$, and the \emph{$B$-vertices} of the $(A,B)$-gadget as  the set $\{y_i \mid i \in[2a], \text{ $i$ even}\}$. That is, the $A$-vertices consist of the $a$ ``even'' attachment vertices of the $(A,B)$-gadget, and the $B$-vertices consist  the $a$ ``odd'' attachment vertices. To conclude the construction of the $(A,B)$-gadget, we add the appropriate edges among the $A$-vertices so that they induce a graph isomorphic to $H[A]$, and the appropriate edges among an arbitrary subset of size $b$ of the $B$-vertices so that they induce a graph isomorphic to $H[B]$ (note that it may hold that $a >b$, and in that case some of the $B$-vertices are not used to create a copy of $H[B]$). See \autoref{fig-ABgadget} for an illustration.

 \begin{figure}[h!tb]
 \begin{center}
  \includegraphics[scale=.8]{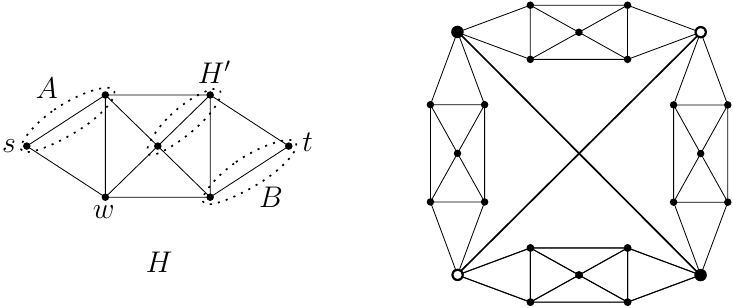}
\end{center}
\caption{On the left, a graph $H$ together with two anticomplete disjoint sets $A,B \subseteq V(H)$ with $|A|=|B|=2$, vertices $w,s,t$, and the subgraph $H'$ of $H$. On the right, an $(A,B)$-gadget, where the two large black (resp. white) vertices correspond to the $A$-vertices (resp. $B$-vertices) of the gadget.}
 \label{fig-ABgadget}
 \end{figure}

    The following claim formalizes the fact that any optimal solution of $H$-\textsc{Subgraph Hitting} in an $(A,B)$-gadget consists either all the $A$-vertices, or all the $B$-vertices.

    \begin{claim}\label{claim:ABgadget}
     Let $S$ be an optimal solution of $H$-\textsc{Subgraph Hitting} in an $(A,B)$-gadget. Then either $S= \{y_i \mid i \in[2a], \text{ $i$ odd}\}$ or $S= \{y_i \mid i \in[2a], \text{ $i$ even}\}$.
    \end{claim}
    \begin{proof}
    Since an $(A,B)$-gadget contains $a$ disjoint copies of $H$ (regardless of the subgraphs isomorphic to $H[A]$ and $H[B]$ added at end of the construction), necessarily $|S| \geq a$. Since $H$ is biconnected, it is easy to verify that either of the sets $\{y_i \mid i \in[2a], \text{ $i$ odd}\}$ or $\{y_i \mid i \in[2a], \text{ $i$ even}\}$ is an optimal solution. Finally, note that any solution $S$ needs to contain at least one vertex from each of the $a$ ``even'' copies of $H$ in the $(A,B)$-gadget, and at least one vertex from each of the $a$ ``odd'' copies of $H$ in the $(A,B)$-gadget. The two sets above are the only ones of size $a$ that hit all these copies of $H$. Indeed, the latter statement follows, for instance, from the fact that an even cycle has exactly two optimal vertex covers, each of them corresponding to a color class of a proper 2-coloring of the cycle.
    \end{proof}

    Equipped with the $(A,B)$-gadget, we can now define the clause gadgets. One should think of these gadgets as a generalized version of the clause gadgets in \autoref{fig-reduction-lambda1} with the role of the vertices $u,v \in V(H)$ replaced by the sets $A,B \subseteq V(H)$. Namely,
   for each clause $C_j$ with $c_j$ literals, we add $c_j - 1$ disjoint $(A,B)$-gadgets. Let the $A$-vertices and $B$-vertices of these gadgets be denoted by $\{A_j^i \mid i \in [c_j -1]\}$ and $\{B_j^i \mid i \in [c_j -1]\}$, respectively. Moreover, we add two graphs with vertex sets $A_j^0$ and $B_j^{c_j}$, both of size $a$, so that $G[A_j^0]$ is isomorphic to $H[A]$,  and $G[B_j^{c_j}]$ consists of a subgraph isomorphic to $H[B]$ plus isolated vertices. For $i \in [c_j]$, we let $\bar{B}_j^{i} \subseteq B_j^{i}$ be the subset of $B_j^{i}$ that induces a graph isomorphic to $G[B]$.

   So far, we have introduced $n$ disjoint copies of $H$ in the variable gadgets, and $a \cdot (c_j -1)$ disjoint copies of $H$ for each clause $C_j$, which amounts to a total budget of
   \begin{equation}\label{eq:budget-t}
       t := n + \sum_{j=1}^m a \cdot (c_j -1).
   \end{equation}

   To interconnect the variable and clause gadgets according to the formula $\phi$, we need a generalization of the operation of ``adding an $(a,b,c,F)$-copy of $H$'' defined in the proof of \autoref{thm:hardness-lambda1}. For a vertex $t \in V(G)$ and three disjoint subgraphs $F_1,F_2,F_3$ of $G$ not containing $t$, respectively isomorphic to $H[A],H[B], H'$,
   by \emph{adding a $(t,V(F_1),V(F_2),V(F_3))$-copy of $H$ to $G$} we mean the operation of, starting from $G[\{t\} \cup V(F_1) \cup V(F_2) \cup V(F_3)]$, adding the missing edges to complete a copy of $H$, where vertex $t \in V(G)$ plays the role of vertex $w \in V(H)$, and $F_1$ (resp. $F_2,F_3$) plays the role of $H[A]$ (resp. $H[B], H'$), with fixed isomorphisms that we suppose to have at hand.

    For each clause $C_j$ of $\phi$, consider an arbitrary ordering of its literals as $\ell_1, \ldots, \ell_{c_j}$, and recall that $G$ contains $c_j - 1$ ordered disjoint $(A,B)$-gadgets along with two extra subgraphs isomorphic to $H[A]$ and $H[B]$. We add a single copy of $H'$ to $G$, which we denote by $F_{j}$ (recall that in the proof of \autoref{thm:hardness-lambda1} we added a distinct copy $F_{j}^i$ for each literal in $C_j$). For $i \in [c_j]$, suppose that literal $\ell_i$ of clause $C_j$ corresponds to an  occurrence of a variable $x_p$.
    If $x_p$ occurs positively (resp. negatively) in $C_j$, then
     we add a $(z_p^+, A_j^{i-1},\bar{B}_j^{i}, V(F_{j}))$-copy (resp. $(z_p^-,A_j^{i-1},\bar{B}_j^{i}, V(F_{j}))$-copy) of $H$ to $G$, and we call such a copy of $H$ a \emph{transversal-copy} of $H$.

    We define $X  \subseteq V(G)$ to be the union of the vertex sets of all the variable-copies of $H$, and note that $|X| = h \cdot n$. This completes the construction of $G$ and $X$. 

We now prove one of the properties claimed in the statement of the theorem.

    \begin{claim}\label{claim:td-atmost-h2}
     $\td(G-X) = \O(h)$.
    \end{claim}
    \begin{proof}
     We describe an elimination forest of $G-X$ with the claimed depth. Note that each connected component of $G-X$ corresponds to a clause gadget associated with a clause $C_j$, consisting of $c_j-1$ disjoint $(A,B)$-gadgets, isolated vertices, and two subgraphs isomorphic to $H[A]$ and $H[B]$, all these appropriately attached to the copy $F_j$ of $H'$. In $|V(H')| < h$ rounds, we can remove the vertices of $F_j$ one by one, and then we are left with isolated vertices, disjoint $(A,B)$-gadgets and two subgraphs isomorphic to $H[A]$ and $H[B]$. For each of the latter two graphs, we can terminate in at most $\td(H)$ rounds. For each of the remaining $(A,B)$-gadgets, we can first remove its $2a < h$ attachment vertices, and then we are left with disjoint subgraphs of $H$, whose treedepth is at most $\td(H)$. So overall, we have required $2\td(H) + h = \O(h)$ rounds.
    \end{proof}

    We now claim that $\phi$ is satisfiable if and only if $G$ contains a solution $S \subseteq V(G)$ of $H$-\textsc{Subgraph Hitting} of size at most $t$.

    Suppose first that $\phi$ is satisfiable and let $\alpha: \{x_1,\ldots,x_n\} \to \{0,1\}$ be a satisfying assignment of the variables. We define a set $S \subseteq V(G)$ of size $t$ as follows. 
    For each variable $x_i$, if $\alpha(x_i)=1$  (resp. $\alpha(x_i)=0$), we add to $S$ the vertex $z_i^+$ (resp. $z_i^-$). For each clause $C_j$, let $\ell_{s_j}$ be a literal in $C_j$ that is satisfied by the assignment $\alpha$. Recall that the $c_j -1$ $(A,B)$-gadgets associated with $C_j$ are ordered from $1$ to $c_j - 1$. We add to $S$ the vertex set
    $$\{ B_j^i \mid 1 \leq i \leq s_j -1  \} \cup \{A_j^i \mid s_j \leq i \leq c_j -1 \}.$$
    In words, we add to $S$ the $B$-vertices in all the $(A,B)$-gadgets from $1$ to $s_j -1 $, and the $A$-vertices in all the $(A,B)$-gadgets from $s_j$ to $c_j -1$. Note that $|S| = t$, and it remains to prove that $G - S$ does not contain $H$ as a subgraph. Suppose for contradiction that $G - S$ contains a subgraph $F$ isomorphic to $H$. The choice of $S$ and  fact that $H$ is biconnected implies that $F$ cannot contain a vertex in a clause gadget that is neither an $A$-vertex nor a $B$-vertex, nor a vertex in an $i$-variable copy of $H$ distinct from $z_i^+$ and $z_i^-$. Thus, necessarily $F$ is a subgraph of the union of some $(A,B)$-gadgets, some copies $F_j$ of $H'$, and some copies of the vertices $z^+$ and $z^-$ in the variable gadgets. We now distinguish two cases.

    Suppose first that $F$ intersects only one clause gadget associated with a clause $C_j$. The definition of $S$ implies that $F$ cannot be a transversal-copy of $H$, as all these copies are intersected by $S$. Thus, $F$ intersects more than one $(A,B)$-gadget associated with $C_j$ or at least two variable gadgets (or maybe both). We proceed to define two disjoint vertex sets $A',B' \subseteq V(F)$ with no edges between them and such that $|A'| + |B'| > |A| + |B|$, contradicting the choice of $A,B \subseteq V(H)$. We start with $A'$ (resp. $B'$) being equal to the intersection of $F$ with the $A$-vertices (resp. $B$-vertices). For every vertex $z$ of $F$ belonging to an $i$-variable-copy of $H$, if any, we proceed as follows. Since $F$ is biconnected, necessarily $z$ is either $z_i^+$ or $z_i^-$. We may clearly assume that variable $x_i$ occurs in clause $C_j$, as otherwise $z$ would not be adjacent to the clause gadget associated with $C_j$. By construction of $S$ and because $z \notin S$, necessarily either all the $A$-vertices or all the $B$-vertices associated with the transversal-copy of $H$ containing $z$ belong to $S$. Note that $z$ is not adjacent to any of the $A$-vertices or $B$-vertices in other transversal-copies of $H$. Hence, $z$ is anticomplete to $A'$ or anticomplete to $B'$. In the former case, we add $t$ to $B'$, and in the latter we add it to $A'$. By construction, $A'$ and $B'$ and two disjoint anticomplete subsets of $V(F)$, and they include all the vertices of $F$ except those contained in $F_j$, the copy of $H'$ associated the clause $C_j$. Thus,
    $$|A'| + |B'| \geq |V(H)| - |V(H')| > |V(H)| - |V(H')| - 1 = |A| + |B|,$$
    a contradiction to the choice of $A,B \subseteq V(H)$.

    It remains two deal with the case where $F$ intersects more than one clause gadget. Since $H$ is biconnected and distinct clause gadgets can be connected only through copies of $z^+$ and $z^-$ in the variable gadgets, necessarily $F$ contains some of the latter vertices. Let $Z=\{z_1, \ldots, z_q\}$ be the set of copies of $z^+$ and $z^-$ contained in $F$. Note that the definition of $S$ implies that $Z$ contains at most one vertex in each $i$-variable-copy of $H$, which implies that $Z$ induces an independent set in $G$. Again by the biconnectivity of $H$, necessarily $|Z| \geq 2$. Since $F$ intersects more than one clause gadget, it follows that $Z$ is a stable cutset of $F$, a contradiction to the hypothesis that $H$ does not admit such a separator.

   Conversely, let $S \subseteq V(G)$ be a solution of $H$-\textsc{Subgraph Hitting} of size at most $t$. Since $G$ contains $t$ disjoint copies of $H$ in the variables and clause gadgets, necessarily $S$ contains exactly one vertex in each of the variable-copies of $H$ (which can be assumed to be a copy of $z^+$ or $z^-$) and, by \autoref{claim:ABgadget}, either all the $A$-vertices or all the $B$-vertices in each of the $(A,B)$-gadgets in the clause gadgets.  We define from $S$ a satisfying assignment $\alpha$ of $\phi$ as follows. If $S$ contains $z_i^+$ (resp. $z_i^-$) we set $\alpha(x_i)=1$  (resp. $\alpha(x_i)=0$). Let us verify that $\alpha$ indeed satisfies all the clauses of $\phi$. Consider an arbitrary clause $C_j$ with $c_j$ literals, and note that $S$ contains $a \cdot (c_j - 1)$ vertices in the clause gadget associated with $C_j$. Therefore, there exists $s_j \in [c_j]$ such that the $s_j$-th transversal-copy of $H$ associated with $C_j$, say $F$, is {\sl not} hit by a vertex in a clause gadget. Thus, since $S \cap V(F) \neq \emptyset$, necessarily there exists an index $i \in [n]$ such that $S \cap V(F)$ is equal to either $z_i^+$ or $z_i^-$, and therefore the defined assignment of variable $x_i$ satisfies clause $C_j$.
   
   \medskip

To conclude the proof, we claim that the same reduction presented above proves the hardness result for the $H$-\textsc{Induced Subgraph Hitting} problem. Indeed, in the proof of the equivalence between the satisfiability of $\phi$ and the existence of a solution $S$ of $H$-\textsc{Subgraph Hitting} with the appropriate size, all that is relevant to the proof are the variable-copies of $H$, copies of $H$ in the $(A,B)$-gadgets belonging to the clause gadgets, and transversal-copies of $H$ between the variable and clause gadgets. As all these occurrences of $H$ in $G$ occur as induced subgraphs,  the same reduction implies the non-existence of polynomial kernels for  $H$-\textsc{Induced Subgraph Hitting}.\end{proof}

Finally, let us mention, even if it is not relevant to our results,  that deciding whether a graph $G$ admits a stable cutset is \NP-hard, even if $G$ is assumed to be 2-connected~\cite{BrandstadtDLS00}.

\section{Further research}
\label{sec:conclusions}

In this paper we studied the existence of polynomial kernels for the $H$-\textsc{Subgraph Hitting} and $H$-\textsc{Induced Subgraph Hitting} problems under structural parameterizations, namely parameterized by the size of a modulator to a graph class $\C$ that has a ``simple structure''. Our  main achievement is the identification of two arguably natural graph parameters $\veds$ and $\beds$ (or $\vedis$ and $\bedis$ for the induced version) that allowed us to prove complexity dichotomies in terms of the forbidden graph $H$. Our results pave the way to a systematic investigation of this topic, where we identify the following avenues for further research.

\subparagraph*{Getting rid of the hypothesis on $H$.}
In our hardness results we need additional assumptions on $H$, mainly that $H$ is biconnected in \autoref{thm:hardness-lambda1}.
Observe that the requirement that $H$ is connected is unavoidable. Indeed, when $H$ is the union of a $K_5$ and a $K_{1,3}$, it is known \cite{JansenK021verte} that $H$-\textsc{Subgraph Hitting} is {\sf para-}\NP-hard, even for $\edcalH=0$. Moreover, when $H$ is a non-edge, $H$-{\sc Induced Subgraph Hitting} parameterized by vertex cover (which is a larger parameter than $\edcalH$) is equivalent to maximum clique parameterized by vertex cover, which does not admit a polynomial kernel under standard complexity assumptions~\cite{BodlaenderJK11}.
Thus, it is natural to wonder whether the biconnectivity hypothesis could be replaced by just connectivity.


\subparagraph*{Improving the degree of the kernel.} The degree of our polynomial kernel depends on the size $t$ of the excluded clique and on the value $\lambda$ of the promised upper bound $\bedt(G-X) \leq \lambda$. Namely, as stated in \autoref{thm:kernellight}, the kernel has size $\O_{\lambda,t}(|X|^{\delta(\lambda,t)})$, where function $\delta$ mainly depends on the upper bound on $\mmbs_t$ given in \autoref{thm:mmbslambda4}. This function behaves as a tower of exponents in $t$ of height $\lambda$. Hence, improving the bound on $\mmbs_t$ directly translates to an improvement of the kernel size. We did obtain such an improvement if instead of assuming that $\bedt(G-X) \leq \lambda$, one assumes that $\td(G-X) \leq \lambda$, namely with a function~$\lambda^\lambda \cdot 2^{\lambda^2}$; see \autoref{lemma:mmbstd}. We leave as an open problem to obtain improved upper bounds for  $\mmbs_t$ in terms of $\vedt$ and $\bedt$. 

\subparagraph*{Computing the modulator.} In our kernelization algorithm we assume that we are {\sl given} a modulator, namely a set $X \subseteq V(G)$ such that $\bedt(G-X) \leq \lambda$. Note that this hypothesis appears also in the related work dealing with \textsc{Feedback Vertex Set}~\cite{DekkerJ22}.
Obtaining a constant factor or even ${\sf poly(opt)}$-approximation of the modulator in polynomial time for fixed $t$ and $\lambda$, which will be enough for our kernelization algorithm (note that minimizing its size is \NP-hard~\cite{LewisY80}), remains an interesting direction. One may start with the probably simpler cases of a modulator to bounded
$\F_{\bar H}$-elimination distance or bounded $\vedt$.

\subparagraph*{Finding the right measure.} The focus of this article is on obtaining kernelization dichotomies as a function of the forbidden (induced) subgraph $H$. Of course, it is also relevant to characterize, for a fixed graph $H$, which is the most general (monotone or hereditary) target family $\C_H$ such that $H$-\textsc{(Induced) Subgraph Hitting} admits a polynomial kernel parameterized by the size of a modulator to a graph in $\C_H$. Needless to say, solving this general problem seems quite challenging. Indeed, even the case of \textsc{Vertex Cover}, that is, $H=K_2$, is far from being well understood for monotone or hereditary target graph classes, as
for instance the only known polynomial kernel for \textsc{Vertex Cover} parameterized by a modulator to a bipartite graph (i.e., an odd cycle transversal) is randomized and relies on quite powerful tools~\cite{Kratsch18,KratschW20}. One may hope that larger cliques allow for simpler characterizations, the natural first candidate being the case where $H$ is a triangle. Let $\C_{\Delta}$ be the, say, hereditary target graph class that we want to characterize. Following the approach of~\cite{BougeretJS22} that characterized the target {\sl minor-closed} graph classes for \textsc{Vertex Cover}, one may hope that $K_3$-\textsc{Subgraph Hitting} admits a polynomial kernel parameterized by a modulator to $\C_{\Delta}$ if and only if the graphs in $\C_{\Delta}$ have bounded $\mmbs_3$. With no extra assumption on $\C_{\Delta}$, this property is probably false due to the results of Hols, Kratsch, and Pieterse~\cite{HolsKP22}, but we conjecture that it is true if we ask $\C_{\Delta}$ to be hereditary and closed under disjoint union, even for hitting $K_t$ for every $t \geq 3$, replacing $\mmbs_3$ by $\mmbs_t$.
Toward an eventual proof of this conjecture, having unbounded minimal blocking sets seems to permit a generic reduction to obtain the lower bound, in the spirit of the one of \autoref{thm:hardness-lambda1} or any similar one in previous work~\cite{JansenK11,FominS16,BougeretJS22,DekkerJ22,CyganLPPS14}. Indeed, Hols, Kratsch, and Pieterse~\cite[Thm 1.1]{HolsKP22} show that for \textsc{Vertex Cover}, lower bounds on kernel sizes directly follow from lower bounds on~$\mmbs_2$.
However, the opposite direction seems way more challenging. In~\cite{BougeretJS22}, this fact was established for \textsc{Vertex Cover} and minor-closed target classes via the notion of bridge-depth by proving, in particular, that there is a {\sl single} minor-obstruction for having large maximum minimal blocking sets, namely the chains of triangles (cf. left part of \autoref{fig-chains}). Unfortunately, we cannot hope the same nice behavior for $K_3$-\textsc{Subgraph Hitting} and monotone or hereditary graph classes, as chains of triangles are still an obstruction in this setting, but there exist other incomparable ones, as depicted in \autoref{fig-chains}.

\begin{figure}[h!tb]
 \begin{center}
  \includegraphics[scale=.87]{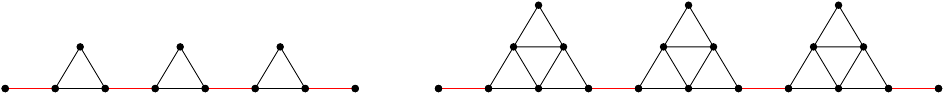}
\end{center}
\caption{Two chains of length three incomparable with respect to the (induced) subgraph relation. In both graphs, it can be verified that the set of four red thicker edges is a minimal blocking set.}
 \label{fig-chains}
 \end{figure}

Finally, in the ambitious quest for finding the appropriate measures that characterize the hereditary or monotone classes $\C_t$ for which $K_t$-\textsc{Subgraph Hitting} admits a polynomial kernel parameterized by the size of a modulator $X$ to $\C_t$, we hope that the techniques we developed to provide  a polynomial kernel for the case
$\bedt(G-X) \leq \lambda$ will play an important role. A natural attempt to generalize $\bedt$ to a more powerful measure is to relax, or even drop, the ``weak attachment'' condition on the sets to be removed in every round of the elimination process. This raises new challenges for obtaining a polynomial kernel that do not seem easy to overcome, at least with the existing tools in this area.





\bibliography{references}


%
%
%
%
%


\end{document}